\documentclass[11pt]{article}
\pdfoutput=1
\usepackage{jheppub}

\usepackage{graphicx}

\usepackage{amsthm,amsmath,amssymb,bbm}
\usepackage{dsfont}  
\usepackage{enumitem} 

\makeatletter
\def\@fpheader{\relax}
\makeatother

\preprint{IGC-17/9-1\\[-2.5em]}

\newlength{\bracewidth}
\newcommand{\myunderbrace}[2]{\settowidth{\bracewidth}{$#1$}#1\hspace*{-1\bracewidth}\smash{\underbrace{\makebox{\phantom{$#1$}}}_{#2}}}
\newcommand*\mystrut[1]{\vrule width0pt height0pt depth#1\relax}

\newcommand{\tra}{^\intercal} 
\newcommand{\ii}{\mathrm{i}} 

\newtheorem{theorem}{Theorem}
\newtheorem{proposition}{Proposition}
\newtheorem{definition}{Definition}
\newtheorem*{corollary}{Corollary}

\toccontinuoustrue

\title{\huge Linear growth of the entanglement entropy and the Kolmogorov-Sinai rate}

\author[a]{Eugenio Bianchi,}
\emailAdd{ebianchi@gravity.psu.edu}
\author[a]{Lucas Hackl,}
\emailAdd{lucas.hackl@psu.edu}
\author[a,b]{Nelson Yokomizo\,}
\emailAdd{yokomizo@fisica.ufmg.br}

\affiliation[a]{Institute for Gravitation and the Cosmos \& Physics Department,\\ Penn State, University Park, PA 16802, USA}
\affiliation[b]{Departamento de F\'isica - ICEx, Universidade Federal de Minas Gerais, \\
CP 702, 30161-970, Belo Horizonte, MG, Brazil}

\abstract{
The rate of entropy production in a classical dynamical system is characterized by the Kolmogorov-Sinai entropy rate $h_{\mathrm{KS}}$ given by the sum of all positive Lyapunov exponents of the system. We prove a quantum version of this result valid for bosonic systems with unstable quadratic Hamiltonian. The derivation takes into account the case of time-dependent Hamiltonians with Floquet instabilities. We show that the entanglement entropy $S_A$ of a Gaussian state grows linearly for large times in unstable systems, with a rate $\Lambda_A \leq h_{KS}$ determined by the Lyapunov exponents and the choice of the subsystem $A$. We apply our results to the analysis of entanglement production in unstable quadratic potentials and due to periodic quantum quenches in many-body quantum systems. Our results are relevant for quantum field theory, for which we present three applications: a scalar field in a symmetry-breaking potential, parametric resonance during post-inflationary reheating and cosmological perturbations during inflation. Finally, we conjecture that the same rate $\Lambda_A$ appears in the entanglement growth of chaotic quantum systems prepared in a semiclassical state.
}

\begin{document}
\maketitle
\flushbottom

\newpage

\section{Introduction}
Entanglement plays a central role in the thermalization of isolated quantum systems \cite{polkovnikov2011colloquium,gogolin2016equilibration,d2016quantum}. The paradigmatic setting consists in a Hamiltonian system prepared in a pure state and evolving unitarily, $|\psi_t\rangle=e^{-\mathrm{i}H t}|\psi_0\rangle$.  The objective is to study the thermalization of observables $\mathcal{O}_A$ belonging to a  subalgebra of observables $\mathcal{A}_A$ which define a bipartition $\mathcal{H}=\mathcal{H}_A\otimes\mathcal{H}_B$ of the system in a subsystem $A$ and its complement $B$. While the von Neumann entropy of the system vanishes at all times, the entropy of the subsystem $A$,
\begin{equation}
S_A(t)=-\mathrm{Tr}_A\big(\rho_A(t)\log\rho_A(t)\big)\qquad\textrm{with}\qquad \rho_A(t)=\mathrm{Tr}_B\big(|\psi_t\rangle\langle\psi_t|\big)\,,
\label{}
\end{equation}
in general does not vanish and has a non-trivial evolution. The origin of this entropy is the entanglement between the degrees of freedom in the subsystem $A$ and its complement. Equilibration in the subsystem $A$ occurs when the entanglement entropy $S_A(t)$ approaches an equilibrium value $S_\mathrm{eq}$, with thermalization corresponding to $S_\mathrm{eq}$ given by the thermal entropy.

A generic behavior has been observed for various systems prepared in a state with initially low entanglement entropy, $S_A(t_0)\ll S_\mathrm{eq}\;$: After a transient which depends on the details of the initial state of the system, the entropy of the subsystem goes through a phase of \emph{linear growth},
\begin{equation}
S_A(t)\sim \Lambda_A \,t\,,
\label{}
\end{equation}
until it saturates to an equilibrium value as described in figure~\ref{fig:A-Illustration}. This behavior is observed in the evolution of various isolated quantum systems, in particular in systems that show the signatures of quantum chaos \cite{Zurek:1994wd,miller1999signatures,pattanayak1999lyapunov,monteoliva2000decoherence,tanaka2002saturation}, in many-body quantum systems \cite{kim2013ballistic} and quantum fields \cite{Calabrese:2005in,Calabrese:2004eu,cotler2016entanglement} after a  quench, and in the thermalization of strongly-interacting quantum field theories studied using holographic methods \cite{balasubramanian2011thermalization,balasubramanian2011holographic,Hartman:2013qma,Liu:2013iza,Liu:2013qca}. Understanding the mechanism of this process is of direct relevance for the puzzle of fast thermalization of the quark gluon plasma produced in heavy-ion collisions \cite{muller2011entropy,kunihiro2010chaotic,Hashimoto:2016wme}, in models of black holes as fast scramblers of quantum information \cite{sekino2008fast}, and in the study of the quantum nature of space-time \cite{VanRaamsdonk:2010pw,Bianchi:2012ev,Bianchi:2015edy,Susskind:2014moa,Jefferson:2017sdb,Chapman:2017rqy}. In particular, being able to predict from first principles the rate of growth $\Lambda_A$ of the entanglement entropy in the phase of linear growth can provide us with crucial information on the time-scale of thermalization.

\begin{figure}[t]
	\centering
	\noindent\makebox[\linewidth]{
		\includegraphics[width=.9\linewidth]{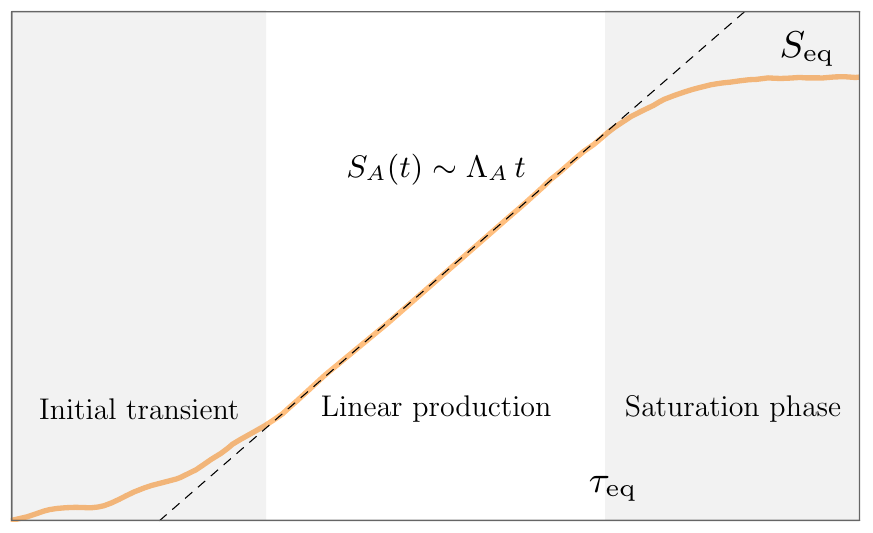}
	}
	\caption{\emph{Sketch of typical entanglement production}. Typical time dependence of the entanglement entropy $S_A(t)$ under unitary time evolution:  After an initial transient (a), linear production occurs with characteristic rate $\Lambda_A$ (b), and finally the system equilibrates in the saturation phase (c). The typical time scale for the equilibration of a state with initially vanishing entanglement entropy is $\tau_\mathrm{eq}\sim S_\mathrm{eq}/\Lambda_A$.}
	\label{fig:A-Illustration}
\end{figure}

On the other hand, at the classical level\,---\,in Hamiltonian chaotic systems\,---\,the coarse-grained entropy $S_\mathrm{cl}(t)$ shows a behavior similar to the one described in figure~\ref{fig:A-Illustration}, with a linear phase which has a known rate of growth $h_{\mathrm{KS}}$ \cite{latora1999kolmogorov,falcioni2005production},
\begin{equation}
S_\mathrm{cl}(t)\sim h_{\mathrm{KS}} \,t\,,
\label{}
\end{equation}
where $h_{\mathrm{KS}}$ is the Kolmogorov-Sinai rate of the system, an information-theoretic quantity that measures the uncertainty remaining on the future state of a system, once
an infinitely long past is known. The Kolmogorov-Sinai rate has dimension of \emph{time}${}^{-1}$ and for regular Hamiltonian systems is given by the sum of the positive Lyapunov exponents $\lambda_i$ of the system \cite{kolmogorov1958new,Sinai:2009,zaslavsky2008hamiltonian,vulpiani2010chaos}.

In quantum systems that have a classical chaotic counterpart, a relation between the rate of growth of the entanglement entropy $\Lambda_A$ and the classical Lyapunov exponents $\lambda_i$ is expected \cite{Zurek:1994wd,miller1999signatures,pattanayak1999lyapunov,monteoliva2000decoherence,tanaka2002saturation,kunihiro2009towards,Asplund:2015osa,Bianchi:2015fra}, despite the fact that Lyapunov exponents are global quantities which probe the phase space of the full system, not just of the subsystem $A$. 

In this paper we investigate the relation between $\Lambda_A$ and the Lyapunov exponents $\lambda_i$ by studying the evolution of Gaussian states in many-body systems and quantum field theories with quadratic time-dependent Hamiltonians. Non-trivial Lyapunov exponents arise in the presence of instabilities and of parametric resonances. In this context we prove that the linear growth of the entanglement entropy $S_A(t)$ has a classical counterpart: The entanglement rate $\Lambda_A$ equals the exponential rate of growth of the volume of a cell in the sub phase space of the subsystem $A$. We then provide an algorithm for computing $\Lambda_A$ in terms of the Lyapunov exponents $\lambda_i$ of the classical system and the choice of subsystem $A$. The methods developed apply both to quantum systems with finitely many degrees of freedom and to quantum fields in external time-dependent backgrounds when the subsystem is given by a finitely-generated Weyl subalgebra $\mathcal{A}_A$.\\

The paper is organized as follows. In section \ref{sec:growth}, we present our main result, theorem~\ref{th:SA}, which determines the asymptotic rate of growth of the entanglement entropy for Gaussian states in systems with quadratic Hamiltonians and establish its relation to the Kolmogorov-Sinai rate, including bounds for non-Gaussian initial states. Our results are then applied in section \ref{sec:applications-finite-dof} to the study of entanglement production in three example systems with finitely many degrees of freedom, including unstable potentials and periodic quantum quenches. Section~\ref{sec:qft} is dedicated to quantum field theories, where we consider again three example systems: a scalar field in a symmetry-breaking potential, parametric resonance during post-inflationary reheating and cosmological perturbations during inflation. We show that our results agree with numerical evaluations of the entanglement entropy for sufficiently large times. In sections~\ref{sec:classical} and \ref{sec:entropy} we present technical results required for the derivation of our main result in section \ref{sec:growth}. In particular, section~\ref{sec:entropy} reflects the structure of our proof for theorem~\ref{th:SA}. Finally, in section~\ref{sec:discussion} we discuss limitations and possible extensions of our work, and in particular a conjecture on entanglement production in chaotic systems. Moreover, we discuss the relation to linear growth of the entanglement entropy after a generic quantum quench. The paper is supplemented with appendices where we provide a summary of the relevant results in the study of dynamical systems and Lyapunov exponents and of general symplectic techniques for the study of the time-evolution and entanglement and R\'enyi entropies of Gaussian states.

\section{Results: Linear growth of the entanglement entropy}\label{sec:growth}
We state the main result which relates the asymptotic rate of growth of the entanglement entropy of a quantum system to classical instabilities encoded in the Lyapunov exponents of the classical system. Our proof is based on a set of technical results presented in sections \ref{sec:classical} and \ref{sec:entropy}.

\subsection{Entanglement entropy growth, instabilities and the volume exponent}\label{sec2:growth}
We consider a quadratic bosonic system with $N$ degrees of freedom. We denote linear observables by $\xi^a=(q_1,\dots,q_N,p_1,\dots,p_N)$ and assume canonical commutation relations $[q_i,q_j]=[p_i,p_j]=0$ and $[q_i,p_j]=\ii\delta_{ij}$. These relations can be more compactly phrased by stating $[\xi^a,\xi^b]=\ii\Omega^{ab}$ where $\Omega^{ab}$ is a symplectic form. The most general quadratic Hamiltonian is given by
\begin{align}
	H(t)=\frac{1}{2}h_{ab}(t)\xi^a\xi^b+f_a(t)\xi^a\,,
\end{align}
where we explicitly allow for dependence on time $t$. The time-evolution of an initial state $|\psi_0\rangle$ under the unitary dynamics $U(t)$ generated by $H(t)$  results in the evolution of the entanglement entropy of a subsystem
\begin{equation}
S_A(t)\equiv S_A\big(U(t)|\psi_0\rangle\big)\,.
\end{equation}
Before we state the main result, let us introduce two important notions:
\begin{itemize}
	\item \textbf{Subsystem exponents in classical dynamical systems}\\
	In classical dynamical systems, a quadratic Hamiltonian $H(t)$ generates a linear symplectic flow $M(t): V\to V$ on the classical phase space $V$ of the theory. The transpose $M(t)^\intercal$ of this flow acts on the dual phase space $V^*$. Given a linear observable $\ell\in V^*$, we can define the Lyapunov exponent of $\ell$ as the limit
	\begin{align}
		\lambda_{\ell}=\lim_{t\to\infty}\log\frac{1}{t}\frac{\lVert M(t)^\intercal\ell\rVert}{\lVert\ell\rVert}\,,
	\end{align}
	which is independent from the metric we choose to measure the length. A system decomposition $V=A\oplus B$ of the classical phase space into subsystem phase spaces $A$ and $B$ induces an equivalent decomposition $V^*=A^*\oplus B^*$ of the dual phase space. Here, we can generalize the notion of Lyapunov exponents to define the subsystem exponent $\Lambda_A$ defined by
	\begin{align}
		\Lambda_A=\lim_{t\to\infty}\frac{1}{t}\log\frac{\mathrm{vol}(M(t)^\intercal\mathcal{V}_A)}{\mathrm{vol}(\mathcal{V}_A)}\,,
		\label{eq:LambdaAintro}
	\end{align}
	where $\mathcal{V}_A\subset A^*$ is an arbitrary parallelepiped in the subspace $A^*$. The subsystem exponent captures the exponential volume growth of subsystem regions. The volume $\mathrm{vol}$ is measured on the subspace where $M(t)^\intercal\mathcal{V}_A$ lives, but the subsystem exponent is independent of the global metric on $V^*$ one chooses to define the volume form on arbitrary subspaces. We explain the relation between $\Lambda_A$ and $\lambda_\ell$ in section \ref{sec:classical}, while more technical details are summarized in appendix \ref{app}.
	\item \textbf{Entanglement of Gaussian states}\\
	It is well-known that a Gaussian bosonic state $|\psi\rangle$ can be completely characterized by its expectation value $\zeta^a=\langle \psi|\xi^a|\psi\rangle$ and its covariance matrix $G^{ab}=\langle \psi|\xi^a\xi^b+\xi^b\xi^a|\psi\rangle-2\,\zeta^a\zeta^b$. Recent progress on unifying methods for bosonic and fermionic Gaussian states \cite{Vidmar:2017uux} suggest an equivalent description where $G^{ab}$ is replaced by a linear complex structure $J^a{}_b=-G^{ac}\omega_{cb}$ with $\omega$ being the inverse of $\Omega$. Choosing a system decomposition $A\oplus B$ with complementary subsystems $A$ and $B$ allows us to compute the entanglement entropy $S_A(|\psi\rangle)$ between them. For a Gaussian state $|\psi\rangle$, this entanglement entropy can be directly computed from $J$, which we use in section \ref{sec:entropy} and review in appendix \ref{sec:quantum}.
\end{itemize}
With these preliminaries in hand, we can state the following theorem that applies to the evolution of the entanglement entropy of any Gaussian initial state.
\begin{theorem}[Entanglement growth] 
	\label{th:SA}
	Given a quadratic time-dependent Hamiltonian $H(t)$ and a subsystem $A$ with subsystem exponent $\Lambda_A$, the long-time behavior of the entanglement entropy of the subsystem is
	\begin{equation}
		S_A(t)\sim \Lambda_A\,t
		\label{}
	\end{equation}
	for all initial Gaussian states $|J_0,\zeta_0\rangle$.
\end{theorem}
\begin{proof}
	The proof of this theorem involves three steps that rely on ingredients reviewed in section \ref{sec:entropy}. 
	\begin{itemize}[noitemsep]
		\item[(i)] \textbf{The entanglement entropy is bounded by the Renyi entropy:}\\
		We define the asymptotic rate of growth of the entanglement entropy as its long-time linear scaling $\lim_{t\to\infty}\frac{1}{t}S_A(U(t)|J_0,\zeta_0\rangle)$. We note that quadratic time-dependent Hamiltonians evolve the initial Gaussian state into a Gaussian state, (\ref{eq:Jt}). In section~\ref{sec2:bounds} we prove that the entanglement entropy of a Gaussian state is bounded from below by the R\'{e}nyi entropy $R_A(U(t)|J_0,\zeta_0\rangle)$ and from above by the R\'{e}nyi entropy plus a state-independent constant, inequality (\ref{eq:SA-RA}). Therefore, we have the equality
		\begin{equation}
			\lim_{t\to\infty}\frac{S_A(U(t)|J_0,\zeta_0\rangle)}{t}\;=\lim_{t\to\infty}\frac{R_A(U(t)|J_0,\zeta_0\rangle)}{t}\,,
			\label{}
		\end{equation}
		i.e., the asymptotic rate of growth of the entanglement entropy and of the R\'{e}nyi entropy coincide.
		\item[(ii)] \textbf{The Renyi entropy is given by a phase space volume:}\\
		In section~\ref{sec2:Renyi-phase-space} we prove that the R\'{e}nyi entropy of a Gaussian state equals the logarithm of the phase space volume of a symplectic cube $\mathcal{V}_A$ spanning the subsystem $A$, (\ref{eq:RAvol}). The volume is measured with respect to the metric induced by the state, (\ref{eq:J=Og}). In the case of the time-dependent Gaussian state $U(t)|J_0,\zeta_0\rangle$, we can measure the volume with respect to the time-dependent induced metric $G_t=M(t) G_0 M\tra(t)$. Equivalently, we can consider the time-dependent symplectic cube $M\tra(t)\mathcal{V}_A$ and measure its volume with respect to the initial metric $G_0$ induced by the initial state,
		\begin{equation}
			R_A(U(t)|J_0,\zeta_0\rangle)=\log \mathrm{Vol}_{G_0}(M\tra(t)\mathcal{V}_A)\,.
			\label{}
		\end{equation}
		\item[(iii)] \textbf{The Renyi entropy grows as regions in phase space are stretched:}\\
		The subsystem exponent $\Lambda_A$ introduced in Eq.~(\ref{eq:LambdaAintro}) and discussed in section~\ref{sec2:LambdaA} provides a generalization of the notion of Lyapunov exponents of a classical Hamiltonian system. It involves the choice of a subsystem $A$, a symplectic dynamics $M(t)$ and a reference metric $G_0$,
		\begin{equation}
			\Lambda_A=\lim_{t\to\infty}\frac{1}{t}\log\frac{\mathrm{Vol}_{G_0}(M\tra(t)\mathcal{V}_A)}{\mathrm{Vol}_{G_0}(\mathcal{V}_A)}\,.
			\label{}
		\end{equation}
		Despite the metric $G_0$ is needed for the definition, the value of the subsystem exponent $\Lambda_A$ is independent of $G_0$ for regular Hamiltonian systems. The subsystem exponent can be expressed in terms of the Lyapunov exponents of the system using the algorithm described in theorem~\ref{th:LambdaA}, (\ref{eq:LambdaAsum}).
	\end{itemize}
	Using (i), (ii) and (iii), we find that the asymptotic rate of growth of the entanglement entropy is given by the subsystem exponent $\Lambda_A$,
	\begin{equation}
		\lim_{t\to\infty}\frac{S_A(U(t)|J_0,\zeta_0\rangle)}{t}=\Lambda_A
		\label{eq:thSA}
	\end{equation}
	for all initial Gaussian states, therefore proving the statement of the theorem.
\end{proof}
We note that, as the entanglement entropies of complementary subsystems  $A$ and $B$ coincide, $S_A(|\psi\rangle)=S_B(|\psi\rangle)$, also their asymptotic rates of growth have to coincide. Consistency with the statement of the theorem implies that, at the classical level, the subsystem exponents defined in section~\ref{sec2:LambdaA} for a symplectic decomposition $V=A\oplus B$ coincide 
\begin{equation}
	\Lambda_A=\Lambda_B\,.
	\label{}
\end{equation}
This statement can be proven using the expression (\ref{eq:LambdaAsum}) of the subsystem exponents or more directly using the property $\det[J]_A=\det[J]_B$ for the restriction of a complex structure $J$ to complementary symplectic subspaces.

\subsection{Entanglement and the Kolmogorov-Sinai entropy rate}
\begin{theorem}[Entanglement growth -- generic subsystem] 
	\label{th:SA-generic}
	Given a quadratic time-dependent Hamiltonian $H(t)$ with Lyapunov exponents $\lambda_i$, the long-time behavior of the entanglement entropy of a generic subsystem $A$ is
	\begin{equation}
		S_A(t)\sim \Big(\sum_{i=1}^{2N_A}\lambda_i\Big)\,t
		\label{eq:SA=LambdaA}
	\end{equation}
	for all initial Gaussian states $|J_0,\zeta_0\rangle$ and all generic subsystems with $N_A$ degrees of freedom. 
	
	In particular, the rate of growth of the entanglement entropy is bounded from above by the Kolmogorov-Sinai rate $h_\mathrm{KS}$,
	\begin{equation}
		\lim_{t\to\infty}\frac{1}{t}S_A(t)\;\leq h_\mathrm{KS}\,.
		\label{eq:KSbound}
	\end{equation}
	The decomposition in two complementary subsystems both with dimension larger than the number of instabilities results in an entanglement growth proportional to the Kolmogorov-Sinai rate,
	\begin{equation}
		S_A(t)\sim h_\mathrm{KS}\;t\qquad\mathrm{for}\qquad 2N_A\geq N_I\quad\mathrm{and}\quad 2N_B\geq N_I\,,
		\label{eq:SAsaturate}
	\end{equation}
	and therefore saturates the bound (\ref{eq:KSbound}).
\end{theorem}
\begin{proof}
	The asymptotic rate of growth of the entanglement entropy of a Gaussian state is given by the subsystem exponent $\Lambda_A$ as stated in theorem~\ref{th:SA}, (\ref{eq:thSA}). For a generic subsystem, theorem~\ref{th:LambdaA-generic} states that the subsystem exponent equals the sum of the $2N_A$ largest Lyapunov exponents, (\ref{eq:LambdaGeneric}). Together with Pesin's theorem (\ref{eq:Pesin}), this result implies that the asymptotic rate of growth is bounded from above by the Kolmogorov-Sinai rate of the system,
	\begin{equation}
		\lim_{t\to\infty}\,\frac{1}{t}S_A(t)\;=\;\sum_{i=1}^{2N_A}\lambda_i\;\leq\; h_\mathrm{KS}\,.
		\label{}
	\end{equation}
	Moreover, the subsystem exponent $\Lambda_A$ equals the Kolmogorov-Sinai rate $h_\mathrm{KS}$ when its dimension is in the range $N_I\leq 2 N_A\leq 2N-N_I$, (\ref{eq:saturate}). Recalling that $N_A+N_B=N$, this range coincides with the requirement that the dimension of each of the two complementary subsystems is larger than the number of instabilities, $2N_A\geq N_I$ and $2N_B\geq N_I$. In this case the bound (\ref{eq:KSbound}) is saturated. 
\end{proof}
We note that quantum many-body systems often have only a small finite number of unstable directions $N_I$ compared to the number of degrees of freedom of the system, $N_I\ll N$. A generic decomposition in two complementary subsystems that encompass the fractions $f_A=N_A/N$ and $f_B=1-f_A$ of the full system satisfies (\ref{eq:SAsaturate}) if the fractions are in the range 
\begin{equation}
	\frac{N_I}{2N}\,\leq\, f_A\,\leq 1-\frac{N_I}{2N}\,.
	\label{}
\end{equation}
As a result, in the limit $N\to \infty$ with $N_I$ finite, we have that the entanglement growth is proportional to the Kolmogorov-Sinai rate $S_A(t)\sim h_\mathrm{KS}\;t$ for all partitions of the system into two complementary subsystems each spanning a finite fraction of the system, except for a set of partitions of measure zero.

\subsection{Bounds on non-Gaussian initial states}
Computing the entanglement entropy growth of non-Gaussian states is a non-trivial problem as efficient tools similar to the ones discussed in (\ref{eq:BHY}) are not available. Nevertheless upper bounds that generalize theorems~\ref{th:SA} and~\ref{th:SA-generic} can be established in the case of evolution driven by a quadratic time-dependent Hamiltonian.

Let us consider an initial non-Gaussian state $|\psi_0\rangle$ and the unitary evolution $U(t)$ generated by a quadratic time-dependent Hamiltonian of the most general form described in (\ref{eq:quadHhat}). The symmetric part of the connected $2$-point function at the time $t$ is given by
\begin{equation}
	G^{ab}_t\equiv \langle\psi_t| \hat{\xi}^a \hat{\xi}^b+\hat{\xi}^b\hat{\xi}^a|\psi_t\rangle-2\langle\psi_t| \hat{\xi}^a|\psi_t\rangle\langle\psi_t| \hat{\xi}^b|\psi_t\rangle\;=\;M^a{}_c(t)M^b{}_d(t)\,G^{cd}_0
	\label{eq:Gt}
\end{equation}
where $|\psi_t\rangle\equiv U(t)|\psi_0\rangle$ and $M^a{}_b(t)$ is the symplectic matrix defined in (\ref{eq:Mh}). There always exists a mixed Gaussian state $\rho_0$ which has the same correlation function $G^{ab}_0$ at the time $t=0$ \cite{holevo2011probabilistic}. By construction, the $2$-point function of the unitarily evolved Gaussian state $U(t)\rho_0 U^{-1}(t)$ is the function $G^{ab}_t$ of (\ref{eq:Gt}). Moreover one can show that the entanglement entropy of the non-Gaussian state $|\psi_t\rangle$ is bounded from above by the entanglement entropy of the mixed Gaussian state having the same $2$-point function $G^{ab}_t$, i.e. $S_A(\rho_{NG})\leq S_A(\rho_{G})$  where $\rho_{NG}=\mathrm{Tr}_B(|\psi_t\rangle\langle\psi_t|)$ is the reduced density matrix of the non-Gaussian state, and $\rho_G=\mathrm{Tr}_B(U(t)\rho_0 U^{-1}(t))=\exp( -\frac{1}{2}k_{rs}(t)\hat{\xi}^r\hat{\xi}^s+E_0(t))$ is the reduced density matrix of the Gaussian state \cite{Bianchi:2015fra}. The proof is immediate: Recalling that the relative entropy is a positive function \cite{vedral2002role,ohya2004quantum}, we have
\begin{align}
	0\;\leq &\;\;S(\rho_{NG}\|\rho_{G})\equiv \mathrm{Tr}_A(\rho_{NG}\log \rho_{NG}\;-\;\rho_{NG}\log \rho_{G})\\[.5em]
	&=-S_A(\rho_{NG})+S_A(\rho_{G})+\frac{1}{2}k_{rs}(t)\underbrace{\Big(\mathrm{Tr}_A(\hat{\xi}^r\hat{\xi}^s \rho_{NG})-\mathrm{Tr}_A(\hat{\xi}^r\hat{\xi}^s \rho_{G})\Big)}_{=0}\;,
	\label{}
\end{align}
where $S(\rho_{NG}\|\rho_{G})$ is the relative entropy and the term in parenthesis vanishes as the two states have the same correlation function  by construction. On the other hand, theorem~\ref{th:SA} generalizes to mixed Gaussian states implying the asymptotic growth $S_A(\rho_{G})\sim \Lambda_A \,t$ for the entanglement entropy of a subsystem $A$. As a result we find the inequality
\begin{equation}
	\lim_{t\to \infty}\frac{1}{t}S_A(|\psi_t\rangle)\;\leq\; \Lambda_A
	\label{eq:upper-bound}
\end{equation}
which states that the asymptotic rate of growth of the entanglement entropy of a non-Gaussian state $|\psi_t\rangle$ which evolves unitarily with a quadratic Hamiltonian is bounded from above by the subsystem exponent $\Lambda_A$. This result generalizes theorems~\ref{th:SA} and~\ref{th:SA-generic} to non-Gaussian states and establishes the Kolmogorov-Sinai rate $h_{\mathrm{KS}}$ as the upper bound for the asymptotic rate of growth of the entanglement entropy. 

Preliminary numerical investigations indicate that, under the unitary evolution given by a quadratic time-dependent Hamiltonian, the upper bound $\Lambda_A$ in (\ref{eq:upper-bound}) might in fact be saturated by all initial states and not just by Gaussian states \cite{Hackl:2017ndi}.

\section{Applications: unstable potentials and periodic quenches} \label{sec:applications-finite-dof}
We briefly discuss three examples of simple systems that show a linearly growing entanglement entropy and allow us to test our results.

\subsection{Particle in a 2d inverted potential}\label{sec2:constrained-particle}
In our first example, we study a simple system consisting of just two degrees of freedom. Despite its simplicity, the example captures the main features of the theorems presented above. It also resembles the system studied in \cite{Asplund:2015osa} and thereby illustrates how our theorems simplify the involved steps to understand the asymptotic behavior of the entanglement entropy.

We consider a system which can be described as a quantum particle with mass $m=1$ moving in a plane with coordinates $(q_1,q_2)$ and corresponding momenta $(p_1,p_2)$. The instabilities arise from an inverted harmonic potential $V(q_1,q_2)=-\frac{\lambda_1^2}{2}q_1^2-\frac{\lambda_2^2}{2}q_2^2$ with $\lambda_1\geq\lambda_2>0$. The Hamiltonian of this system is explicitly given by
\begin{align}
	H=\frac{1}{2}\left(p_1^2+p_2^2-\lambda^2_1q_1^2-\lambda^2_2q_2^2\right)\,.
	\label{eq:2d-inverted}
\end{align}
If we choose the Darboux basis $\mathcal{D}_V=(p_1,p_2,q_1,q_2)$, the matrices $h$ and $K=\Omega h$ become
\begin{align}
	h=\left(\begin{array}{cc|cc}
		1 & & &\\
		& 1 & &\\
		\hline
		& & -\lambda_1 & \\
		& & & -\lambda_2
	\end{array}\right)\quad\Rightarrow\quad K=\left(\begin{array}{cc|cc}
		& & -\lambda_1^2 & \\
		& & & -\lambda_2^2\\
		\hline
		-1 & & &\\
		& -1 & &
	\end{array}\right)\,.
\end{align}
The Lyapunov exponents $(\lambda_1,\lambda_2,-\lambda_2,-\lambda_1)$ are given by the eigenvalues of $K$ and the Lyapunov basis $\mathcal{D}_L=(\ell^1,\ell^2,\ell^3,\ell^4)=(Q_1,Q_2,P_2,P_1)$ are the corresponding eigenvectors
\begin{align}
	\left\{
	\begin{array}{cl}
		Q_1&=q_1-\lambda_1\,p_1\\
		Q_2&=q_2-\lambda_2\,p_2\\
		P_2&=p_2+\frac{1}{\lambda_2}\,q_2\\
		P_1&=p_1+\frac{1}{\lambda_1}\,q_1
	\end{array}
	\right.
\end{align}

\begin{figure}[t]
	\centering
	\noindent\makebox[\linewidth]{
		\includegraphics[width=.9\linewidth]{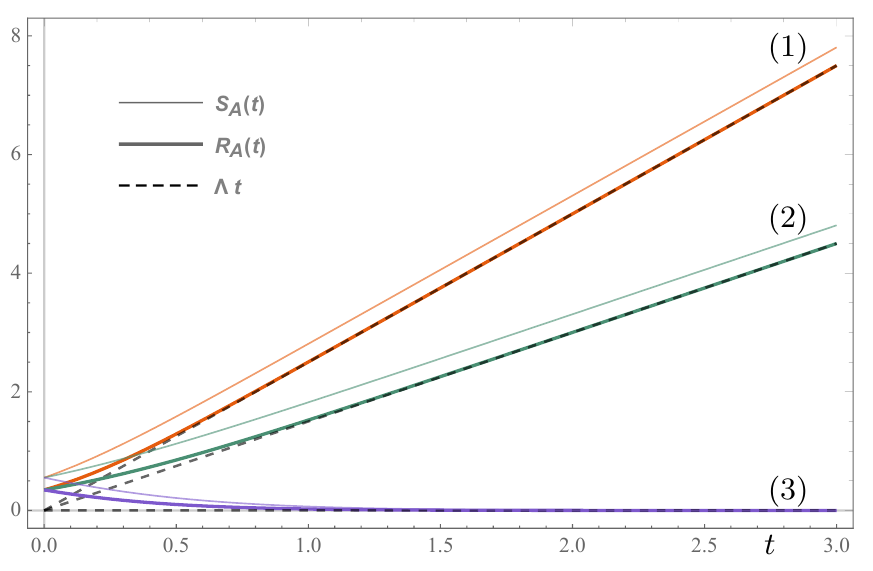}}
	\caption{\emph{Particle in a 2d inverted potential}. The plot shows the exact behavior of the R\'{e}nyi entropy $R_A(t)$ (thick) and entanglement entropy $S_A(t)$ (thin) in comparison to the predicted asymptotics $\Lambda t$ (dashed). The system is defined in (\ref{eq:2d-inverted}) with subsystems specified in (\ref{eq:2d-inverted-subsystem1}--\ref{eq:2d-inverted-subsystem3}). For the computation, we choose $\lambda_1=-\lambda_4=2$ and $\lambda_2=-\lambda_3=1/2$. The initial state is chosen to be $|J_0\rangle$ with associated metric $G_0(\ell^i,\ell^j)=\delta^{ij}$. In the case of examples (1) and (2), we have $S_A(t)-R_A(t)\to c=0.31$, while for example (3), we have $S_A(t)\to R_A(t)\to 0$ for large $t$.}
	\label{fig:Entanglement-Examples}
\end{figure}

With these definitions, let us consider the three different choices of subsystem $A$ discussed also in section~\ref{sec2:LambdaA-Lyapunov}:
\begin{align}
	\label{eq:2d-inverted-subsystem1}
	(1)\,\,&\left\{
	\begin{array}{cl}
		\phi&=\,Q_1\\[4pt]
		\pi&=\,Q_2+P_1
	\end{array}
	\right.&\Rightarrow\quad T=\left(\begin{array}{cccc}
		1 & 0 & 0 & 0\\
		0 & 1 & 0 & 1
	\end{array}\right)\quad\Rightarrow\quad
	\Lambda_A=\lambda_1+\lambda_2\geq 0\,,\\[1em]
	(2)\,\,&\left\{
	\begin{array}{cl}
		\phi&=\,Q_1+Q_2\\[4pt]
		\pi&=\,P_2
	\end{array}
	\right.&\Rightarrow\quad T=\left(\begin{array}{cccc}
		1 & 1 & 0 & 0\\
		0 & 0 & 1 & 0
	\end{array}\right)\quad\Rightarrow\quad
	\Lambda_A=\lambda_1-\lambda_2\geq 0\,,
	\label{}\\[1em]
	(3)\,\,&\left\{
	\begin{array}{cl}
		\phi&=\,Q_1+Q_2\\[4pt]
		\pi&=\,P_1
	\end{array}
	\right.&\Rightarrow\quad T=\left(\begin{array}{cccc}
		1 & 1 & 0 & 0\\
		0 & 0 & 0 & 1
	\end{array}\right)\quad\Rightarrow\quad
	\Lambda_A=\lambda_2-\lambda_2=0\,.
	\label{eq:2d-inverted-subsystem3}
\end{align}
We can study the entanglement entropy for these subsystems numerically where we start with the (entangled) initial state given by $G_0(\ell^i,\ell^j)=\delta^{ij}$. Figure~\ref{fig:Entanglement-Examples} shows excellent agreement with our predictions. In particular, we also see that the entanglement entropy $S_A(t)$ and the R\'{e}nyi entropy $R_A(t)$ only differ by the constant $1-\log(2)$ if the system is strongly entangled.

\subsection{Quadratic potential with instabilities}
The second example consists in the evolution in a time-independent potential with instabilities. Let us consider a classical system with $N$ degrees of freedom which we parametrize by $N$ conjugate pairs $(q_i,p_i)$ of coordinates in phase space. The Hamiltonian $H$ of the system consists of a standard kinetic term and a quadratic potential,
\begin{align}
	H=\sum^{N}_{i=1}\frac{1}{2}p_i^2+\sum^{N}_{i,j=1}\frac{1}{2}\,V_{ij}\,q_i\,q_j\,.
	\label{eq:potential}
\end{align}
The potential is determined by the symmetric matrix $V_{ij}$ with eigenvalues $v_i$. This classical system has  $2N$ Lyapunov exponents $\lambda_i$ determined by the eigenvalues of  $V_{ij}$ and given by $\lambda_i=\pm \mathrm{Im}(\sqrt{v_i})$.  Positive eigenvalues correspond to stable directions of the potential, lead to oscillatory motion and vanishing Lyapunov exponent. In the presence of negative eigenvalues $v_i<0$, the system is unstable, the classical motion is unbounded and nearby trajectories in phase space diverge at an exponential rate given by the $\lambda_i=+ \mathrm{Im}(\sqrt{v_i})$. Now we consider the associated quantum system prepared in a Gaussian state and study the behavior of a generic subsystem. If the potential $V_{ij}$ couples the subsystem $A$ with the rest of the system, we expect that the entanglement entropy of the subsystem changes in time. Our theorem states that the entanglement entropy of a generic subsystem $A$ with $N_A$ degrees of freedom asymptotically grows at a rate given by the sum of the $2N_A$ largest Lyapunov exponents, (\ref{eq:SA=LambdaA}). We give a concrete example: We consider a system with $N=20$ degrees of freedom and quadratic potential specified by a $N\times N$ real symmetric random matrix $V_{ij}$. Negative eigenvalues of $V_{ij}$ correspond to unstable directions of the potential and non-vanishing Lyapunov exponents with the following values:
\begin{equation}
	\begin{array}{cccccccccccccccccccc}
		\lambda_1 & \lambda_2 &\lambda_3 &\lambda_4 &\lambda_5 &\lambda_6 &\lambda_7 &\lambda_8 &\lambda_9 &\lambda_{10} & \lambda_{11} & \lambda_{12} &\lambda_{13}&\lambda_{14} &\lambda_{15} &\lambda_{16} &\lambda_{17} &\lambda_{18} &\lambda_{19} &\lambda_{20}\\
		.55 & .45 & .34 & .31 & .29 & 0 & 0 & 0 & 0 & 0 &0 &0 &0 &0 &0 &{-.29} &{-.31} &{-.34} &{-.45} &{-.55}
	\end{array}
	\label{eq:random}
\end{equation}
figure~\ref{fig:Example} shows the growth of the entanglement entropy of an initially un-entangled Gaussian state for generic subsystems of different dimensions. For a one-dimensional subsystem $N_A=1$, theorem~\ref{th:SA-generic} predicts the asymptotic growth $S_A(t)\sim (\lambda_1+\lambda_2)\, t$. Note that, as the Lyapunov exponents appear in couples $(\lambda,-\lambda)$, the asymptotic growth of the rest of the system has the same rate, $S_B(t)\sim (\lambda_1+\ldots+\lambda_{18})\, t \;=\;(\lambda_1+\lambda_2)\, t$, as expected for the entanglement entropy of a pure state. 

Clearly, the statement that the entanglement growth is linear in time applies only to generic subsystems and not to subsystems that are aligned to the shape of the potential. As an example, let us call $(Q_1,\ldots,Q_N,\,P_N,\ldots,P_1)$ the Lyapunov basis of the system and consider the subsystem spanned by the canonical couple $(Q_1,P_1)$. Theorem~\ref{th:SA} predicts a sublinear rate $S_A(t)\sim (\lambda_1+\lambda_{20})\, t=(\lambda_1-\lambda_1)\, t\,\sim\,o(t)$, which is consistent with the statement that the entropy will stay constant as the subsystem is isolated. Moreover, note that the difference of Lyapunov exponents can also appear in the asymptotic rate of the entanglement growth, for instance the subsystem spanned by $(Q_2,P_2+P_3)$ has asymptotic rate $S_A(t)\sim (\lambda_2-\lambda_3)\, t$. This is the quantum version of the example discussed in (\ref{eq:l2-l3}). It is important to remark however that these subsystems are non-generic and form a subset of measure zero as discussed in the proof of theorem~\ref{th:LambdaA-generic}.

\begin{figure}[t]
	\centering
	\noindent\makebox[\linewidth]{
		\includegraphics[width=.9\linewidth]{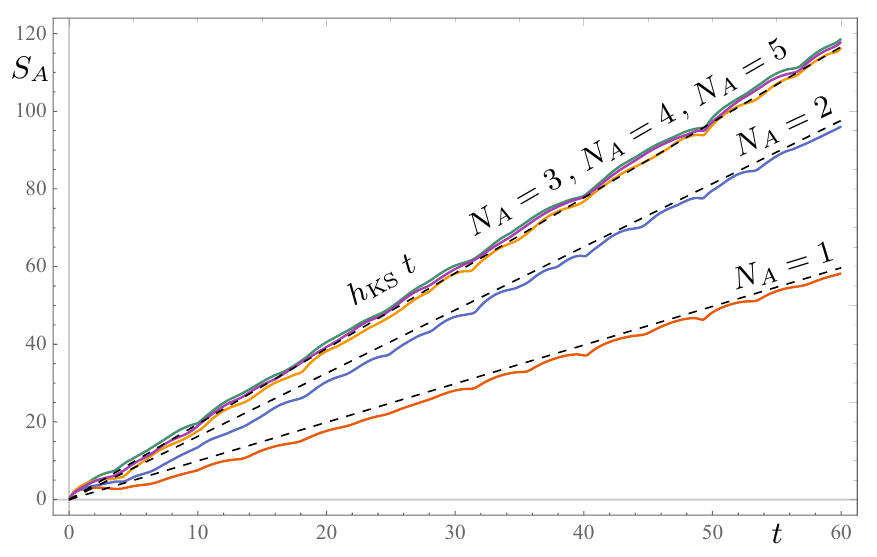}
	}
	\caption{\emph{Quadratic potential with instabilities}. This plot compares the behavior of the entanglement entropy with the asymptotic prediction of our theorem. The system consists of $N=10$ degrees of freedom. The time evolution is determined by a Hamiltonian with random quadratic potential $V$. The $2N=20$ Lyapunov exponents are given by $(.55, .45, .34, .31, .29, 0, \ldots, 0,{-.29},{-.31},{-.34},{-.45},{-.55})$. We plot five subsystems with $1\leq N_A\leq 5$. The entanglement entropy (colored lines) agrees with the predicted asymptotics (dashed lines). Note that for $N_A\geq 3$, the asymptotic behavior is the same due to the stable Lyapunov exponents $\lambda_i=0$ for $6\leq i\leq 15$. These are exactly the cases for which we have $\Lambda_A=h_{\mathrm{KS}}$, which means that entropy production rate coincides with the classical Kolmogorov-Sinai entropy rate.}
	\label{fig:Example}
\end{figure}

The random quadratic potential $V_{ij}$ with Lyapunov exponents specified in (\ref{eq:random}) has $N_I=5$ unstable directions and classical Kolmogorov-Sinai rate $h_\mathrm{KS}=\lambda_1+\ldots+\lambda_5\simeq 1.94$. Theorem~\ref{th:SA-generic} states that at the quantum level the long-time behavior of the entanglement entropy is linear with rate $h_\mathrm{KS}$ for all generic subsystem decompositions such that $2N_A\geq N_I$ and $2N_B\geq N_I$, i.e., $S_A(t)\sim h_\mathrm{KS}\,t$ for generic subsystems of dimension $3\leq N_A\leq 18$.  Figure~\ref{fig:Example} shows the numerical evolution of the entanglement entropy and the cases $N_A=3$, $N_A=4$, $N_A=5$ exhibit a linear growth with rate given by the Kolmogorov-Sinai rate as predicted.

A remarkable feature of the predictions stated in theorems~\ref{th:SA} and~\ref{th:SA-generic} is that the asymptotic growth of the entanglement entropy is determined by the subsystem and completely independent from the choice of initial Gaussian state. This feature is a consequence at the quantum level of the fact that the Lyapunov exponents of a classical system are independent of the choice of metric used to measure the distance between trajectories. In the language of complex structures $J_0$ that specify the initial Gaussian state $|J_0,\zeta_0\rangle$, the Lyapunov exponents are given by the eigenvalues of the matrix $L=\lim_{t\to\infty}\frac{1}{2t}\log\big( M^{-1}(t)\,J_0\,M(t)\big)$ defined in (\ref{eq:Lmatrix}) and are independent of $J_0$.
Figure~\ref{fig:Example} shows only initial states with vanishing entanglement entropy, but the two theorems~\ref{th:SA} and~\ref{th:SA-generic} apply to all Gaussian states, even to ones that have large initial entanglement entropy. Clearly the theorem applies only to the asymptotic behavior of the entanglement entropy. In fact we could take as initial state the time-reversal of the Gaussian state used in  figure~\ref{fig:Example} at late times. In this case the entanglement entropy would initially decrease, reach a minimum and eventually start growing linearly as predicted by the theorems on the asymptotic growth.

\subsection{Periodic quantum quenches in a harmonic lattice}\label{sec:periodic-quenches}
As a third example of a system that displays a linear growth of the entanglement entropy, we discuss the case of a harmonic lattice subject to periodic quantum quenches. The Hamiltonian of the system is 
\begin{equation}
	H(t)=\sum^{N}_{i=1}\frac{1}{2}\Big(p_i^2+\,\Omega^2(t)\;q_i^2 +\,\kappa \,(q_{i+1}-q_i)^2\Big)\,,
	\label{eq:Hperiodic}
\end{equation}
which describes the dynamics of a one-dimensional chain of $N$ bosons with nearest-neighbor coupling $\kappa$ and boundary conditions $q_{N+1}=q_1$, $p_{N+1}=p_1$. The one-particle oscillation frequency $\Omega(t)$ is periodically switched between the values $\Omega_0\pm\varepsilon$ with period $2\,T_0$,
\begin{align}
	&\Omega(t)\;=\;\left\{
	\begin{array}{cl}
		\Omega_0-\varepsilon&\qquad \textrm{for}\quad \;\;0\leq t < T_0\\[4pt]
		\Omega_0+\varepsilon&\qquad \textrm{for}\quad T_0\leq t < 2\,T_0
	\end{array}
	\right.\\[4pt]
	&\Omega(t+2\,T_0)=\Omega(t)\qquad \textrm{and}\qquad \varepsilon\ll \Omega_0\,.
	\label{}
\end{align}
The system is prepared in the ground state of the instantaneous Hamiltonian $H(0)$ at the time $t=0$ and then let evolve unitarily. The state of the system at stroboscopic times $t_n=2\,n\,T_0$ which are multiples if the period $2\,T_0$ can be obtained by computing the Floquet Hamiltonian of the system, $H_F$:
\begin{equation}
	U(2\,n\,T_0)=\Big(U(2\,T_0)\Big)^n=e^{-\mathrm{i}\, 2\,n\,T_0\, H_F}\quad \textrm{with}\quad H_F\equiv \frac{\mathrm{i}}{2\,T_0}\log\big(e^{-\mathrm{i}\,H_2 T_0}\,e^{-\mathrm{i}\,H_1 T_0}\big)\,,
	\label{}
\end{equation}
where $H_1=H(T_0)$ and $H_2=H(2\,T_0)$. We note that the Floquet Hamiltonian $H_F$ is quadratic and time-independent, but is not of the standard form consisting of the sum of a kinetic and a potential term as in the case of (\ref{eq:potential}). This general form is taken into account in theorems~\ref{th:SA} and~\ref{th:SA-generic} which apply to all quadratic Hamiltonians.

\begin{figure}[t]
	\centering
	\noindent\makebox[\linewidth]{
		\includegraphics[width=.9\linewidth]{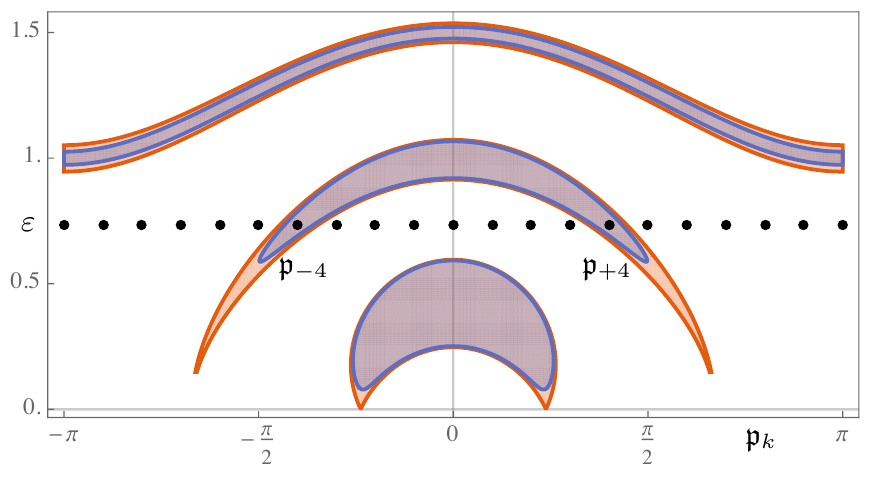}
	}
	\caption{\emph{Instability bands in a periodic quantum quench}. We sketch the instability region of Floquet exponents $\mu$ with positive real part (red: $\mathrm{Re}(\mu)>0.03$, blue: $\mathrm{Re}(\mu)>0.05$) for the system described in (\ref{eq:Hperiodic-k}) as a function of $\varepsilon$ and $\mathfrak{p}_k$. We chose the values $T_0\simeq \pi$, $\Omega_0\simeq 0.3$ and $\kappa\simeq 0.3$. Furthermore, we indicate the discrete momenta $\mathfrak{p}_k$ for a periodic chain with $N\simeq20$ and $\varepsilon\simeq 0.735$, such that there are two modes with unstable Floquet exponents, namely $\mathfrak{p}_{\pm4}=\pm \frac{2\pi}{5}$.}
	\label{fig:Floquet}
\end{figure}

The classical system described by (\ref{eq:Hperiodic}) shows instabilities when small perturbations from the equilibrium configuration are amplified via the mechanism of parametric resonance. Floquet theory \cite{floquet1883equations,chicone1999or} provides the tools for the description of the dynamics driven by an Hamiltonian which is periodic in time, as is the case for (\ref{eq:Hperiodic}). The eigenvalues of the symplectic evolution matrix (\ref{eq:Mh}) evaluated at a period $M(2\,T_0)$ come in quadruplets $(e^{+2\,T_0\,\mu},e^{-2\,T_0\,\mu},e^{+2\,T_0\,\mu^*},e^{-2\,T_0\,\mu^*})$ where the complex numbers $\mu$ are the Floquet exponents of the system. The stability of the system is measured by the real part of the Floquet exponents which coincide with the Lyapunov exponents, $\lambda=\mathrm{Re}(\mu)$.

\begin{figure}[t]
	\centering
	\noindent\makebox[\linewidth]{
		\includegraphics[width=.9\linewidth]{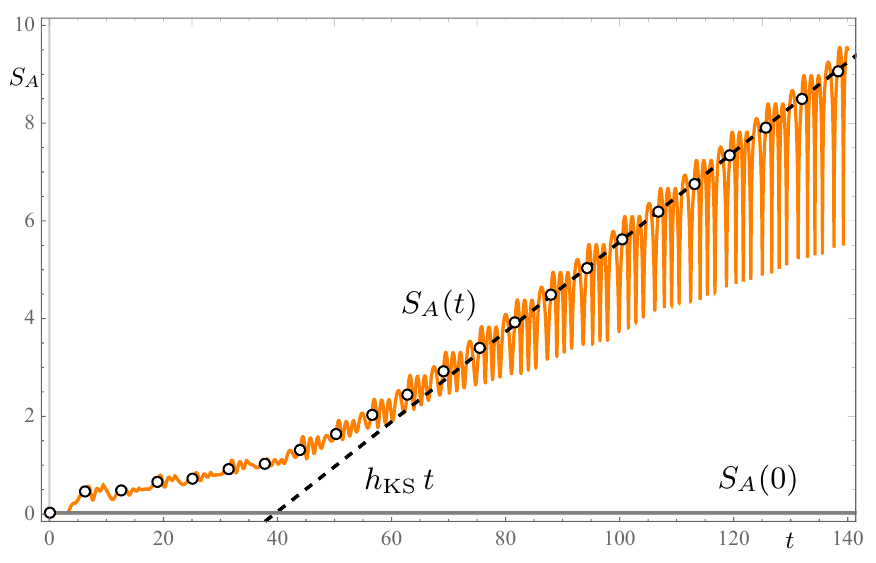}
	}
	\caption{\emph{Periodic quantum quenches in a harmonic lattice}. We show the entanglement entropy $S_A(t)$ as a function of time for the subsystem spanned by $(q_1,p_1)$. The stroboscopic entanglement entropy $S_A(n2T_0)$ is indicated by white dots. The asymptotic prediction of the Kolmogorov-Sinai production rate $S_A(t)\sim h_{\mathrm{KS}}\,t$ with $h_{\mathrm{KS}}=0.092$ is shown as a dashed line where we adjusted the offset for easy comparison of the slope. Note that the entanglement entropy stays constant in the interval $[0,T_0]$ as expected from the fact that the system is prepared in the ground state of the initial Hamiltonian $H(0)$.}
	\label{fig:Floquet-production}
\end{figure}

The Lyapunov exponents of the system (\ref{eq:Hperiodic}) can be easily determined. In Fourier transformed variables\footnote{The Fourier transformed canonical variables are defined as $Q_k=\frac{1}{\sqrt{N}}\sum_l q_l\, e^{\mathrm{i}\frac{2\pi k}{N}l}$, $P_k=\frac{1}{\sqrt{N}}\sum_l p_l \,e^{\mathrm{i}\frac{2\pi k}{N}l}$, so that $[Q_k,P_{-k'}]=\mathrm{i}\,\delta_{k,k'}$} $Q_k$, $P_k$ with $k=0,\pm1,\pm2,\ldots,\pm (N-1)/2$, the Hamiltonian takes the form
\begin{equation}
	H(t)=\sum_k\frac{1}{2}\Big(|P_k|^2+\,\omega_k^2(t)\;|Q_k|^2\Big)\,,
	\label{eq:Hperiodic-k}
\end{equation}
with 
\begin{equation}
	\omega_k(t)\equiv\,\sqrt{\Omega^2(t)+4\kappa\sin^2(\mathfrak{p}_k/2)}\,\qquad\textrm{and}\qquad \mathfrak{p}_k\equiv\frac{2\pi k}{N}\,.
	\label{eq:modek}
\end{equation}
In particular, the speed of sound of the mode of momentum $\mathfrak{p}_k$ switches periodically between the values $v_k(T_0)$ and $v_k(2\,T_0)$, with $v_k(t)\equiv\partial\omega_k(t)/\partial \mathfrak{p}_k$. As Fourier modes with different $|k|$ are decoupled, we can analyze the stability of the system mode by mode. The classical evolution of the coupled modes $(Q_{k},Q_{-k},P_{-k},P_{k})$ is given by
\begin{equation}
	\left(
	\begin{array}{l}
		Q_{k}(t)\\[4pt]
		P_{k}(t)
	\end{array}
	\right)\;=
	M_k(t)\;
	\left(
	\begin{array}{l}
		Q_{k}(0)\\[4pt]
		P_{k}(0)
	\end{array}
	\right)
	\label{}
\end{equation}
with
\begin{equation}
	M_k(2\,T_0)\;=\;
	\left(
	\begin{array}{cc}
		\cos (\omega_2 T_0) \;& \;-\omega_2 \sin (\omega_2 T_0) \\[4pt]
		\frac{1}{\omega_2} \sin (\omega_2 T_0)\;& \;\cos (\omega_2 T_0)
	\end{array}
	\right)
	\left(
	\begin{array}{cc}
		\cos (\omega_1 T_0) \;& \;-\omega_1 \sin (\omega_1 T_0) \\[4pt]
		\frac{1}{\omega_1} \sin (\omega_1 T_0)\;& \;\cos (\omega_1 T_0)
	\end{array}
	\right)
	\label{}
\end{equation}
given by a symplectic block of the symplectic matrix $M(2\,T_0)$ defined in (\ref{eq:Mh}), and $\omega_1=\omega(T_0)$ and $\omega_2=\omega(2\,T_0)$. The Lyapunov exponents of the system are the real parts of the Floquet exponents, i.e.
\begin{equation}
	\pm\lambda_k=\mathrm{Re}\left(\frac{1}{2T_0}\log \mathrm{Eig\big[M_k(2\,T_0)\big]}\right)\,.
	\label{}
\end{equation}
Analytic expressions of $\lambda_k$ can be found assuming that the periodic perturbation is small $\varepsilon\ll\Omega_0$ and the mode is at or near a parametric resonance. Defining $\delta \omega\equiv \omega_2-\omega_1$, $\,\omega_0\equiv\frac{\omega_1+\omega_2}{2}$, with $\delta\omega\ll \omega_0$, we find that the system is in parametric resonance when the average frequency $\omega_0$ of the mode is an half-integer multiple of the frequency of the perturbation, i.e.,
\begin{equation}
	\omega_0=\frac{n\,\pi}{2T_0}\,.
	\label{}
\end{equation}
At the parametric resonance, the positive Lyapunov exponents of the system are given by
\begin{equation}
	\lambda_k=
	\left\{
	\begin{array}{ll}
		+\frac{\delta\omega}{n\,\pi}&\qquad \textrm{if}\;\;n\;\;\textrm{odd},\\[1em]
		+\frac{\,T_0\,(\delta\omega)^2}{4\,n\,\pi}&\qquad \textrm{if}\;\;n\;\;\textrm{even}.
	\end{array}
	\right.
	\label{}
\end{equation}
For a finite perturbation, the stability of the system can be determined numerically. Figure~\ref{fig:Floquet} shows which modes $\mathfrak{p}_k$ are unstable for a given finite value of the perturbation parameter $\varepsilon$. In the example we have $N=20$, $\Omega_0\simeq 0.3$, $\kappa\simeq 0.3$ and $T_0\simeq \pi$. For $\varepsilon\simeq 0.735$ we have two unstable modes with $k=\pm 4$ and Lyapunov exponents $\lambda_{+4}\simeq \pm 0.046$, $\lambda_{-4}\simeq \pm 0.046$. The Kolmogorov-Sinai rate of the system is $h_{\mathrm{KS}}=0.092$. Figure~\ref{fig:Floquet-production} shows the growth of the entanglement entropy of a subsystem and the relation to $h_{\mathrm{KS}}$.

\section{Applications: quantum field theory}\label{sec:qft}
In section~\ref{sec:growth} we presented our main results for a bosonic quantum system with $N$ degrees of freedom. These results can be extended with minor modifications to the case of a bosonic quantum field. In particular, the formulation of theorems~\ref{th:SA} and~\ref{th:SA-generic} in terms of complex structures $J$ is motivated by and tailored to applications to quantum field theory in curved spacetimes \cite{ashtekar1980geometrical,ashtekar1975quantum,wald1994quantum}.

\subsection{Definition of a subsystem and the algebraic approach}
The presence of infinitely many degrees of freedom in quantum field theory has two immediate consequences which are relevant for our analysis \cite{haag2012local}:
\begin{itemize}
	\item[i)] the existence of unitarily inequivalent representations of the algebra of observables,
	\item[ii)] the lack of a factorization of the Hilbert space into a tensor product over local factors.
\end{itemize}
The algebraic approach to quantum field theory\,---\,together with the language of complex structures\,---\,provides a natural setting for discussing both aspects and formulating the analysis of the growth of the entanglement entropy of a subsystem in quantum field theory.

At the classical level, the phase space $V$ of a free scalar field has coordinates $\xi^a=(\varphi(\vec{x}),\pi(\vec{x}))$ with $\vec{x}$ a point on a Cauchy slice $\Sigma$. We adopt abstract indices and use the symbol $\omega_{ab}$ for the symplectic form on the infinite-dimensional vector space $V$. Carrying out the rigorous construction of the infinite-dimensional phase space requires the choice of a positive definite metric $g_{ab}$ compatible with the symplectic form $\omega_{ab}$, such that $V$ arises as the completion with respect to this metric. Contracting $\omega_{ab}$ with the inverse metric $G^{ab}$ gives rise to the complex structure $J^a{}_b=-G^{ac}\omega_{cb}: V\to V$. Given a reference complex structure $J_0$ and a symplectic transformation $M$, we can define a transformed complex structure $J_M=M^{-1}J_0 M$. The transformation $M$ is said to belong to the \emph{restricted} symplectic group if the commutator $A=[J_0,J_M]$ is a Hilbert-Schmidt operator, i.e. $\mathrm{tr}( A^\dagger A)<+\infty$ \cite{shale1962linear,shale1964states,ottesen2008infinite}.

At the quantum level, the choice of a complex structure $J_0$ defines a Gaussian state $|J_0\rangle$ which can be used as vacuum for building a Fock representation of the algebra of observables \cite{ashtekar1980geometrical,ashtekar1975quantum,wald1994quantum}. Representations built over Fock vacua $|J_0\rangle$ and $|J_M\rangle$ are unitarily equivalent if and only if the symplectic transformation $M$ belongs to the restricted symplectic group described above. When interpreted in terms of particle excitations, the state $|J_M\rangle$ describes a superposition of particle pairs over the vacuum $|J_0\rangle$. A  symplectic transformation $M$ which does not belong to the restricted group corresponds to a Bogoliubov transformation that produces an infinite number of particles \cite{berezin:1966}. 

Gaussian states and quadratic time-dependent Hamiltonians appear in the description of particle production in the early universe \cite{Birrell:1982ix,parker2009quantum}, Hawking radiation in black hole evaporation \cite{Hawking:1974sw}, in the Schwinger effect \cite{schwinger1951gauge,Greiner:1985ce}, in the dynamical Casimir effect \cite{casimir1948attraction,moore1970quantum} and more generally in all cases where the free quantum field evolves in a time-dependent background. It is known that, for some time-dependent backgrounds, the time-evolution\,---\,which, at the classical level, is encoded in a symplectic transformation $M$\,---\,cannot be implemented as a unitary operator in a Fock space at the quantum level \cite{Torre:1998eq}. Nevertheless the correlation functions in the quantum theory are still well-defined in terms of a complex structure $J_0$ and a symplectic transformation $M$ as described in (\ref{eq:2point}) \cite{Agullo:2015qqa}. The algebraic approach  focuses on correlation functions and does not involve the construction of a Fock space. It provides sufficient structure for defining the (abstract) state of the system and computing the evolution of the entanglement entropy of a subsystem, despite the potential lack of a standard unitary implementation of the time evolution in a Fock space, (i).\\

The second aspect which needs some clarification regards the definition of a subsystem in quantum field theory, (ii). It is a well-known fact about the ground state of a quantum field that the entanglement entropy of a region of space is divergent and\,---\,when an ultraviolet cutoff is introduced\,---\,it scales as the area of the boundary of the region \cite{sorkin1983entropy,srednicki1993entropy,eisert2010colloquium}. The divergence of the geometric entanglement entropy has an algebraic origin: The local subalgebra of observables associated to a region in space is of type III, i.e. it does not identify a factorization of the Fock space in a tensor product of Hilbert spaces \cite{haag2012local,Hollands:2017dov}. Three standard strategies to address this issue are: (a) a modification the ultraviolet behavior of the theory, for instance introducing a lattice cut-off \cite{sorkin1983entropy,bombelli1986quantum,srednicki1993entropy}, or (b) computing the mutual information between a region and a carved version of its complement so to introduce a ``safety corridor'' \cite{Casini:2008wt,Bianchi:2014bma,Hollands:2017dov}, or (c) focusing on the excess entropy of a state with respect to the one of the ground state \cite{holzhey1994geometric,Bianchi:2014bma}. Here we illustrate a different strategy: We focus on the entanglement entropy of a subsystem with a finite number $N_A$ of degrees of freedom. The geometric entanglement entropy which captures infinitely many degrees of freedom can be recovered in the limit of increasingly large subsystems \cite{Bianchi2017entropy}.\\

A simple example of a subsystem with a single degree of freedom, $N_A=1$, is provided by a linear smearing of the fields against given test functions $f(\vec{x})$ and $g(\vec{x})$:
\begin{equation}
	\hat{\varphi}_f=\int f(\vec{x})\, \hat{\varphi}(\vec{x})\,d^3\vec{x}\,,\quad \hat{\pi}_g=\int g(\vec{x})\,\hat{\pi}(\vec{x})\, d^3\vec{x}\,.
	\label{}
\end{equation}
The observables $\hat{\varphi}_f$ and $\hat{\pi}_f$ generate a Weyl algebra $\mathcal{A}_A$ of type I which, as in section~\ref{sec:entropy}, induces a factorization of the Hilbert space into $\mathcal{H}=\mathcal{H}_A\otimes \mathcal{H}_B$. The symplectic structure $\Omega_A$ of the subsystem can be computed from the commutator,
\begin{equation}
	[\hat{\varphi}_f,\hat{\pi}_g]=\mathrm{i}\int f(\vec{x})\,g(\vec{x})\,d^3\vec{x}\,.
	\label{}
\end{equation} 
Given a Gaussian state $|J\rangle$ of the quantum field, the symmetrized correlation function restricted to the subsystem $A$ is
\begin{equation}
	[G]_A=
	\left(
	\begin{array}{cc}
		2\;\langle J|\hat{\varphi}_f \,\hat{\varphi}_f|J\rangle &
		\langle J|\hat{\varphi}_f \,\hat{\pi}_g+\hat{\pi}_g \,\hat{\varphi}_f|J\rangle\\[.5em]
		\langle J|\hat{\varphi}_f \,\hat{\pi}_g+\hat{\pi}_g \,\hat{\varphi}_f|J\rangle & 
		2\; \langle J|\hat{\pi}_g \,\hat{\pi}_g|J\rangle
	\end{array}
	\right)\,.
	\label{}
\end{equation}
The restricted complex structure $[J]_A$ is given by ${[J]_A}^a{}_b=-[G]^{ac}_A \,(\Omega^{-1}_A)_{cb}$ and its eigenvalues $\pm\nu$ determine the entanglement entropy of the subsystem $A$ through (\ref{eq:sentropy}).

We illustrate our result on three paradigmatic cases in quantum field theory where the time-dependence of the entanglement entropy of a subsystem can be computed and our results on the linear growth can be tested.

\subsection{Dynamics of symmetry breaking and the inverted quadratic potential}
We consider a scalar field $\varphi(x)$ which goes through a symmetry breaking transition in real time \cite{weinberg1995quantum,Strocchi:2015uaa}. A simple model is described by the action\footnote{We adopt the notation $x=(t,\vec{x})$ for a spacetime point and use the signature $(-+++)$.}
\begin{equation}
	S[\varphi]=\int\textstyle \Big(-\frac{1}{2}\partial_\mu \varphi\,\partial^\mu\varphi-V(\varphi)\Big)\,d^4x
	\label{}
\end{equation}
with a quartic potential,
\begin{equation}
	\textstyle V(\varphi)=\frac{1}{2}\alpha(t)\,\varphi^2+\frac{1}{4!}\,\varepsilon\, \varphi^4 \,.
	\label{eq:Vquartic}
\end{equation}
The quadratic coupling $\alpha(t)$ is chosen so that, for $t>0$, a minimum of the potential breaks the symmetry $\varphi\to -\varphi$. We set
\begin{equation}
	\alpha(t)=
	\left\{
	\begin{array}{ll}
		+m^2, \qquad&t\leq 0\\[.5em]
		-\mu^2,& t>0
	\end{array}
	\right.
	\qquad \textrm{and}\qquad 0<\varepsilon \ll 1\,.
	\label{eq:Valpha}
\end{equation}
The system is initially prepared in the ground state at $t<0$ and then let evolve. For small quartic coupling $\varepsilon$ and short time, the evolution is described perturbatively by a tachyonic instability: At the onset of the symmetry-breaking transition,  the scalar field evolves as if it was free and had a negative mass-squared, $-\mu^2$. We focus on this initial phase.\\

We set $\varepsilon=0$ and study the free evolution governed by a quadratic Hamiltonian which transitions from a stable phase to an unstable phase. It is useful to adopt Fourier transformed canonical variables
\begin{equation}
	\varphi(\vec{k})=\int d^3\vec{x}\, \varphi(\vec{x})\,e^{\mathrm{i}\vec{k}\cdot\vec{x}}\quad,\qquad 
	\pi(\vec{k})=\int d^3\vec{x}\, \pi(\vec{x})\,e^{\mathrm{i}\vec{k}\cdot\vec{x}}
	\label{}
\end{equation}
so that the canonical commutation relations read $[\varphi(\vec{k}),\pi(\vec{k}')]=i (2\pi)^3\delta^3(\vec{k}+\vec{k}')$ and the symplectic structure in these coordinates is
\begin{equation}
	\Omega(\vec{k},\vec{k}')=
	\left(
	\begin{array}{cc}
		0 &
		+1\\
		-1& 
		0
	\end{array}
	\right)\;(2\pi)^3\,\delta^{3}(\vec{k}+\vec{k}')\,.
	\label{}
\end{equation}
For $t<0$, the Hamiltonian is 
\begin{equation}
	H=\int\frac{d^3\vec{k}}{(2\pi)^3}\,\frac{1}{2}\Big(|\pi(\vec{k})|^2+(\vec{k}^2+m^2)\,|\varphi(\vec{k})|^2\Big)\,
	\label{eq:Hstab}
\end{equation}
and the system is stable. The ground state is the Gaussian state $|J_0\rangle$ with correlation functions
\begin{align}
	G_0(\vec{k},\vec{k}')&=
	\left(
	\begin{array}{cc}
		2\;\langle J_0|\varphi(\vec{k})\varphi(\vec{k}')|J_0\rangle &
		\langle J_0|\varphi(\vec{k})\pi(\vec{k}')+\pi(\vec{k}')\varphi(\vec{k})|J_0\rangle\\[.5em]
		\langle J_0|\varphi(\vec{k})\pi(\vec{k}')+\pi(\vec{k}')\varphi(\vec{k})|J_0\rangle & 
		2\; \langle J_0|\pi(\vec{k})\pi(\vec{k}')|J_0\rangle
	\end{array}
	\right)\\[1em]
	&=
	\left(
	\begin{array}{cc}
		\frac{1}{\sqrt{\vec{k}^2+m^2}} &
		0\\[.5em]
		0& 
		\sqrt{\vec{k}^2+m^2}
	\end{array}
	\right)\;(2\pi)^3\,\delta^{3}(\vec{k}+\vec{k}')\,.
	\label{}
\end{align}
The complex structure of the ground state is therefore $J_0=-G_0\,\Omega^{-1}$, i.e.,
\begin{equation}
	J_0(\vec{k},\vec{k}')=
	\left(
	\begin{array}{cc}
		0 & \frac{1}{\sqrt{\vec{k}^2+m^2}}\\[.5em]
		-\sqrt{\vec{k}^2+m^2}& 0
	\end{array}
	\right)\;(2\pi)^3\,\delta^{3}(\vec{k}-\vec{k}')\,.
	\label{}
\end{equation} 

For $t>0$, that is, after the transition from the stable to the unstable phase, the Hamiltonian governing the free evolution is given by
\begin{equation}
	H=\int\frac{d^3\vec{k}}{(2\pi)^3}\,\frac{1}{2}\Big(|\pi(\vec{k})|^2+(\vec{k}^2-\mu^2)\,|\varphi(\vec{k})|^2\Big)\,.
	\label{eq:Hinst}
\end{equation}
Modes with $\vec{k}^2$ smaller than $\mu^2$ are unstable and have Lyapunov exponents which come in pairs $\pm\lambda(\vec{k})$, with
\begin{equation}
	\lambda(\vec{k})=\sqrt{\mu^2-\vec{k}^2}\qquad\textrm{for}\qquad 0\leq\vec{k}^2<\mu^2\,.
	\label{}
\end{equation}
As a result, infrared modes are unstable and the largest Lyapunov exponent $\lambda(0)=\mu$ is associated to the homogeneous mode. 

The classical evolution generated by the unstable Hamiltonian (\ref{eq:Hinst}) is given by the symplectic transformation $M_t(\vec{k},\vec{k}')=M_t(\vec{k})\,(2\pi)^3\,\delta^{3}(\vec{k}+\vec{k}')$, with
\begin{equation}
	M_t(\vec{k})=
	\left(
	\begin{array}{cc}
		\cosh\Big(\sqrt{\mu^2-\vec{k}^2}\;t\Big) & \sqrt{\mu^2-\vec{k}^2}\;\sinh\Big(\sqrt{\mu^2-\vec{k}^2}\;t\Big)\quad \\[1em]
		\frac{1}{\sqrt{\mu^2-\vec{k}^2}}\;\sinh\Big(\sqrt{\mu^2-\vec{k}^2}\;t\Big)& 
		\cosh\Big(\sqrt{\mu^2-\vec{k}^2}\;t\Big)\quad
	\end{array}
	\right)\,.
	\label{}
\end{equation}
As a result, in the quantum theory, the correlation functions at the time $t$ for the system initially prepared in the Gaussian state $|J_0\rangle$ are given by the evolved complex structure $J_t(\vec{k},\vec{k}')$,
\begin{equation}
	J_t(\vec{k},\vec{k}')=
	M_t(\vec{k}){}^{-1}\,\left(
	\begin{array}{cc}
		0 & \frac{1}{\sqrt{\vec{k}^2+m^2}}\\[.5em]
		-\sqrt{\vec{k}^2+m^2}& 0
	\end{array}
	\right)\;
	M_t(\vec{k})\;
	\;(2\pi)^3\,\delta^{3}(\vec{k}-\vec{k}')\,,
	\label{}
\end{equation}
which defines a Gaussian state $|J_t\rangle$ at the time $t$.\\

Let us consider a measuring device which probes the field and its momentum only in a neighborhood of the point $\vec{x}=0$ with a linear size $R$. This device can be modeled by a Gaussian smearing function $f(\vec{x})$. The subsystem $A$ defined by such measurements is encoded in the subalgebra of observables $\mathcal{A}_A$  generated by
\begin{equation}
	\hat{\varphi}_A=\int \hat{\varphi}(\vec{x})\, f(\vec{x})\,d^3\vec{x}\,,\quad \hat{\pi}_A=\int \hat{\pi}(\vec{x})\, f(\vec{x})\,d^3\vec{x}\qquad\textrm{with}\qquad 
	\textstyle
	f(\vec{x})=\frac{1}{(\sqrt{2\pi}\, R)^3}e^{-\frac{|\vec{x}|^2}{2R^2}}\,.
	\label{eq:subalgebra-phiA}
\end{equation}
This subsystem has $N_A=1$ bosonic degrees of freedom. Fluctuations of the observables $\hat{\varphi}_A$ and $\hat{\pi}_A$ at the time $t$ are encoded in the correlation functions of the subsystem
\begin{equation}
	[G_t]_A=
	\left(
	\begin{array}{cc}
		2\;\langle J_t|\hat{\varphi}_A \,\hat{\varphi}_A|J_t\rangle &
		\langle J_t|\hat{\varphi}_A \,\hat{\pi}_A+\hat{\pi}_A \,\hat{\varphi}_A|J_t\rangle\\[.5em]
		\langle J_t|\hat{\varphi}_A \,\hat{\pi}_A+\hat{\pi}_A \,\hat{\varphi}_A|J_t\rangle & 
		2\; \langle J_t|\hat{\pi}_A \,\hat{\pi}_A|J_t\rangle
	\end{array}
	\right)
	=\int G_t(\vec{k},\vec{k}')f(\vec{k})f(\vec{k}')\frac{d^3\vec{k}}{(2\pi)^3}\frac{d^3\vec{k}'}{(2\pi)^3}\,,
	\label{}
\end{equation}
where $f(\vec{k})=e^{-\frac{1}{2}R^2|\vec{k}|^2}$ is the Fourier transform of $f(\vec{x})$. Similarly, the restricted complex structure is the $2\times 2$ matrix
\begin{equation}
	[J_t]_A
	=\int J_t(\vec{k},\vec{k}')f(\vec{k})f(\vec{k}')\frac{d^3\vec{k}}{(2\pi)^3}\frac{d^3\vec{k}'}{(2\pi)^3}\,.
	\label{}
\end{equation}
The eigenvalues of $\mathrm{i}\,[J_t]_A$ come in pairs $\pm\nu(t)$ and the entanglement entropy of the subsystem $A$ is given by
\begin{equation}
	S_A(t)=S(\nu(t))
	\label{}
\end{equation}
where $S(\nu)$ is the function (\ref{eq:sentropy}). The predicted asymptotic rate of growth of the entanglement entropy of a subsystem with $N_A=1$ is given by the subsystem exponent $\Lambda_A=2\mu$, which is the sum of the two largest Lyapunov exponents, i.e.,
\begin{equation}
	S_A(t)\sim\,2\mu \,t\,.
	\label{}
\end{equation}
A numerical plot of the entanglement entropy as a function of time, together with the predicted rate of growth, is shown in figure~\ref{fig:G-Inverted-mass}.\\

\begin{figure}[t]
	\centering
	\noindent\makebox[\linewidth]{
		\includegraphics[width=.9\linewidth]{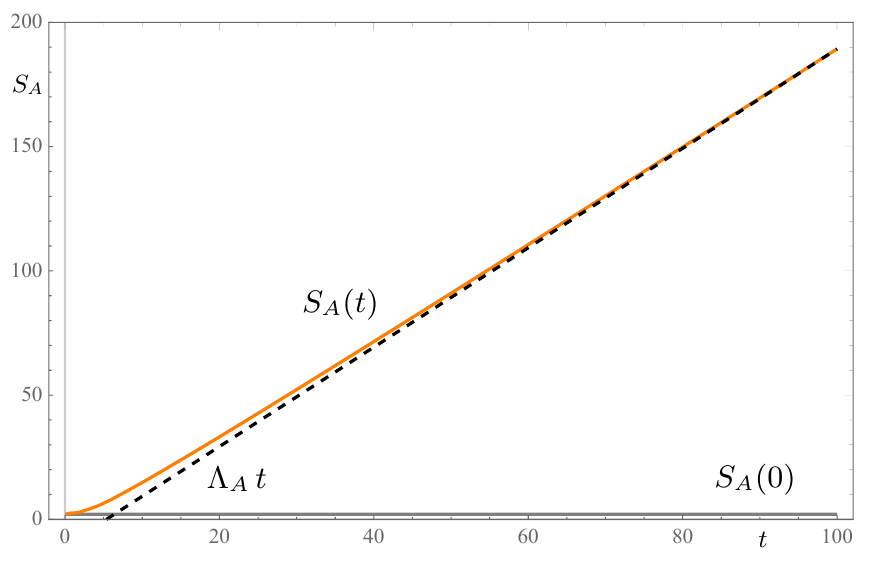}
	}
	\caption{\emph{Symmetry breaking and the inverted quadratic potential}. We compute the entanglement entropy $S_A(t)$ numerically for the time evolution with the unstable Hamiltonian (\ref{eq:Hinst}). We set $\mu=1$. The subsystem $A$ is defined in (\ref{eq:subalgebra-phiA}) with $R=1$.}
	\label{fig:G-Inverted-mass}
\end{figure}

We note that, as the positive Lyapunov exponents of the system appear in a continuous band $\lambda(\vec{k})=\sqrt{\mu^2-\vec{k}^2}$, the prediction for the asymptotic growth of the entanglement entropy of a subsystem with $N_A$ degrees of freedom is simply $S_A(t)\sim\,2N_A\,\mu \,t$. In the case of a subsystem with infinitely many degrees of freedom, it is useful to induce an infrared cutoff, for instance a cubic volume $V=L^3$. The boundary conditions induce a quantization of the momentum $\vec{k}=\big(\frac{2\pi}{L}n_x,\frac{2\pi}{L}n_y,\frac{2\pi}{L}n_z\big)$ which splits the degeneracy of the Lyapunov exponents and results in a discrete sequence $\lambda(\vec{k})$. We can now define the number $N_I$ of unstable degrees of freedom of the system. In the limit $L\gg \frac{2\pi}{\mu}$ we find
\begin{equation}
	N_I\sim L^3\int \Theta(\lambda(\vec{k}))\frac{d^3\vec{k}}{(2\pi)^3}\,=\frac{(\mu L)^3}{6\pi^2}\,.
	\label{}
\end{equation}
A generic subsystem which probes infinitely many degrees of freedom, as in the case of the geometric entanglement entropy of a region of space, would probe all the unstable degrees of freedom of the system. As a result the asymptotic growth of the entanglement entropy is expected to be given by
\begin{equation}
	S_A(t)\sim \mathfrak{h}_{\mathrm{KS}}\;L^3\,t\,,
	\label{}
\end{equation}
where $\mathfrak{h}_{\mathrm{KS}}$ is the Kolmogorov-Sinai rate per unit volume,
\begin{equation}
	\mathfrak{h}_{\mathrm{KS}}=\int \Theta(\lambda(\vec{k}))\,\lambda(\vec{k})\,\frac{d^3\vec{k}}{(2\pi)^3}\;=\;\frac{\mu^4}{32\pi}\,.
	\label{}
\end{equation}
As a result, for a generic subsystem with infinitely many degrees of freedom, the quantity $\mathfrak{h}_{\mathrm{KS}}$ describes the asymptotic behavior of the entanglement entropy per unit space-time volume.\\

In this analysis we assumed that the quartic term $\frac{1}{4!}\,\varepsilon\, \varphi^4$ is not present in the potential and the evolution is simply described by a quadratic Hamiltonian with instabilities. In section~\ref{sec:interaction} we discuss when this approximation is expected to be valid.

\subsection{Preheating and parametric resonance}
The simplest model of parametric resonance in quantum field theory is described by the Hamiltonian
\begin{equation}
	H(t)=\int\frac{d^3\vec{k}}{(2\pi)^3}\,\frac{1}{2}\Big(|\pi(\vec{k})|^2+\big(\vec{k}^2+v_0^2\, \sin^2(M_0 \,t)\big)\,|\varphi(\vec{k})|^2\Big)
	\label{eq:Hresonance}
\end{equation}
which has a quadratic potential that oscillates in time with period $2\pi/M_0$. For small values of the amplitude of oscillation,
\begin{equation}
	v_0^2\ll M_0^2\,,
	\label{}
\end{equation}
we have a narrow resonance band around the frequency of the perturbation,
\begin{equation}
	|\vec{k}|\in\left[M_0-\frac{v_0^2}{4M_0},\,M_0+\frac{v_0^2}{4M_0}\right]\,.
	\label{eq:narrow-band}
\end{equation}
Modes with momentum $|\vec{k}|\approx M_0$ are parametrically amplified. The Lyapunov exponents of the system can be determined via Floquet analysis as we already did in section~\ref{sec:periodic-quenches}. The canonical subsystem spanned by $\big(\varphi(\vec{k}),\pi(-\vec{k})\big)$ with $\vec{k}$ in the band (\ref{eq:narrow-band}) has Lyapunov exponents $\pm\lambda(\vec{k})$ with 
\begin{equation}
	\lambda(\vec{k})\approx \sqrt{\Big(\frac{v_0^2}{4M_0}\Big)^2-\big(M_0-|\vec{k}|\big)^2}\,.
	\label{}
\end{equation}
Given an initial Gaussian state\,---\,for instance the Minkowski vacuum\,---\,the evolution of the correlation functions of the system can be computed analytically in terms of Mathieu functions. In figure~\ref{fig:I-reheating-production} we show the time-evolution of the entanglement entropy of a subsystem $A$ defined by a subalgebra of observables $\mathcal{A}_A$ generated by the linear observables (\ref{eq:subalgebra-phiA}). At the classical level, the subsystem exponent $\Lambda_A$ can be easily computed: It is given by the sum of the two largest Lyapunov exponents of the systems, which are degenerate in value and correspond to modes exactly at the resonance $|\vec{k}|=M_0$. Therefore we have
\begin{equation}
	\Lambda_A= 2\,\lambda(M_0)=\frac{v_0^2}{2M_0}\,.
	\label{}
\end{equation}
As predicted by theorem~\ref{th:SA}, the entanglement entropy of the subsystem initially prepared in a Gaussian state $|J_0\rangle$ is observed to grow as $S_A(t)\sim \Lambda_A\,t$. \\

\begin{figure}[t]
	\centering
	\noindent\makebox[\linewidth]{
		\includegraphics[width=.9\linewidth]{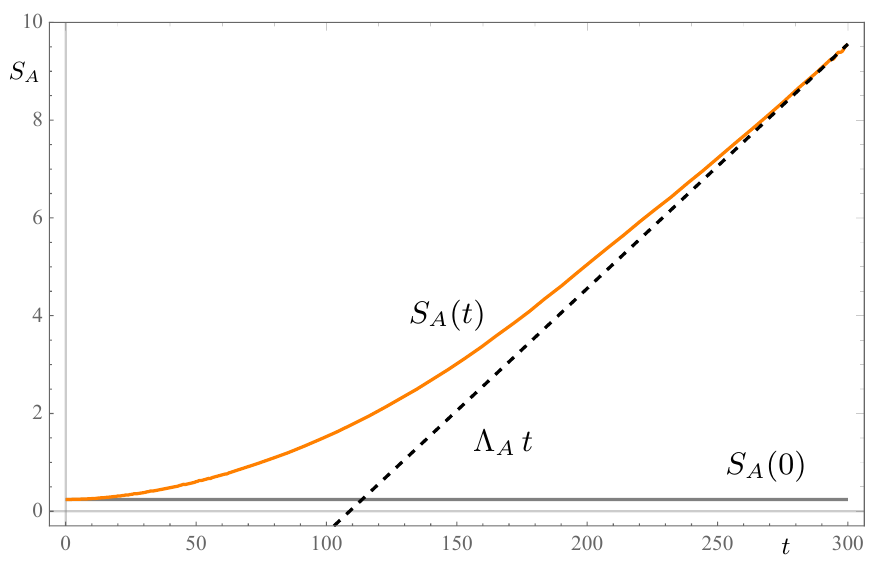}
	}
	\caption{\emph{Preheating and parametric resonance}. We compute the entanglement entropy $S_A(t)$ numerically. The subsystem is defined in (\ref{eq:subalgebra-phiA}), the time evolution is governed by the Hamiltonian presented in (\ref{eq:Hresonance}) and we start in the Minkowski vacuum state. We set $R=1$, $M_0=1$ and $v_0^2=0.1$ leading to the entanglement production rate $\Lambda_A=0.05$. The quantity $S_A(0)$ is the entanglement entropy of the Minkowski vacuum and the linear production phase is reached after an initial transient.}
	\label{fig:I-reheating-production}
\end{figure}

The phenomenon of parametric resonance plays a central role in a variety of far-from-equilibrium processes in quantum field theory \cite{calzetta1988nonequilibrium,Berges:2015kfa,Berges:2002cz}. We briefly discuss three examples: preheating in cosmology \cite{traschen1990particle,kofman1994reheating,allahverdi2010reheating,amin2015nonperturbative},
the formation of the chiral condensate in relativistic heavy-ion collisions \cite{mrowczynski1995reheating}, and the dynamical Casimir effect in trapped Bose-Einstein condensates \cite{busch2014quantum}.\\

At the end of cosmological inflation, the inflaton oscillates coherently around the minimum of its potential. Such oscillations excite the vacuum of matter fields via the phenomenon of parametric resonance. This phase of explosive non-thermal particle production is called preheating and is followed by a thermalization phase which provides the initial conditions for Big Bang Nucleosynthesis. A simple model of preheating consists in a coupling $V=\frac{1}{2}g^2 \Phi^2\,\varphi^2$ between the inflaton $\Phi$ and a field $\varphi$ which serves as proxy for Standard Model fields. Coherent oscillations of the inflaton, $\langle\Phi(\vec{x},t)\rangle=\Phi_0\,\sin(M_0\,t)$, result in an effective dynamics for the matter field described by an Hamiltonian of the form (\ref{eq:Hresonance}) with a coupling constant $v_0^2=g^2\Phi_0^2$. At the beginning of the oscillatory phase, the state of matter can be assumed to be the vacuum $|J_0\rangle$ because of the dilution effect of the inflationary phase. Its evolution in the preheating phase results in a Gaussian state $|J_t\rangle$ which is far from equilibrium: The Floquet instability of the Hamiltonian results in an explosive production of particles with momenta in the resonance band. For a given subalgebra of observables $\mathcal{A}_A$, such as the one discussed in (\ref{eq:subalgebra-phiA}), theorem~\ref{th:SA} predicts a linear growth of the  entropy with a rate given by the subsystem exponents $\Lambda_A$. Phenomenologically, the relevant choice of subalgebra of observables or coarse graining of the system is dictated by the interaction with its environment. The interaction of the produced particles and the expansion of the universe have the effect of reducing the efficiency of the resonance and eventually lead to a thermal-equilibrium radiation-dominated phase, with an expected entropy profile qualitatively similar to the one illustrated in figure~\ref{fig:Example}.\\

A similar preheating phenomenon is discussed in the context of relativistic heavy-ion collisions where, in the late stages of the evolution of a quark-gluon plasma, a chirally symmetric state rolls down and oscillates around to the minimum of the effective chiral potential \cite{mrowczynski1995reheating}. Coherent pion excitations are described by a quark condensate $\Phi=\langle \bar{q}q\rangle$ and $\vec{\phi}=\langle \bar{q}\vec{\sigma}q\rangle$ and a $O(4)$ linear sigma model with $\phi_a=(\Phi,\vec{\phi}\,)$ with explicit symmetry breaking. The action of the system is
\begin{equation}
	S[\Phi,\vec{\phi}\,]=\int d^4x\Big( -\frac{1}{2}\partial_\mu\phi_a\partial_\mu\phi^a\,-\frac{1}{4}g\,(\phi_a\phi^a-f_\pi^2)^2-m_\pi^2\,f_\pi \Phi_0\Big)
	\label{}
\end{equation}
with parameters $g\approx 20$, $f_\pi\approx 90$ MeV and $m_\pi\approx 140$ MeV. The coherent field $\phi_a$ is initially in a chirally symmetric state. As it rolls down the potential and oscillates around the minimum $\phi_a\approx (f_\pi,\vec{0}\,)$, a squeezed state of coherent pion pair excitations is produced via parametric resonance. In this phase, the entanglement entropy of a generic subsystem $A$ is predicted to grow with a rate given by the subsystem exponent $\Lambda_A$. As discussed in \cite{muller2011entropy}, the linear entropy growth is expected to be bounded by the Kolmogorov-Sinai rate of the system.\\

Time-dependent Hamiltonians of the form (\ref{eq:Hresonance}) appear also in the description of stimulated quasi-particle production in cold atomic Bose gases \cite{fedichev2004cosmological,carusotto2010density}. A periodic modulation of the external potential that traps the gas induces a response in the condensed portion of the gas which acts as a time-dependent background for quasi-particles. The study of entanglement entropy growth in cold atomic Bose gases is of particular relevance because of current experiments which can probe the non-separability of phonon pair creation \cite{jaskula2012acoustic,Steinhauer:2015saa}.

\subsection{Cosmological perturbations and slow-roll inflation}
During slow-roll inflation, quantum perturbations of the metric and the inflaton field are stretched and squeezed. We illustrate this phenomenon\,---\,together with the associated growth of the entanglement entropy\,---\,using a simple model consisting of a minimally-coupled massless scalar field in a cosmological spacetime. The action of the system is
\begin{equation}
	S[\varphi]=-\int d^4 x \textstyle \; \frac{1}{2}\sqrt{-g}g^{\mu\nu}\,\partial_\mu\varphi\,\partial_\nu\varphi
	\label{}
\end{equation}
where with a metric  $g_{\mu\nu}$ that defines the line element of a Friedmann-Lema\^itre-Robertson-Walker spacetime, $ds^2=g_{\mu\nu}dx^\mu dx^\nu=-dt^2+a(t)^2\,d\vec{x}^2$. The evolution of the field in the cosmic time $t$ is generated by the time-dependent Hamiltonian
\begin{equation}
	H(t)=\int\frac{d^3\vec{k}}{(2\pi)^3}\,\frac{1}{2}\Bigg(\frac{|\pi(\vec{k})|^2}{a(t)^3}+a(t) \,\vec{k}^2\,|\varphi(\vec{k})|^2\Bigg)\,,
	\label{eq:H-FLRW}
\end{equation}
where $\vec{k}$ is the comoving momentum and $\pi(\vec{k})=a(t)^3 \,d\varphi(\vec{k})/dt$. During slow-roll inflation, the Hubble rate changes slowly in time. To illustrate the analysis of the stability of the system, here we model this quasi-de Sitter phase with a de Sitter scale factor,
\begin{equation}
	a(t)=e^{H_0\,t}\,.
	\label{eq:scale-factor}
\end{equation}
The canonical subsystem spanned by $\big(\varphi(\vec{k}),\pi(-\vec{k})\big)$ with comoving momentum $\vec{k}$ is not a regular Hamiltonian system because of exponential collinearity (see appendix~\ref{app:regular}). In fact the angle between the two Lyapunov vectors $\ell_1(\vec{k})$ and $\ell_2(\vec{k})$  approaches $0$ as $e^{-H_0t}$. As a result, the Lyapunov exponents of the mode $\vec{k}$ do not have to be opposite in sign. In fact they are found to be given by
\begin{equation}
	\lambda_1(\vec{k})=H_0\quad\textrm{and}\quad \lambda_2(\vec{k})=0\,,
	\label{}
\end{equation}
where $H_0$ is the Hubble rate.
\begin{figure}[t]
	\centering
	\noindent\makebox[\linewidth]{
		\includegraphics[width=.9\linewidth]{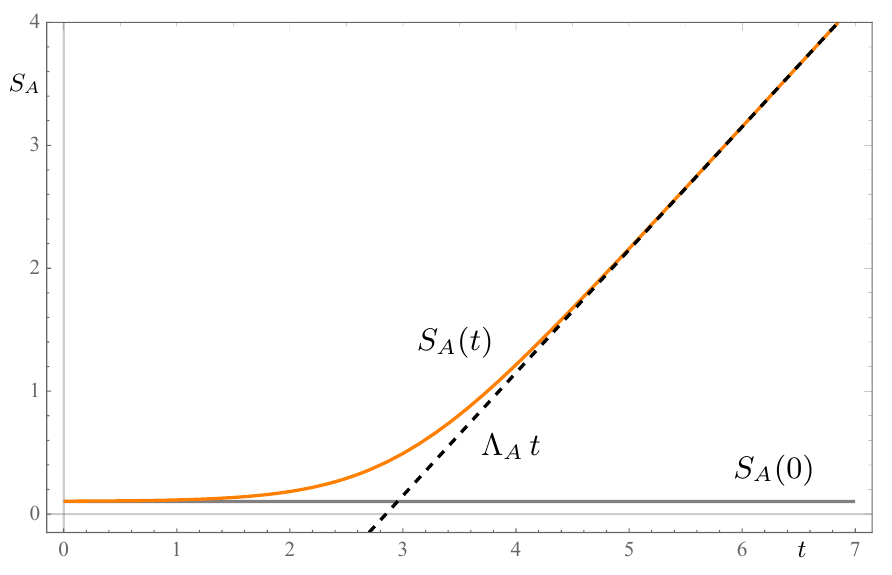}
	}
	\caption{\emph{Quantum field in de Sitter space}. We plot the entanglement entropy $S_A(t)=S(\nu(t))$ from (\ref{eq:expansion-J-eigenvalues}) associated to the subsystem described in (\ref{eq:subalgebra-phiA-Laplacian}). In the phase of linear growth, the entropy is observed to grow with rate given by the Hubble rate $H_0$ as predicted by theorem~\ref{th:SA}. We set $H_0=1$ and $R=1$ in this plot.}
	\label{fig:H-Comoving-expansion}
\end{figure}
In the quantum theory we consider an initial state at the time $t_0\to-\infty$ given by the Bunch-Davies vacuum. The correlation functions of this state at the time $t$ can be determined in closed form and are given by
\begin{equation}
	G_t(\vec{k},\vec{k}')=
	2\left(
	\begin{array}{cc}
		\frac{1}{2|\vec{k}|}\,e^{-2H_0 t} \,+\,\frac{H_0^2}{2|\vec{k}|^3}\;\;\quad&
		-\frac{H_0}{2|\vec{k}|}\,e^{+H_0 t} \\[1.5em]
		-\frac{H_0}{2|\vec{k}|}\,e^{+H_0 t}& 
		\frac{|\vec{k}|}{2}\,e^{+2H_0 t}
	\end{array}
	\right)\;(2\pi)^3\,\delta^{3}(\vec{k}+\vec{k}')\,,
	\label{}
\end{equation} 
from which we can read the complex structure $J_t(\vec{k},\vec{k}')$.

We analyze the entanglement growth of a subsystem spanned by a linear smearing of the field and the momentum. In order to guarantee that the dispersion of the linear observables are finite, we consider a smearing of the form
\begin{equation}
	\hat{\varphi}_A=\int \Delta\hat{\varphi}(\vec{x})\, f(\vec{x})\,d^3\vec{x}\,,\quad \hat{\pi}_A=\int \Delta\hat{\pi}(\vec{x})\, f(\vec{x})\,d^3\vec{x}\qquad\textrm{with}\qquad 
	\textstyle
	f(\vec{x})=\frac{1}{(\sqrt{2\pi}\, R)^3}e^{-\frac{|\vec{x}|^2}{2R^2}}\,,
	\label{eq:subalgebra-phiA-Laplacian}
\end{equation}
where $\Delta\hat{\varphi}(\vec{x})=\delta^{ij}\partial_i\partial_j\hat{\varphi}(\vec{x})$ is the comoving Laplacian and the Gaussian smearing is over a region of comoving size $R$. The eigenvalues of the restricted complex structure $[\mathrm{i}J_t]_A$ come in pairs $\pm\nu(t)$ and are given by
\begin{equation}
	\textstyle \nu(t)=\frac{16}{5\sqrt{3\pi}}\sqrt{1+\frac{1}{6}H_0^2 R^2\,e^{+2H_0\,t}}\,.
	\label{eq:expansion-J-eigenvalues}
\end{equation}
The entanglement entropy $S_A(t)=S(\nu(t))$ of the subsystem is plotted in figure~\ref{fig:H-Comoving-expansion} and for long time, i.e. for large number of efoldings, grows linearly as $S_A(t)\sim H_0\,t$. This is exactly the asymptotic growth predicted by theorem~\ref{th:SA} written in terms of the subsystem exponent $\Lambda_A=H_0$.\footnote{At the classical level, the momentum $\pi_A$ grows exponentially fast as $e^{H_0t}$ while the smeared field $\varphi_A$ has a norm which does not change exponentially and  does not approach $\pi_A$ exponentially fast. As a result the two vectors span a parallelogram whose area grows as $e^{H_0t}$ leading to a subsystem exponent $\Lambda_A=H_0$.} 

We note that previous studies of the growth of the entanglement entropy of cosmological perturbations focus on the $(\vec{k},-\vec{k})$ subsystem \cite{Campo:2005sy,Polarski:1995jg,Kiefer:1999sj,Martin:2015qta}. On the other hand the results presented here apply to all subsystems defined by smeared fields.

\section{Proof, part I: classical ingredients}\label{sec:classical}
In this section and the subsequent section we collect and prove results used in the proof of the main theorem presented in section~\ref{sec:growth}.\\

We consider a classical dynamical system with $N$ degrees of freedom. We assume that the system has a Hamiltonian dynamics defined in a linear phase space \cite{arnol2013mathematical}. We also restrict attention to quadratic time-dependent Hamiltonians. In this case we discuss the notions of stability, Lyapunov exponents and the growth of the volume of subsystems \cite{zaslavsky2008hamiltonian,vulpiani2010chaos}.

\subsection{Linear phase space and quadratic time-dependent Hamiltonians}\label{sec2:phasespace}
We consider a system with $N$ degrees of freedom described by a linear phase space $V=\mathbb{R}^{2N}$. Phase space observables $\mathcal{O}$ are smooth functions of $2N$ real variables denoted $\xi^a$,
\begin{equation}
\begin{split}
\mathcal{O}:\;\mathbb{R}^{2N}&\to \mathbb{R}\\
\xi^a\;\;&\mapsto \mathcal{O}(\xi)\,.
\label{}
\end{split}
\end{equation}
The space of observables is equipped with a Lie algebra structure defined by the Poisson brackets
\begin{equation}
\{f(\xi),g(\xi)\}=\Omega^{ab}\,\partial_a f(\xi)\,\partial_b g(\xi)
\label{}
\end{equation}
where $\Omega^{ab}$ is a nondegenerate antisymmetric matrix. In this paper we mostly focus on linear observables $v=v_a \xi^a$ and quadratic observables $\mathcal{O}=\frac{1}{2}h_{ab}\xi^a\xi^b$. We call $V^*$ the vector space formed by all linear observables, and denote by $v_a$ the elements of $V^*$ and by $w^a$ the elements of  $V$. The restriction of the Poisson brackets to the space of linear observables is
\begin{equation} 
\{u,v\} = \Omega^{ab} u_a v_b \, .
\label{eq:Omega-V}
\end{equation}
A Darboux basis\footnote{Technically, $\mathcal{D}_V$ is a basis of the dual phase space $V^*$, but we refrained from bloating our notation by writing $\mathcal{D}_{V^*}$. All Darboux bases $\mathcal{D}$ in this paper will live in the dual phase space $V^*$.} (also called symplectic basis) of phase space functions consists of a set $\mathcal{D}_V = (q_1, \dots, q_N, p_1,\dots, p_N)$ of linear observables\footnote{When using abstract indices, a Darboux basis $\mathcal{D}_V=(\xi^1_a,\dots,\xi^{2N}_a)$ consisting of $2N$ linear observables $\xi^a_b$ can be read as a concrete representation of the Kronecker delta $\delta^a{}_b=\xi^a_c\xi^c_b=\xi^a_b$ when we read both indices as abstract indices. However, when referring to an explicit basis $(q_1,\dots,q_N,p_1,\dots,p_N)$, we will use lower indices to match standard conventions.}
\begin{equation}
q_i=q_{ia}\,\xi^a \quad \mathrm{and}\quad p_i=p_{ia}\,\xi^a\quad \mathrm{with}\quad i=1,\dots,N
\label{eq:qp}
\end{equation}
satisfying canonical Poisson brackets $\{q_i,q_j\}=0$, $\{p_i,p_j\}=0$, $\{q_i,p_j\}=\delta_{ij}$. \\

The notions of symplectic vector space and symplectic transformations play a central role in the description of the system.  A symplectic structure on $V$ is an antisymmetric non-degenerate bilinear map $\omega_{ab}:V\times V\to\mathbb{R}$. It provides a canonical map from $V$ to $V^*$ given by $v_a=\omega_{ab}v^b$. The couple $(V,\omega_{ab})$ defines a symplectic vector space. The inverse of the symplectic structure, denoted $\Omega^{ab}$, is the antisymmetric bilinear map defined by $\Omega^{ac}\,\omega_{cb}=\delta^a{}_b$ and is a symplectic structure on $V^*$. The space $V^*$ of linear observables on a linear phase space, equipped with the bilinear map $\Omega^{ab}$ describing the restriction of the Poisson brackets to $V^*$ as in \eqref{eq:Omega-V}, is a symplectic vector space. In a Darboux basis, the symplectic structure $\Omega^{ab}$ and its inverse $\omega_{ab}$ take the $2N\times 2N$ matrix form $\Omega=(\Omega^{ab})$ and $\omega=(\omega_{ab})$,
\begin{equation}
\Omega=\left(
\begin{array}{c|c}
0& +\mathbbm{1}\\
\hline
-\mathbbm{1} & 0
\end{array}
\right)\,,\qquad
\omega\equiv\Omega^{-1}=\left(
\begin{array}{c|c}
0& -\mathbbm{1}\\
\hline
+\mathbbm{1} & 0
\end{array}
\right).
\label{}
\end{equation}
The linear symplectic group $\mathrm{Sp}(2N)$ is the group of $2N\times 2N$ matrices $M^a{}_b$ satisfying the relation $M^a{}_c\, M^b{}_d\,\Omega^{cd}=\Omega^{ab}$. In matrix form we have $M\Omega M\tra=\Omega$. Note that the inverse of a symplectic matrix is given by $M^{-1}=\Omega M\tra \omega$. The matrices $M^a{}_b$ can be interpreted as linear maps either on $V$ or on $V^*$, and preserve the corresponding symplectic structures. \\

The dynamics of a Hamiltonian system is prescribed by a Hamilton function $H(t)$ that we allow to be time-dependent. The Hamilton equations of motion of an observable $\mathcal{O}$ are
\begin{equation}
\dot{\mathcal{O}}(t)=\{\mathcal{O}(t),H(t)\}+\frac{\partial \mathcal{O}(t)}{\partial t}\,.
\label{}
\end{equation}
In particular, for the linear observables $\xi^a$ we have
\begin{equation}
\dot{\xi}^a(t)=\Omega^{ab}\partial_b H(t)\,.
\label{eq:xidot}
\end{equation}
In this paper we focus on time-dependent quadratic Hamiltonians, i.e. phase space functions of the form
\begin{equation}
H(t)=\frac{1}{2}h_{ab}(t)\,\xi^a\xi^b+f_a(t)\,\xi^a\,.
\label{eq:quadH}
\end{equation}
In this case the Hamilton equations simplify to the linear equation
\begin{equation}
\dot{\xi}^a(t)=K^a{}_b(t)\,\xi^b\,+\,\Omega^{ab}f_b(t)
\label{}
\end{equation}
where the matrix $K^a{}_b(t)$ is defined in terms of the quadratic term in the Hamiltonian by
\begin{equation}
K^a{}_b(t)=\Omega^{ac}h_{cb}(t)\,.
\label{eq:Kab}
\end{equation}
The solution of this equation provides the time evolution of the linear observable $\xi^a$,
\begin{equation}
\xi^a(t)=M^a{}_b(t)\,\xi^b(0)+\eta^a(t)\,.
\label{eq:xit}
\end{equation}
The matrix $M^a{}_b(t)$ solves the differential equation $\dot{M}^a{}_b(t)=K^a{}_c(t)M^c{}_b(t)$ with the identity as initial condition, and can be expressed as a time-ordered exponential,
\begin{equation}
M^a{}_b(t)=\mathcal{T}\!\exp\left(\int_0^t K^a{}_b(t')\,dt'\right).
\label{eq:Mh}
\end{equation}
As the time evolution preserves the Poisson brackets, the matrix $M^a{}_b(t)$ belongs to the linear symplectic group $\mathrm{Sp}(2N)$, i.e. $M^a{}_c(t) M^b{}_d(t)\Omega^{cd}\,=\Omega^{ab}$. The time-dependent shift $\eta^a(t)$ in (\ref{eq:xit}) satisfies $\dot{\eta}^a(t)=K^a{}_b(t)\eta^b(t)+\Omega^{ab}f_b(t)$. It is given by 
\begin{equation}
\eta^a(t)=M^a{}_b(t)\int_0^t M^{-1}(t')^b{}_c\,\Omega^{cd}f_d(t')\,dt'
\label{eq:etat}
\end{equation}
and it vanishes if the linear term $f_a(t)\xi^a$ is not present in the Hamiltonian.\\

A simple example of time-dependent quadratic Hamiltonian of the form (\ref{eq:quadH}) is given by a system of coupled harmonic oscillators with time-dependent couplings or driven by external forces. Another important example arises in the description of the lowest-order expansion of the evolution of a time-independent non-linear system around a classical solution chosen as background. In this case the time-dependence of the effective Hamiltonian arises from the background classical solution.

\subsection{Linear stability and Lyapunov exponents}
To characterize the linear stability of a dynamical system we consider a small perturbation $\delta \xi^a(t)$ of a classical solution $\xi^a_0(t)$ that satisfies the Hamilton equations. Substituting $\xi^a(t)=\xi^a_0(t)+\delta \xi^a(t)$ into  (\ref{eq:xidot}) and expanding at linear order in the perturbation we find the linear equation
\begin{equation}
\delta \dot{\xi}^a(t)=\left.\Omega^{ac}\partial_c\partial_b H(t)\right|_{\xi_0}\,\delta\xi^b(t)\,,
\label{}
\end{equation}
with $K^a{}_b(t)\equiv \left.\Omega^{ac}\partial_c\partial_b H(t)\right|_{\xi_0}$ the stability matrix of the classical solution $\xi_0^a(t)$. For the quadratic Hamiltonian (\ref{eq:quadH}) the stability matrix is simply given by $K^a{}_b(t)=\Omega^{ac}h_{cb}(t)$. As a result, the time evolution of the perturbation is given by
\begin{equation}
\delta \xi^a(t)=M^a{}_b(t)\;\delta\xi^b(0)
\label{eq:time-pert}
\end{equation}
where the symplectic matrix $M(t)$ is given by (\ref{eq:Mh}). In order to measure the separation of two configurations in phase space we introduce a metric $g_{ab}$, i.e. a positive definite symmetric bilinear, and define the norm $||\delta \xi||\equiv \sqrt{g_{ab}\,\delta \xi^a\,\delta \xi^b}$. The exponential rate of separation of two sufficiently close classical solutions is given by the Lyapunov exponent $\lambda_{\delta \xi}$ defined as
\begin{equation}
\lambda_{\delta \xi}=\lim_{t\to\infty}\frac{1}{t}\log\frac{\lVert\delta \xi(t)\rVert}{\lVert\delta \xi(0)\rVert}\,.
\label{}
\end{equation}
We note that the Lyapunov exponent $\lambda_{\delta \xi}$ is independent from the choice of metric $g_{ab}$ used to measure the distance between the classical trajectories $\xi_0^a(t)$ and $\xi_0^a(t)+\delta \xi^a(t)$. See appendix~\ref{app:Lyapunov} for a proof of this statement.

It is also useful to define Lyapunov exponents of linear observables $\ell(\delta\xi)=\ell_a\delta\xi^a$ that probe the perturbation $\delta\xi^a(t)$ and live in the dual phase space $\ell_a\in V^*$. From the time evolution equation~(\ref{eq:time-pert}), we can read off that $\ell(t)$ evolves as
\begin{equation}
	\ell_a(t)=M^b{}_a(t)\,\ell_b(0)=\left(M\tra\!(t)\,\ell(0)\right)_a\,.
\end{equation}
In order to define a norm $\lVert \ell\rVert=\sqrt{G^{ab}\ell_a\ell_b}$, we use the inverse metric $G^{ab}$, such that $G^{ac}g_{cb}=\delta^a{}_b$. Again, the Lyapunov exponent
\begin{equation}
\lambda_{\ell}=\lim_{t\to\infty}\frac{1}{t}\log\frac{\lVert\ell(t)\rVert}{\lVert\ell(0)\rVert}\,.
\label{}
\end{equation}
will be independent of the metric that we choose.

The metric $g_{ab}$, used above to define Lyapunov exponents, is said to be compatible with the symplectic structure $\omega_{ab}$ with inverse $\Omega^{ab}$ if the matrix $J^a{}_b\equiv\Omega^{ac}g_{cb}$ is symplectic, $J^a{}_c\, J^b{}_d\,\Omega^{cd}=\Omega^{ab}$, and squares to minus the identity $J^a{}_cJ^c{}_b=-\delta^a{}_b$. In this case, $J^a{}_b$ defines a complex structure. The inverse metric $G^{ab}$ is then compatible with the symplectic structure $\Omega^{ab}$ in the dual space. A compatible metric $g_{ab}$ allows us to define the limiting matrix $L_a{}^b$,
\begin{equation}
L_a{}^b\equiv\lim_{t\to\infty}\frac{1}{2t}\log \Big(g_{ac}\,M^{c}{}_d\,G^{de}\,M^{b}{}_e \Big)\,,
\label{eq:Lmatrix}
\end{equation}
that characterizes the long-time stability of the system.\footnote{In matrix form, $L=\lim_{t\to\infty}\frac{1}{2t}\log\big( gM(t)GM\tra(t)\big)$.} Provided that the Hamiltonian system is regular in the sense of appendix~\ref{app:regular}, Lyapunov exponents exist for all linear observables $\ell_a$ and are given by the eigenvalues of the limiting matrix $L_a{}^b$. As the matrix $L_a{}^b$ is symmetric and belongs to the symplectic algebra $\mathrm{sp}(2N)$, its eigenvalues are real and come in pairs with opposite sign $(\lambda,-\lambda)$. The Lyapunov spectrum consists of the ordered Lyapunov exponents given by
\begin{equation}
\lambda_1\geq\dots\geq \lambda_N\geq 0\geq \lambda_{N+1}\geq\dots\geq\lambda_{2N}\,,
\label{eq:lambdaOrder}
\end{equation}
with $\lambda_{2N+1-i}=-\lambda_i$ for regular Hamiltonian systems, as explained in appendix~\ref{app:Hamiltonian}. The dimension of an eigenspace is how often the associated exponent appears in this list. The eigenvectors $\ell_b$ of $L_a{}^b$ provide us with a Darboux basis of $V^*$ adapted to the unstable directions of the system. This basis, called the Lyapunov basis $\mathcal{D}_L$,
\begin{equation}
\mathcal{D}_L=(\ell^1,\dots, \ell^{2N}),
\label{}
\end{equation}
is defined so that the only non-trivial Poisson brackets are $\{\ell^i,\ell^{2N-i+1}\}=1$ for $i=1,\dots,N$ and $\lim_{t\to\infty}\frac{1}{t}\log \lVert\ell^i(t)\rVert/\lVert\ell^i(0)\rVert\,=\lambda_i$ with $\ell^i(0)=\ell^i$. Note that $\mathcal{D}_L$ is not unique because it depends on our choice of metric $G^{ab}$, but subsequent results will be independent of this choice \cite{eckmann1985ergodic,ginelli2007characterizing}.\\

We will discuss examples of unstable quadratic systems in section~\ref{sec:growth}. The prototypical example is the inverted harmonic oscillator, with a potential $V$ that is unbounded from below. The Lyapunov exponents of the system are related to the unstable directions of the potential. Another example is provided by periodically driven systems, i.e., systems with a quadratic time-dependent Hamiltonian of the form (\ref{eq:quadH}) with periodic coefficients $h_{ab}(t+T)=h_{ab}(t)$. In this case instabilities appear due to the phenomenon of parametric resonance \cite{arnol2013mathematical}. The real part of the Floquet exponents of the system coincide with the notion of Lyapunov exponents described above.

\subsection{Subsystems and the subsystem exponent $\Lambda_A$}\label{sec2:LambdaA}
The partition of a Hamiltonian system in two complementary subsystems corresponds to a decomposition of the phase space $V$ and its dual $V^*$ into direct sums
\begin{equation}
V=A\oplus B\qquad\text{and}\qquad V^*=A^*\oplus B^*
\label{eq:V=A+B}
\end{equation}
with dimension $\dim A=2N_A$, $\dim B=2N_B$ where $N_A$ is the number of degrees of freedom in the subsystem $A$ and $N_A+N_B=N$. This decomposition can be understood as induced by a choice of subspace of linear observables $\phi_i=\phi_{ia}\xi^a$ and $\pi_i=\pi_{ia}\xi^a$ with $i=1,\dots, 2N_A$ and canonical Poisson brackets $\{\phi_i,\phi_j\}=0$, $\{\pi_i,\pi_j\}=0$, $\{\phi_i,\pi_j\}=\delta_{ij}$.

This set of observables provides us with a Darboux basis of linear observables $A^*$ that only probe the degrees of freedom in $A$
\begin{equation}
\mathcal{D}_A=(\theta^1,\dots,\theta^{2N_A})=(\phi_i,\pi_i),
\label{}
\end{equation}
and it can be completed to a Darboux basis of the full system by introducing a Darboux basis of $B^*$,
\begin{equation}
\mathcal{D}_B=(\varTheta^1,\dots,\varTheta^{2N_B})=(\Phi_i,\Pi_i).
\label{}
\end{equation}
so that
\begin{equation}
\mathcal{D}_V=(\mathcal{D}_A,\mathcal{D}_B).
\label{}
\end{equation}

Given a Darboux basis $\theta^r$ of $A^*$ and its dual basis $\vartheta_r$ of $A$ satisfying $\theta^r_a\vartheta_s^a=\delta^r{}_s$, we can restrict tensors to the subsystem by appropriate contractions. Most importantly, we will consider the restriction $[J]_A$ of the complex structure $J^a{}_b$ and $[G]_A$ of a metric $G^{ab}$:
\begin{equation}
[J]_A=(\theta^r_a\,J^a{}_b\,\vartheta^b_s)\qquad\text{and}\qquad [G]_A=(\theta^r_a\,G^{ab}\,\theta_b^s)\,.
\label{}
\end{equation}
Note that, as $\theta^r_a$ is a Darboux basis, the restriction of the symplectic structure $\Omega^{ab}$ is still a symplectic structure, $[\Omega]_A=\Omega_A$. On the other hand, the restriction $[J]_A$ of a complex structure $J$ is not in general again a complex structure, because it does not necessarily satisfy $[J]_A^2=-\mathds{1}_A$. \\

We give some examples of subsystems. Consider for instance a linear chain of $N$ oscillators, with the oscillator at site $i$ having canonical coordinates $(q_i,p_i)$. A first example of subsystem $A$ corresponds to the subset of observables $(q_i,p_i)$ with $i=1,\dots,N_A$ associated to a geometric decomposition of the chain in two complementary intervals. A second example of subsystem is given by a subset of normal-mode observables $(\tilde{\phi}_k,\tilde{\pi}_k)$ with $k=1,\dots,N_A$ corresponding to the long-wavelength perturbations of the system. A third example is provided by a detector that makes measurement of the localized observables $(Q,P)$ only, with $Q=\frac{1}{N_d}\sum_{i=1}^{N_d} q_i$ and $P=\sum_{i=1}^{N_d} p_i$ probing only average properties of a localized subset of oscillators. Each example shows that the choice of a subsystem $A$ corresponds to a coarse graining of the system that preserves the symplectic structure of the accessible observables.\\

Given a subsystem $A$ we introduce a new notion of characteristic exponent $\Lambda_A$ that generalizes the notion of Lyapunov exponents of the system. A Darboux basis $\mathcal{D}_A=(\theta^1,\dots,\theta^{2N_A})$ of the subsystem defines a symplectic cube
\begin{equation}
	\mathcal{V}_A=\left\{\sum^{2N_A}_{r=1} c_r\,\theta^r\Bigg|\;0\leq c_i\leq 1\right\}\subset A^*\,.
\end{equation}
Given the metric $G^{ab}$, we can compute the volume $\mathrm{Vol}_G(\mathcal{V}_A)$ of the symplectic cube $\mathcal{V}_A$ as the square root of the determinant of the $2N_A\times 2N_A$ Gramian matrix $(\theta_a^r\, G^{ab}\, \theta_b^s)$,
\begin{equation}
\mathrm{Vol}_G(\mathcal{V}_A)\equiv\sqrt{\det (\theta_a^r\, G^{ab}\, \theta_b^s)}=\sqrt{\det [G]_A}\,.
\label{}
\end{equation}
We will be interested in how this volume changes under time evolution. Let us recall that the action on $V^*$ is given by the transpose $M\tra(t)_a{}^b=M^b{}_a(t)$. If we evolve the symplectic cube $\mathcal{V}_A$ with $M\tra$, we have
\begin{equation}
\mathrm{Vol}_G(M\tra(t)\mathcal{V}_A)=\sqrt{\det([M(t)\, G\,M\tra(t)]_A)}\,.
\label{}
\end{equation}
We define the subsystem exponent $\Lambda_A$ as the exponential rate of growth of the volume of the subsystem measured with respect to the metric $G^{ab}$,
\begin{equation}
\Lambda_A=\lim_{t\to\infty}\frac{1}{t}\log\frac{\mathrm{Vol}_G(M\tra(t)\mathcal{V}_A)}{\mathrm{Vol}_G(\mathcal{V}_A)}\,.
\label{eq:LambdaA}
\end{equation}
For a regular Hamiltonian system this limit exists and is independent of the metric $g_{ab}$, see appendix~\ref{app:regular}. Note that the exponent of the full system $\Lambda_V$ vanishes because $[M\tra(t)]_V=M\tra(t)$ and the determinant of a symplectic matrix is equal to one. The vanishing of $\Lambda_V$ is a special case of the Liouville theorem. We say that the subsystem $A$ is unstable under time evolution if it has a positive exponent $\Lambda_A$.

\subsection{Relation of the exponent $\Lambda_A$ to Lyapunov exponents}\label{sec2:LambdaA-Lyapunov}
We now show how to compute the subsystem exponent $\Lambda_A$ once the Lyapunov spectrum of the system is known. The result is stated and proven below.

\begin{theorem}[Subsystem exponent]
\label{th:LambdaA}
	Given a regular Hamiltonian system with Lyapunov spectrum $(\lambda_1,\dots,\lambda_{2N})$ and Lyapunov basis $\mathcal{D}_{L}=(\ell^1,\dots,\ell^{2N})$, the subsystem exponent $\Lambda_A$ associated to the symplectic decomposition $V=A^*\oplus B^*$ can be determined as follows:
	\begin{enumerate}
		\item Choose a Darboux basis $\mathcal{D}_A=(\theta^1,\dots,\theta^{2N_A})$ of the symplectic subspace $A^*\subset V^*$.
		\item Compute the unique transformation matrix $T$ that expresses $\mathcal{D}_A$ in terms of the Lyapunov basis $\mathcal{D}_{L}=(\ell^1,\dots,\ell^{2N})$:
		\begin{equation}
		\left(\begin{array}{c}
		\theta^1\\
		\vdots\\
		\theta^{2N_A}
		\end{array}\right)=\left(\begin{array}{ccc}
		\smash{\fbox{\color{black}\rule[-37pt]{0pt}{1pt}$\,\,T^1_1\,\,$}} & \cdots & \smash{\fbox{\color{black}\rule[-37pt]{0pt}{1pt}$\,T_{2N}^1$}}\\
		\vdots & \ddots & \vdots\\
		\myunderbrace{\mystrut{1.5ex}T_1^{2N_A}}{\vec{t}_1} & \cdots & \myunderbrace{\mystrut{1.5ex}T_{2N}^{2N_A}}{\vec{t}_{2N}}
		\end{array}\right)\left(\begin{array}{c}
		\ell^1\\
		\vdots\\
		\ell^{2N}
		\end{array}\right)\color{white}{\begin{array}{c}
			.\\ \\ \\ \\ ,
			\end{array}}\color{black}
			\label{eq:Tv}
		\end{equation}

We refer to the $2N$ columns of $\,T$ as $\vec{t}_i$.
		\item Find the first $2N_A$ linearly independent \footnote{Here, we mean that $\vec{t}_i$ cannot be expressed as a linear combination of the vectors $(\vec{t}_1,\dots,\vec{t}_{i-1})$ standing to the left in the matrix $T$.} columns $\vec{t}_i$ of $T$ which we can label by $\vec{t}_{i_k}$ with $k$ ranging from $1$ to $2N_A$. The result is a map $k\mapsto i_k \in (1,\dots,2N)$ with $i_{k+1}>i_k$.
	\end{enumerate}
	The subsystem exponent $\Lambda_A$ is then given by the sum over the $2N_A$ Lyapunov exponents $\lambda_{i_k}$,
	\begin{equation}
	\Lambda_A=\sum^{2N_A}_{k=1}\lambda_{i_k}\,,
	\label{eq:LambdaAsum}
	\end{equation}
	where the index $i_k$ is defined above.
\end{theorem}
\begin{proof}
The rectangular matrix $T$ in (\ref{eq:Tv}) allows us to express the elements of the Darboux basis $\mathcal{D}_A$ of the subsystem in terms of the Lyapunov basis, $\theta^r=\sum^{2N}T^r_i \ell^i$. Denoting the columns of $T$ by $\vec{t}_i$ we can select the first $2N_A$ linearly independent columns in the ordered set $(\vec{t}_1,\dots,\vec{t}_{2N})$. We label them $\vec{t}_{i_k}$ and organize them in the $2N_A\times 2N_A$ square matrix $U$,
\begin{equation}
U=\left(\begin{array}{c|c|c}
\vec{t}_{i_1} & \dots & \vec{t}_{i_{2N_A}}
\end{array}\right)\,.
\end{equation}
Due to their linear independence, the inverse $U^{-1}$ exists and turns $T$ into an upper triangular matrix $\tilde{T}$ of the form
\begin{equation}
\tilde{T}=U^{-1}T=\left(\begin{array}{lll}
0 &\cdots \, 0\;  1 \,*\, * \cdots\cdots\cdots\cdots\cdots\cdots&*\\
0 &\cdots \cdots\cdots \, 0\;  1 \,*\, *\cdots\cdots\cdots\cdots&*\\
\,\vdots &\qquad\qquad \vdots\qquad\qquad \vdots\qquad\qquad&\vdots\\
0 &\cdots \cdots\cdots\cdots \cdots \cdots 0\;  1 \,*\, *\;\cdots&*
\end{array}\right)\,,
\end{equation}
where the $*$ represents an unspecified value. Acting with $U^{-1}$ on the left and the right-hand side of (\ref{eq:Tv}) we find
\begin{equation}
\tilde{\theta}^k=\ell^{i_k}+\sum^{2N}_{j>i_k}\tilde{T}^k_j\, \ell^j\,,
\label{eq:wk}
\end{equation}
where $\tilde{\theta}_k=(U^{-1}\theta)_k$. Note that the vectors $\tilde{\theta}^k$ satisfy $\lim_{t\to\infty}\frac{1}{t}\log \lVert M\tra(t)\tilde{\theta}^k\rVert/\lVert\tilde{\theta}^k\rVert\,=\,\lambda_{i_k}$. Moreover the $2N_A$ vectors $\tilde{\theta}^k$ are linearly independent and form a (generally non-symplectic) basis of $A^*$. Therefore the cube $M\tra(t)\,\mathcal{V}_A$ is given by a time-independent linear transformation of the one spanned by $M\tra(t)\,\tilde{\theta}_k$. In the limit $t\to\infty$ the volume of the subsystem scales as $\mathrm{Vol}_G(M\tra(t)\,\mathcal{V}_A)\sim \exp(\sum_{k=1}^{2N_A} \lambda_{i_k} t)$ if there are no directions that become collinear in an exponentially fast way under time evolution. As the exponential collinearity is excluded by the assumption of regularity (see appendix~\ref{app:regular}), the subsystem exponent is given by (\ref{eq:LambdaAsum}).
\end{proof}

\begin{figure}
	\centering
	\noindent\makebox[\linewidth]{
		\includegraphics[width=.9\linewidth]{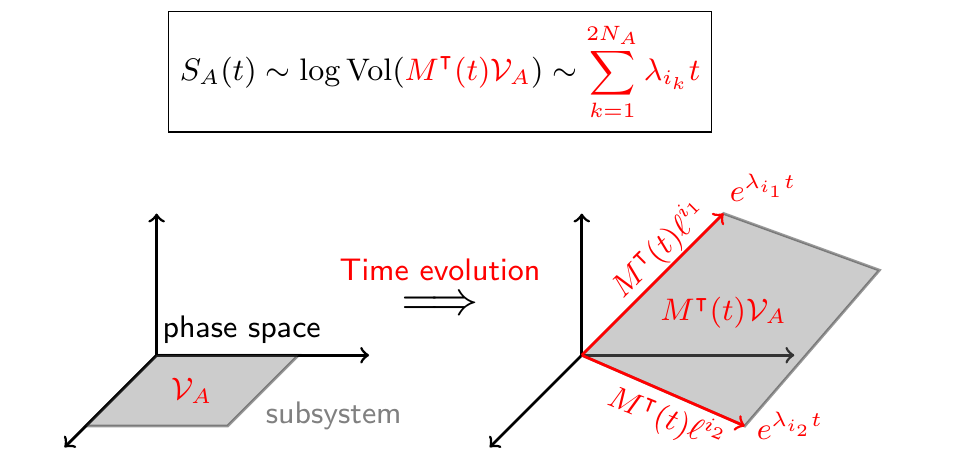}
	}
	\caption{\emph{Subsystem exponents due to phase space stretching}. We illustrate the statement of theorem~\ref{th:LambdaA}. We start with a symplectic cube $\mathcal{V}_A\subset A^*$ in the subsystem and time-evolve it to the deformed cube $M\tra(t)\mathcal{V}_A$ that is dominantly stretched into the $2N_A$ directions of $M\tra(t)\ell^{i_k}$ with Lyapunov exponents $\lambda_{i_k}$. Consequently, the logarithm of its metric volume behaves as $\log \mathrm{Vol}(M\tra(t)\mathcal{V}_A\color{black})\sim\sum^{2N_A}_{k=1}\lambda_{i_k}t$. In generic situations, $\lambda_{i_k}$ are just the $2N_A$ largest Lyapunov exponents, as explained in theorem~\ref{th:LambdaA-generic}. The quantity $\log\mathrm{Vol}(M\tra(t)\mathcal{V}_A)$ is related to the entanglement entropy $S_A$ as explained in section~\ref{sec2:Renyi-phase-space}.}
	\label{fig:phasespace-strething}
\end{figure}

This theorem, together with the fact that both $\mathcal{D}_{L}=(\ell^1,\dots,\ell^{2N})$ and $\mathcal{D}_A=(\theta^1,\dots,\theta^{2N_A})$ are symplectic bases, implies the following important property of the subsystem exponent $\Lambda_A$.

\begin{corollary}
	The subsystem exponent is non-negative, 
\begin{equation}
\Lambda_A\geq 0\,.
\label{}
\end{equation}
\end{corollary}
\begin{proof}
Let us denote the elements of the Lyapunov basis by $\ell^i=Q_i$ for $i\leq N$ and $\ell^i=P_{2N+1-i}$ for $i>N$ so that $(\ell^1,\dots,\ell^{2N})=(Q_1,\dots,Q_N,\,P_N,\dots,P_1)$. Each vector $\tilde{v}_{i_k}$ with $i_k>N$ consists of a linear superposition of momenta $P_i$ only, as follows from (\ref{eq:wk}). As a result, to span a symplectic subspace, for each such $\tilde{v}_{i_k}$ there has to be a $\tilde{v}_{i_{k'}}$ with $i_{k'}\leq 2N+1-i_k$ so to contain the conjugate position $Q_i$ in the linear superposition. Therefore, negative Lyapunov exponents $\lambda_{i_k}$ with $i_k>N$ are paired with positive Lyapunov exponents $\lambda_{i_{k'}}$, resulting in a sum of non-negative terms $\lambda_{i_k}+\lambda_{i_{k'}}\geq 0$ in (\ref{eq:LambdaAsum}).
\end{proof}

We illustrate this result with some examples of subsystems and the associated exponents. Consider a system with $N=2$ degrees of freedom, Lyapunov spectrum 
\begin{equation}
(\lambda_1,\lambda_2,-\lambda_2,-\lambda_1)
\label{}
\end{equation}
and Lyapunov basis $\mathcal{D}_L=(\ell^1,\ell^2,\ell^3,\ell^4)=(Q_1,Q_2,P_2,P_1)$. A subsystem $A$ with $N_A=1$ degree of freedom can be identified by specifying a canonical couple $(\phi,\pi)$. Here we give three examples:
\begin{align}
\label{eq:sub-example}
(1)\,\,&\left\{
\begin{array}{cl}
\phi&=\,Q_1\\[4pt]
\pi&=\,Q_2+P_1
\end{array}
\right.&\Rightarrow\quad T=\left(\begin{array}{cccc}
1 & 0 & 0 & 0\\
0 & 1 & 0 & 1
\end{array}\right)\quad\Rightarrow\quad
\Lambda_A=\lambda_1+\lambda_2\geq 0\,,\\[1em]
(2)\,\,&\left\{
\begin{array}{cl}
\phi&=\,Q_1+Q_2\\[4pt]
\pi&=\,P_2
\end{array}
\right.&\Rightarrow\quad T=\left(\begin{array}{cccc}
1 & 1 & 0 & 0\\
0 & 0 & 1 & 0
\end{array}\right)\quad\Rightarrow\quad
\Lambda_A=\lambda_1-\lambda_2\geq 0\,,
\label{eq:l2-l3}\\[1em]
(3)\,\,&\left\{
\begin{array}{cl}
\phi&=\,Q_1+Q_2\\[4pt]
\pi&=\,P_1
\end{array}
\right.&\Rightarrow\quad T=\left(\begin{array}{cccc}
1 & 1 & 0 & 0\\
0 & 0 & 0 & 1
\end{array}\right)\quad\Rightarrow\quad
\Lambda_A=\lambda_2-\lambda_2=0\,.
\label{}
\end{align} 
In particular the subsystem given by a single couples $(Q_i,P_i)$ has vanishing subsystem exponent $\Lambda_A=0$. Note also that the difference of positive Lyapunov exponents can appear as in example (\ref{eq:l2-l3}). We will reconsider these examples in section~\ref{sec2:constrained-particle} and relate the subsystem exponents $\Lambda_A$ to the production of entanglement entropy.

\subsection{Relation of the exponent $\Lambda_A$ to the Kolmogorov-Sinai entropy rate}\label{sec2:LambdaA-KS}
From an information-theory perspective, the Hamiltonian evolution of a dynamical system with sensitive dependence on initial conditions produces entropy. This is because two initial conditions that are indistinguishable at a fixed resolution will evolve into distinguishable states after a finite time. 
The Kolmogorov-Sinai entropy rate provides a quantitative characterization of this behavior: It measures the uncertainty remaining on the next state of a system, if
an infinitely long past is known \cite{kolmogorov1958new,Sinai:2009,zaslavsky2008hamiltonian,vulpiani2010chaos}. It is defined as follows.

We decompose the phase space $V$ into cells $(\mathcal{C}_1,\dots,\mathcal{C}_n)$ belonging to a partition $\mathcal{P}$. Given a sampling time $\Delta t$, we can compute the probability $\mu(\mathcal{C}_1,\dots,\mathcal{C}_n)$ that a trajectory starting in cell $\mathcal{C}_1$ will successively go through $\mathcal{C}_2$, $\mathcal{C}_3$ and so on. The entropy per unit time with respect to such a given partition is given by Shannon's formula,
\begin{equation}
	\mathfrak{h}(\mathcal{P})=-\lim_{\Delta t\to 0}	
	\lim_{n\to\infty}\frac{1}{n\,\Delta t}\sum_{\mathcal{C}_1,\dots,\mathcal{C}_n}\mu(\mathcal{C}_1,\dots,\mathcal{C}_n)\log\mu(\mathcal{C}_1,\dots,\mathcal{C}_n)\,.
\end{equation}
The Kolmogorov-Sinai entropy rate is then defined as the supremum over all possible partitions:
\begin{equation}
	h_\mathrm{KS}\equiv\mathrm{sup}_{\mathcal{P}}\,\mathfrak{h}(\mathcal{P})\,.
	\label{eq:hKSdef}
\end{equation}
The quantity $h_\mathrm{KS}$ is a global invariant of the system and it provides a quantitative characterization of the notion of deterministic chaos in a Hamiltonian system.

A positive Lyapunov exponent corresponds to the exponential divergence in time of some initially nearby trajectories. This phenomenon results in the unpredictability of the evolution at finite resolution, and therefore contributes to $h_\mathrm{KS}$. Pesin's theorem \cite{pesin1977characteristic,eckmann1985ergodic} states that, for Hamiltonian dynamical systems, the Kolmogorov-Sinai entropy rate is equal to the sum over all the positive Lyapunov exponents of the system. Let us call $N_I\leq N$ the number of non-vanishing positive Lyapunov exponents. Using the ordering (\ref{eq:lambdaOrder}) of the Lyapunov spectrum, we have
\begin{equation}
h_\mathrm{KS}=\sum_{i=1}^{N_I}\lambda_i\,.
\label{eq:Pesin}
\end{equation}
This formula, together with (\ref{eq:LambdaAsum}), clearly shows that the Kolmogorov-Sinai entropy rate provides an upper bound to the characteristic exponent $\Lambda_A$ of a subsystem,
\begin{equation}
\Lambda_A\leq h_\mathrm{KS}\,.
\label{eq:LleqKS}
\end{equation}
In the following we discuss when this inequality is saturated and show that, for a large class of system decompositions, the characteristic exponent $\Lambda_A$ equals the rate $h_\mathrm{KS}$. The following theorem is instrumental.

\begin{theorem}[Subsystem exponent -- generic subsystem]
\label{th:LambdaA-generic}
	The subsystem exponent of a generic subsystem $A$ of dimension $N_A$ is given by the sum of the first $2N_A$ Lyapunov exponents, a special case of (\ref{eq:LambdaAsum}),
	\begin{equation}
	\Lambda_{A\,\mathrm{generic}}=\sum_{i=1}^{2N_A}\lambda_i\,.
	\label{eq:LambdaGeneric}
	\end{equation}
This behavior holds for all subsystems $A\in V$, except for a set of measure zero. 
\end{theorem}

\begin{proof}
The space of $2N_A$-dimensional symplectic subspaces of $V$ has the structure of a differentiable manifold and is called the symplectic Grassmannian $\mathrm{SpGr}(2N_A,V)$. Let us consider the set of points on this manifold where the generic asymptotics (\ref{eq:LambdaGeneric}) does \emph{not} apply. The statement of the theorem is that this set forms a lower dimensional submanifold. All standard measures on differentiable manifolds will therefore assign a measure zero to this subset.\\
By applying theorem \ref{th:LambdaA}, (\ref{eq:LambdaAsum}), we find $\Lambda_{A\,\mathrm{generic}}=\sum_{i=1}^{2N_A}\lambda_i$ whenever the first $2N_A$ columns of the transfer matrix $T$ are linearly independent. Let us therefore analyze for how many system decompositions this does not hold.
The space of $2N_A$-dimensional symplectic subspaces $\mathrm{SpGr}(2N_A,V)$ can be identified with the space of transformation matrices such that the restricted symplectic form $[\Omega]_A$ is non-degenerate, modulo $\mathrm{GL}(2N_A)$:
\begin{equation}
	\mathrm{SpGr}(2N_A,V)=\{T\in\mathrm{Mat}(2N\!\times\!2N_A)|\det{(T\Omega T\tra)}\neq 0\}\,/\,\mathrm{GL}(2N_A)\,.
\end{equation}
This follows from the fact that, for a given choice of Lyapunov basis and of a Darboux basis of $A$, every subspace $A \in  \mathrm{SpGr}(2N_A,V)$ defines a unique transfer matrix $T$. The different basis choices are equivalent to acting with a $\mathrm{GL}(2N_A)$-matrix on $T$ from the left. The space of full rank $(2N\!\times\! 2N_A)$-matrices is $(2N)(2N_A)$-dimensional and $\mathrm{GL}(2N_A)$ is $(2N_A)^2$-dimensional. The condition $\det{(T\Omega T^t)}\neq 0$ only cuts out a lower dimensional submanifold which does not change the dimension. This fact implies that the dimension of $\mathrm{SpGr}(2N_A,V)$ is $4N_A(N-N_A)$.\\
Let us now compare this to the space of subspaces for which the subsystem exponent $\Lambda_A$ is not given by the sum over the first $2N_A$ largest Lyapunov exponents. For this to happen, it is a necessary condition that the first $2N_A$ columns of the transfer matrix are linearly dependent. This space is $(4NN_A-1)$-dimensional which we still need to quotient by $\mathrm{GL}(2N_A)$. Therefore, the subset of spaces of subsystems $A$, for which we find $\Lambda_A\neq\Lambda_{A\,\mathrm{generic}}$, has a dimension of at most $4N_A(N-N_A)-1$. This is a set of measure zero with respect to any standard measure on $\mathrm{SpGr}(2N_A,V)$ because it lies in a lower dimensional submanifold.
\end{proof}

This behavior was conjectured by Zurek and Paz in \cite{Zurek:1994wd} and later discussed by Asplund and Berenstein \cite{Asplund:2015osa} and ourselves \cite{Bianchi:2015fra}.

Of the three examples discussed at the end of section~\ref{sec2:LambdaA}, only the one-dimensional subsystem $(\phi,\pi)$ with $\Lambda_A=\lambda_1+\lambda_2$ is generic, (\ref{eq:sub-example}). Note that most numerical algorithms for the computation of the Lyapunov exponents of a dynamical system start with the computation of the exponential rate of expansion of the volume of a subsystem \cite{bennetin1980lyapunov}. Lyapunov exponents are computed by taking the difference between the exponential rate of expansion of subsystems of different dimension. The efficiency of these algorithms relies on the generic behavior discussed above.\\

Now we investigate when the subsystem exponent  equals the rate $h_{\mathrm{KS}}$ assuming that the subsystem is generic, (\ref{eq:LambdaGeneric}). 

In a stable Hamiltonian system, all Lyapunov exponents vanish. The system becomes unstable as soon as a single Lyapunov exponent turns positive. We call $N_I$ the number of non-vanishing positive Lyapunov exponents. Pesin's formula for the Kolmogorov-Sinai entropy rate then reads $h_\mathrm{KS}=\sum^{N_I}_{i=1}\lambda_i$. On the other hand the characteristic exponent of a generic subsystem $A$ of dimension $N_A$ in the range $N_I\leq 2 N_A\leq 2N-N_I$ is given by $\Lambda_{A\,\mathrm{generic}}=\sum_{i=1}^{N_I}\lambda_i$. Therefore, we have the equality
\begin{equation}
\Lambda_{A\,\mathrm{generic}}=h_\mathrm{KS}\qquad\textrm{for}\qquad N_I\leq 2 N_A\leq 2N-N_I\,,
\label{eq:saturate}
\end{equation}
that identifies subsystems that saturates the inequality (\ref{eq:LleqKS}).

Note that unstable many-body systems often have a number of unstable directions $N_I$ that is much smaller that the number of degrees of freedom of the system, $N_I\ll N$. A generic subsystem that encompasses a fraction $f=N_A/N$ of the full system satisfies (\ref{eq:saturate}) if the fraction is in the range 
\begin{equation}
\frac{N_I}{2N}\,\leq\, f\,\leq 1-\frac{N_I}{2N}\,.
\label{}
\end{equation}
In particular, in the limit $N\to \infty$ with $N_I$ finite, we have $\Lambda_{A}=h_\mathrm{KS}$ for all partitions of the system into two complementary subsystems each spanning a finite fraction $f$ of the system, except for a set of partitions of measure zero.

\section{Proof, part II: quantum ingredients}\label{sec:entropy}
The main result, theorem~\ref{th:SA}, is proven in three steps which heavily rely on the three ingredient presented in the following subsections. We show how the Renyi entropy provides bounds for the entanglement entropy, we explain how the Renyi entropy can be understood as the logarithm of the volume of a region in the dual phase space and finally, we derive the time evolution of the Renyi entropy as the volume deformation of this region under the classical symplectic flow.

\subsection{Upper and lower bounds on the entanglement entropy}
\label{sec2:bounds}
\begin{figure}[t]
	\centering
	\noindent\makebox[\linewidth]{
		\includegraphics[width=.9\linewidth]{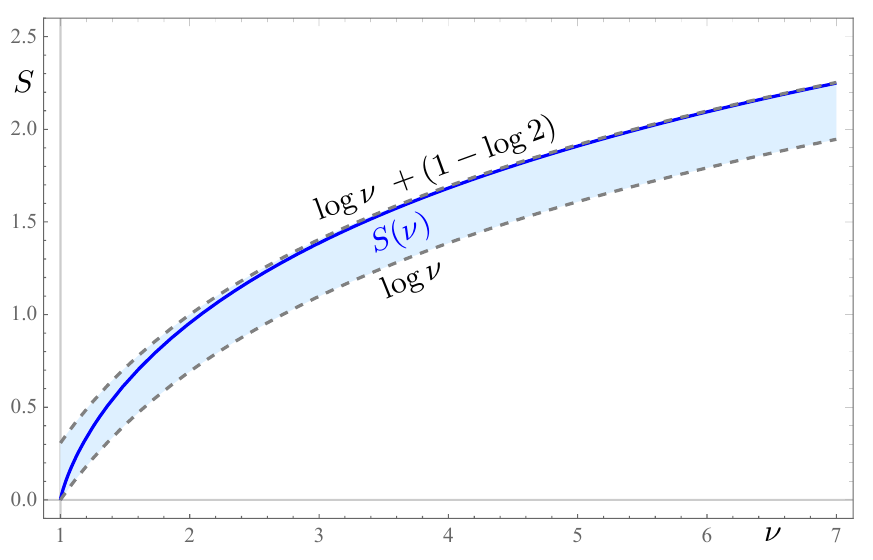}
	}
	\caption{\emph{Bounds on the entanglement entropy}. The plot shows how the contribution $S(\nu)$ to the entanglement entropy coming from a single entangled pair is bounded from below by $\log \nu$ and from above by $\log\nu\;+(1-\log 2)$. For large $\nu$, the asymptotic behavior is $S(\nu)\sim \log\nu\;+(1-\log 2)-O(\nu^{-2})$.}
	\label{fig:Renyi-Entanglement}
\end{figure}
We recall that there are different entanglement measures that quantify the amount of correlations in a state $|\psi\rangle$ with respect to some system decomposition into subsystems $A$ and $B$. Beside the entanglement entropy $S_A(|\psi\rangle)$, we have the class of Renyi entropies defined by
\begin{align}
	R_A^{(n)}(|\psi\rangle)=-\frac{1}{n-1}\log\mathrm{Tr}_{\mathcal{H}_A}(\rho_A{}^n)\,,\label{eq:Sentropy}
\end{align}
where $S_A(|\psi\rangle)=\lim_{n\to1} R_A^{(n)}(|\psi\rangle)$. For a Gaussian state $|J,\eta\rangle$ labeled by a complex structure $J$, all these entropies can be computed directly from the eigenvalues $\pm\ii\nu_i$ of $[J]_A$, the complex structure restricted to the subsystem $A$. If we take the positive value $\nu_i$ of each eigenvalue pair, the Renyi entropy\footnote{From now on, we will refer to the Renyi entropy of order 2 as \emph{the} Renyi entropy.} $R_A=R_A^{(2)}$ and the entanglement entropy $S_A$ are given by
\begin{align}
	R_A=\sum^{N_A}_{i=1}\log(\nu_i)\quad\text{and}\quad S_A=\sum^{N_A}_{i=1}S(\nu_i)\quad\text{with}\quad S(\nu)=\frac{\nu+1}{2}\log\frac{\nu+1}{2}-\frac{\nu-1}{2}\log\frac{\nu-1}{2}\,,
\end{align}
which is derived in appendix \ref{sec2:renyi}.
Here, we derive upper and lower bounds on the entanglement entropy of Gaussian states. Consider the function $S(\nu)$ defined in (\ref{eq:Sentropy}) and the inequality
\begin{equation}
0\leq \;S(\nu)-\log \nu \;<(1-\log 2)\,\approx 0.31
\label{}
\end{equation}
holding for $\nu\geq 1$ as shown in figure~\ref{fig:Renyi-Entanglement}. An immediate consequence of this inequality is that the entanglement entropy of a Gaussian state is bounded from below by the R\'{e}nyi entropy and from above by the R\'{e}nyi entropy plus a state-independent constant,
\begin{equation}
R_A(|J,\zeta\rangle)\;\leq\; S_A(|J,\zeta\rangle)\;<\; R_A(|J,\zeta\rangle)+ (1-\log 2)\,\mathrm{min}(N_A,N_B)\,.
\label{eq:SA-RA}
\end{equation}
This means that the Renyi entropy $R_A$ determines a corridor for the entanglement entropy $S_A$. This implies immediately that both of them will grow asymptotically with the same rate which we use for the main result of this paper.

\subsection{Renyi entropy as phase space volume}
\label{sec2:Renyi-phase-space}
A Gaussian state $|J,\eta\rangle$ equips the dual phase space $V^*$ with a metric $G^{ab}$ defined by (\ref{eq:Jg}), which is really just the covariance matrix of the state. The complex structure $J$ can be expressed in terms of the metric
\begin{equation}
J^a{}_b=-G^{ac}\omega_{cb}\,.
\label{eq:J=Og}
\end{equation}
Furthermore, the restriction of the complex structure to the subsystem $A$ can be written in matrix form as a product of the symplectic $\Omega_A$ and the restriction of the metric,
\begin{equation}
[J]_A=-[G]_A \, \omega_A\,.
\label{}
\end{equation}
In a Darboux basis, where we have $\det\omega_A\,=1$ and $\det [G]_A\,>0$, we find that the determinant of the restriction of the complex structure can be expressed in terms of the phase space volume $\mathrm{Vol}_G(\mathcal{V}_A)$ of a symplectic cube $\mathcal{V}_A$ (spanned by a Darboux basis and with symplectic volume $1$) measured with respect to the induced metric.
\begin{equation}
|\det[\mathrm{i}J]_A\big|=\det [G]_A\;\det \omega_A=\Big(\mathrm{Vol}_G(\mathcal{V}_A)\Big)^2.
\label{eq:Renyi=Vol2}
\end{equation}
As a result, the R\'{e}nyi entropy of a Gaussian state is given by the logarithm of the phase space volume of a symplectic cube $\mathcal{D}_A$ defining the subsystem, measured with respect to the metric $G^{ab}$ defined by the state,
\begin{equation}
R_A(|J,\zeta\rangle)=\log\mathrm{Vol}_G(\mathcal{V}_A)\,.
\label{eq:RAvol}
\end{equation}
Note that the symplectic cube $\mathcal{V}_V$ associated to a Darboux basis of the full system satisfies $\mathrm{Vol}_G(\mathcal{V}_V)=1$ and therefore the R\'{e}nyi entropy vanishes $R_V(|J,\zeta\rangle)$. On the other hand, the restriction to a subsystem $A$ can result in a larger volume $\mathrm{Vol}_G(\mathcal{V}_A)\geq 1$ and a non-vanishing R\'{e}nyi entropy.

\subsection{Entanglement entropy growth as phase space stretching}\label{sec2:stretching}
Let us consider a one-parameter family of Gaussian states $|J_t,\eta_t\rangle=U(t) |J_0,\eta_0\rangle$ generated under time evolution of some quadratic Hamiltonian. In particular, we have $J_t=M(t)J_0M^{-1}(t)$ where $M(t): V\to V$ is the classical Hamiltonian flow on phase space. We call $G_t$ the time-dependent metric associated with $J_t$, and $G_0$ the initial metric associated with $J_0$. The evolution of the R\'enyi entropy is given by \eqref{eq:RAvol}, where the volume is now measured with respect to the time-varying metric:
\begin{equation}
R_A(U(t) |J_0,\eta_0\rangle)=\log\mathrm{Vol}_{G_t}(\mathcal{V}_A)\, .
\label{eq:RAvol-t}
\end{equation}
In this formula, the symplectic cube $\mathcal{V}_A$ is kept fixed while the metric evolves. However, the same volume is obtained if we let the symplectic cube evolve according to $M\tra(t) \mathcal{V}_A$ for a fixed metric $G_0$. Hence, we can compute
\begin{equation}
R_A(U(t) |J_0,\eta_0\rangle)=\log\mathrm{Vol}_{G_0}(M\tra(t) \mathcal{V}_A)\, .
\label{eq:RAvol-t}
\end{equation}
The symplectic basis of the subsystem $A$ is stretched by the Hamiltonian flow $M\tra(t):V^* \to V^*$ on the dual phase space, and the variation in its volume determines the evolution of the R\'enyi entropy.

Since the absolute difference between the entanglement entropy and the R\'enyi entropy of a Gaussian states is bounded by a state independent constant, we have that:
\begin{equation}
\lim_{t\to\infty} \frac{1}{t} \left[ S_A(U(t)|J_0,\eta_0\rangle)- R_A(U(t)|J_0,\eta_0\rangle)\right] = 0\,,
\end{equation}
i.e., the asymptotic rate of growth of the entanglement entropy and of the R\'{e}nyi entropy coincide. This allows us to compute the asymptotic rate of growth of the entanglement entropy from \eqref{eq:RAvol-t} as:
\begin{equation}
\lim_{t\to\infty} \frac{S_A(U(t)|J_0,\eta_0\rangle)}{t} = \lim_{t\to\infty} \frac{1}{t} \log\mathrm{Vol}_{G_0}(M\tra(t) \mathcal{V}_A)\, ,
\end{equation}
in terms of the stretching of the symplectic cube under time-evolution.

\section{Discussion}\label{sec:discussion}
We discuss the role of interactions in the saturation phase of the entanglement growth, explain the relation to results on quantum quenches, present a conjecture on entanglement and chaos, and summarize our results.

\subsection{Interactions and the production of non-Gaussianities}\label{sec:interaction}
Quadratic Hamiltonians appear naturally in the analysis of small perturbations around equilibrium configurations, both stable and unstable. Let us consider for instance a dynamical system with a ``Mexican hat'' potential and an initial Gaussian state which is sufficiently peaked at the top of the potential. For short times the evolution of this initial state is well described by a perturbative quadratic Hamiltonian with unstable directions as in (\ref{eq:potential}). As a result, if the scales of the problem are sufficiently separated, the entanglement entropy of a subsystem will show an intermediate linear growth with rate $\Lambda_A$, followed by a non-Gaussian phase. In particular the linear growth driven by the perturbative instability stops when the spread of the state starts to probe the bottom of the potential and interactions become non-negligible. As the full Hamiltonian of the system is stable at the non-perturbative level, the entanglement entropy of the subsystem is bounded from above by the entanglement entropy $S_{\mathrm{eq}}$ of the thermal state with the same energy as the initial Gaussian state. In the presence of an equilibration and a thermalization mechanism, the entropy $S_{\mathrm{eq}}$ provides also the saturation value as shown in figure~\ref{fig:A-Illustration}.

Quadratic Hamiltonians appear also in the analysis of small perturbations of classical solutions. In this case the perturbative Hamiltonian inherits the time-dependence of the classical solution. For instance if the classical solution is periodic in time, then it provides a time-dependent background for the perturbations which leads to a perturbative Hamiltonian that is periodic in time. The stroboscopic dynamics of the system can be analyzed with the same methods discussed for the Hamiltonian (\ref{eq:Hperiodic}). In particular, in the presence of parametric resonances, the Floquet exponents of the system determine the subsystem exponent $\Lambda_A$ and the growth of the entanglement entropy as described in theorems~\ref{th:SA} and~\ref{th:SA-generic}. When the conditions (\ref{eq:SAsaturate}) for the subsystem are satisfied, the rate of growth is given by the classical Kolmogorov-Sinai entropy rate as discussed also in \cite{Asplund:2015osa}. After the initial phase of linear growth, two distinct phenomena render the parametric resonance inefficient and lead to a saturation phase. The first phenomenon is dephasing: large perturbations are not harmonic; their period depends on the amplitude of the oscillation and, when the period is driven far from resonance, the periodic background cannot pump energy efficiently into the perturbation. The second phenomenon is backreaction: clearly the linear entropy growth is accompanied by the production of a large number of excitations that at some point start to interact and backreact, thus leading to a saturation phase in which non-Gaussianities cannot be neglected. The phase of linear growth manifests itself only if the typical scales of the problem are sufficiently separated. A preliminary numerical investigation of the effect of interactions and non-Gaussianities on the entanglement growth can be found in \cite{Hackl:2017ndi}.

\subsection{Relation to linear growth in quantum quenches}
Quantum quenches lead also to a phase of linear growth of the entanglement entropy, followed by a saturation phase. This phenomenon has been studied extensively in free field theories \cite{Calabrese:2005in,Calabrese:2004eu,cotler2016entanglement} and in many-body quantum systems \cite{calabrese2007quantum,dechiara2006entanglement,fagotti2008evolution,eisler2008entanglement,lauchli2008spreading,kim2013ballistic,alba2017entanglement}. Despite the similarities in the behavior of the entanglement entropy, the mechanism behind this phenomenon is distinct from the one discussed in this article.

A standard example of global quantum quench is provided by an harmonic lattice similar to the one discussed in section~\ref{sec:periodic-quenches}. The Hamiltonian of the system consists of two terms,
\begin{equation}
	H_\kappa=\frac{1}{2}\sum^{N}_{i=1}\big(p_i^2+\,\Omega^2_0\;q_i^2\big)\; +\;\,\frac{1}{2}\kappa\sum^{N}_{i=1} \,(q_{i+1}-q_i)^2\,,
	\label{eq:Hquench}
\end{equation}
the first term is ultralocal, while the second encodes the coupling of first neighboring oscillators. A quantum quench consists in preparing the system in the ground state $|\psi_0\rangle$ of the Hamiltonian $H_0$ with vanishing coupling $\kappa=0$. At the time $t=0$ the coupling is instantaneously switched on, and the state is let to evolve unitarily, $|\psi(t)\rangle=e^{\text{i}H_\kappa\,t}|\psi_0\rangle$. The entanglement entropy of a local subset of the lattice evolves in a way similar to the one depicted in Fig.~\ref{fig:A-Illustration}. Instabilities play no role in this phenomenon. In fact the Hamiltonian (\ref{eq:Hquench}) is stable for $\kappa>-\Omega^2_0/4$, as it can be seen from Eq.~(\ref{eq:modek}). The relevant mechanism for the linear growth of the entanglement entropy in this quantum quench is not instabilities, but transport instead. The quench results in the local production of quasi-particles which travel at a finite speed. The phase of linear growth can be understood to be the result of the entanglement produced by the free propagation of entangled couples of quasi-particles, with the entanglement production rate determined by their propagation speed \cite{Calabrese:2005in,Calabrese:2004eu,cotler2016entanglement,alba2017entanglement}.

Interactions or coupling between many degrees of freedom, together with propagation of quasi-particles, play a key role in the phenomenon of entanglement growth in quantum quenches. On the other hand, the phenomenon studied in this paper relies on the existence of instabilities of some modes of a quantum system, as discussed in section~\ref{sec:applications-finite-dof} and 
\ref{sec:qft}. The difference between the two phenomena is easily illustrated by the case of bosonic systems and Gaussian states for which formula (\ref{eq:sentropy}) holds, \cite{braunstein2005quantum,ferraro2005gaussian,weedbrook2012gaussian,adesso:2014co}
\begin{equation}
	S_A=\sum_{i=1}^{N_e}\Big(  \frac{\nu_i+1}{2}\log\frac{\nu_i+1}{2}-\frac{\nu_i-1}{2}\log\frac{\nu_i-1}{2}\Big)\,.
	\label{eq:sentropy0}
\end{equation}
In the case of quantum quenches, the number of entangled pairs $N_e$ with fixed weight $\nu_i$ grows linearly in time until saturation. On the other hand, in the presence of instabilities, unstable modes have weight $\nu_i$ which grows exponentially in time until saturation, therefore leading to an entanglement growth of the form depicted in Fig.~\ref{fig:A-Illustration}. While the linear regime for quantum quenches can only be seen for a sufficiently large number of degrees of freedom in the system, linear growth due to instabilities can already occur for a system with two degrees of freedom and a single instability \cite{Hackl:2017ndi}. As a result, despite the intriguing similarity, the two phenomena are distinct.

\subsection{A conjecture on entanglement, chaos and thermalization times}
There is an intimate relationship between chaos, thermalization and entanglement \cite{polkovnikov2011colloquium,gogolin2016equilibration,d2016quantum,deutsch_91,srednicki_94,rigol_dunjko_08}. Here we discuss how semiclassical methods allow us to estimate the rate of growth of the entanglement entropy in the early phase of the thermalization process.

Let us consider a classical Hamiltonian system with linear phase space $(\mathbb{R}^{2N},\Omega)$ and a Hamiltonian $H$ which does not depend on time so that, as a result, the energy of the system is conserved. We assume also that the Hamiltonian is bounded from below and, at fixed energy, trajectories in phase space are bounded. This classical system displays a chaotic behavior if its Kolmogorov-Sinai entropy rate is non-vanishing, i.e. $h_{\mathrm{KS}}>0$  with $h_{\mathrm{KS}}$ defined by (\ref{eq:hKSdef}), \cite{gutzwiller2013chaos,haake2013quantum,reichl2013transition}. We are interested in the process of thermalization in the associated quantum system with Hamiltonian $H$. We argue that the Kolmogorov-Sinai entropy rate $h_{\mathrm{KS}}$ studied in this paper plays a central role in determining the relevant time scale in the process of quantum thermalization.\\

An isolated quantum system thermalizes when observables that probe only part of the system cannot distinguish a pure state from a thermal state. More precisely, let us consider a pure state and a thermal state with the same energy,
\begin{equation}
|\psi_t\rangle=e^{-\text{i}Ht}\,|\psi_0\rangle \qquad \text{and}\qquad \sigma=\frac{e^{-\beta H}}{Z}\,.
\label{eq:}
\end{equation}
The requirement that they have the same energy fixes the temperature $\beta^{-1}$ of the thermal state, i.e., $E=\langle\psi_t|H|\psi_t\rangle\,=\,\text{Tr}(H\sigma)$. Now we consider a subsystem $A$ and the subalgebra of bounded observables\footnote{Bounded observables have finite norm defined as $\|\mathcal{O}\|^2=\text{Tr}(\mathcal{O}^\dagger\mathcal{O})<\infty$.} $\mathcal{O}_A$ in $A$. We say that the subsystem $A$ thermalizes if all bounded observables $\mathcal{O}_A$ attain a thermal expectation value, i.e.,
\begin{equation}
\langle \psi_t|\mathcal{O}_A|\psi_t\rangle\,\longrightarrow\, \text{Tr}(\mathcal{O}_A\,\sigma)\,.
\label{eq:}
\end{equation}
This condition can be formulated in terms of entanglement between the subsystem $A$ and its complement $B$. Let us define the restricted states\footnote{Note that the operator $\tilde{H}_A$ is defined in terms of the restricted thermal state and in general it does not coincide with the restriction of the Hamiltonian $H$ to the subsystem $A$.} 
\begin{equation}
\rho_A(t)=\text{Tr}_B\big(|\psi_t\rangle\langle\psi_t|\big)\qquad \text{and} \qquad \sigma_A=\text{Tr}_B(\sigma)\,\equiv \frac{e^{-\beta \tilde{H}_A}}{Z_A}\,.
\label{eq:sigma}
\end{equation}
Thermalization in $A$ is a measure of how distinguishable is the restricted states $\rho_A(t)$ from the restricted thermal state $\sigma_A$. The relative entropy \cite{vedral2002role,ohya2004quantum},
\begin{equation}
S(\rho_A(t)\|\sigma_A)\equiv \text{Tr}_A(\rho_A\log\rho_A\;-\rho_A\log\sigma_A)\,,
\label{eq:}
\end{equation}
provides a measure of such distinguishability. In fact, using the inequalities $S(\rho\|\sigma)\geq \frac{1}{2}\|\rho-\sigma\|^2$ together with the Schwarz inequality $\|\sigma\|\geq \text{Tr}\big(\mathcal{O}\sigma\big)/\|\mathcal{O}\|$, one can prove the inequality
\begin{equation}
\frac{\Big(\langle \psi_t|\mathcal{O}_A|\psi_t\rangle\,-\, \text{Tr}(\mathcal{O}_A\,\sigma)\Big)^2}{2\,\|\mathcal{O}_A\|^2}\leq\;S(\rho_A(t)\|\sigma_A)\,,
\label{eq:}
\end{equation}
which holds for all bounded observables in $A$. Therefore, proving $S(\rho_A(t)\|\sigma_A)\to 0$ as $t\to\infty$ provides a proof of thermalization in $A$. Now, the relative entropy can be expressed in turn as the sum of two terms,
\begin{equation}
S(\rho_A(t)\|\sigma_A)\;=\; \Big(\,S_A^{\text{eq}}(E)-S_A(t)\,\Big)\;+\;\beta\,\Big(\langle \psi_t|\tilde{H}_A|\psi_t\rangle\,-\, \text{Tr}(\tilde{H}_A\,\sigma)\Big)\,.
\label{eq:}
\end{equation}
The first term is the difference between the equilibrium entropy $S_A^{\text{eq}}(E)=-\text{Tr}_A(\sigma_A\log\sigma_A)$ and the entanglement entropy $S_A(t)$ of the subsystem. The second term measures energy flow between the subsystem $A$ and its complement, as measured by the effective Hamiltonian $\tilde{H}_A$ of the subsystem defined in Eq.~(\ref{eq:sigma}). At equilibrium, both terms vanish independently. This paper and the following conjecture focus on the evolution of the first term, i.e., the growth and saturation of the entanglement entropy $S_A(t)$.\\

When a subsystem thermalizes, the entanglement entropy $S_A(t)$ approaches the equilibrium value $S_A^{\text{eq}}(E)$. The eigenstate thermalization hypothesis (ETH) \cite{deutsch_91,srednicki_94,rigol_dunjko_08} provides a sufficient condition for such subsystem thermalization to occur. In a chaotic quantum system with local interactions one observes that energy eigenstate, $H|E_n\rangle=E_n|E_n\rangle$, in the bulk of the energy spectrum have a non-trivial entanglement structure: their restriction to a local subsystem results in a thermal state of the form of Eq.~(\ref{eq:sigma}), i.e., $\text{Tr}_B(|E_n\rangle\langle E_n|)\approx \sigma_A(E_n)$. As a result, the restriction $\rho_A(t)$ of a pure state $|\psi_t\rangle=\sum_n e^{-\text{i}E_n t}c_n|E_n\rangle$ with support in a narrow band of energy $E$ is also well approximated by a thermal state when averaged over time, i.e. $\frac{1}{T}\int_0^T\rho_A(t)dt\;\approx \sigma_A(E)$ for large $T$. This condition is sufficient to prove thermalization in average, but it does not provide a time-scale for the thermalization process.\\

We propose a conjecture which complements previous arguments to the quantum thermalization of subsystems \cite{polkovnikov2011colloquium,gogolin2016equilibration,d2016quantum,deutsch_91,srednicki_94,rigol_dunjko_08}. The conjecture applies to semiclassical states and provides a time-scale for subsystem thermalization:
\begin{itemize}[leftmargin=*,noitemsep]
\item[] \emph{Given an initial state $|\psi_0\rangle=\sum_n c_n |E_n\rangle$ peaked on a classical configuration of energy $E=\langle\psi_0|H|\psi_0\rangle$ with $E$ large compared to the energy gap of the system, and a local subsystem $A$ such that its initial entanglement entropy is small compared to the thermal entropy at the same energy, $S_A(|\psi_0\rangle)\ll S_A^{\mathrm{eq}}(E)$, the time-evolution of the entropy $S_A(t)\equiv S_A(e^{-\mathrm{i}H t}|\psi_0\rangle)$ displays a linear phase $S_A(t)\sim \Lambda_A(E)\,t$ before saturating to the plateau at $S_A^{\mathrm{eq}}(E)$ as described in figure~\ref{fig:A-Illustration}. The rate $\Lambda_A(E)$ can be computed from the classical chaotic dynamics of the Hamiltonian $H$ on the energy-shell $E$. The rate is given by the subsystem exponent discussed in section~\ref{sec2:LambdaA} and, apart from its energy, it is largely independent of the initial state $|\psi_0\rangle$. In particular, the subsystem exponent sets the time-scale of subsystem thermalization, $\tau_\mathrm{eq}\sim S_A^{\mathrm{eq}}(E)/\Lambda_A(E)$.}
\end{itemize}
The conjecture is based on theorem~\ref{th:SA} presented in section~\ref{sec:growth}, together with semiclassical arguments. Let us consider a classical solution $(q^\text{cl}_i(t),p^\text{cl}_i(t))$ with energy $H(q^\text{cl}_i(t),p^\text{cl}_i(t))=E$. At the leading order in a semiclassical expansion, the evolution of a perturbation $\xi^a=(q^\text{cl}_i(t)+\delta q_i,\,p^\text{cl}_i(t)+\delta p_i)$ of the classical solution is governed by the perturbative Hamiltonian
\begin{equation}
H_{\text{pert}}(t)=\frac{1}{2}h_{ab}(t)\,\delta\xi^a\,\delta\xi^b
\label{eq:}
\end{equation}
where $\delta \xi^a=(\delta q_i,\delta p_i)$ and
\begin{equation}
h_{ab}(t)=\left.\frac{\partial^2 H}{\partial \xi^a\,\partial\xi^b}\right|_{\xi^a_{\text{cl}}(t)}\,.
\label{eq:}
\end{equation} 
The Lyapunov exponents of a non-perturbative chaotic system with Hamiltonian $H$ can be computed directly from the perturbative Hamiltonian $H_{\text{pert}}(t)$, which is quadratic time-dependent and for which our theorem~\ref{th:SA-generic} applies. In fact, because of ergodicity of a chaotic system, all trajectories on the same energy-shell $E$ (except a set of measure zero) have the same Lyapunov exponents.\footnote{For ergodic dynamics, time averages along an endless trajectory equal ensemble averages over the energy shell. Short periodic orbits may still retain their individual Lyapunov exponents, but they form a set of measure zero.} Moreover, under standard assumptions of regularity \cite{eckmann1985ergodic}, the Kolmogorov-Sinai rate $h_\mathrm{KS}(E)$ on the shell of energy $E$ is  given by Pesin's formula (\ref{eq:Pesin}) in terms of the positive Lyapunov exponents $\lambda_i(E)$. We consider now a symplectic subsystem $(A,\Omega_A)$ and define its subsystem exponent $\Lambda_A(E)$ as in section~\ref{sec2:LambdaA}. This is also the rate of growth of the entanglement entropy derived assuming a Gaussian state and a quadratic Hamiltonian in theorem~\ref{th:LambdaA}. The conjecture extends this result to a full non-quadratic system with bounded and chaotic motion, within the regime of validity of the semiclassical expansion. The inequality $\Lambda_A(E)\leq h_\mathrm{KS}(E)$ provides an upper bound on the rate of entanglement growth during the linear phase. Clearly, the linear phase ends when the semiclassical approximation breaks down, i.e. when the spread of the wavefunction is so large that higher-order terms in the expansion $H=E+H_{\text{pert}}(t)+\ldots$ cannot be neglected. An estimate of this time is provided by $\tau_\mathrm{eq}\sim S_A^{\mathrm{eq}}(E)/h_\mathrm{KS}(E)$ which measures the ratio between the accessible volume in phase space and the rate of growth of the phase space volume occupied by the perturbation.\\

The conjectured behavior of the entanglement entropy $S_A(t)$ depicted in figure~\ref{fig:A-Illustration} is expected to manifest itself only in the regime where the semiclassical approximation holds. This conjecture can be tested on a model system such as the one described by the Hamiltonian
\begin{equation}
H=\frac{1}{2}(p_x^2+p_y^2+p_z^2)+\frac{1}{2}(x^2y^2+x^2z^2+y^2z^2)\,.
\label{eq:}
\end{equation}
This is a well-studied model which appears in the analysis of the homogeneous sector of Yang-Mills gauge theory \cite{biro1994chaos,martens1989classical}. Its Lyapunov exponents are known to scale with the energy as $\lambda_i(E)\sim E^\frac{1}{4}$ and its equilibrium entropy, estimated as the $\log$ of the phase space volume at fixed energy, scales as $S^{\mathrm{eq}}(E)\sim \log E$. As a result, for a semiclassical initial state of energy $E$ we expect our conjecture to apply: the entanglement entropy of a subsystems such as $(x,p_x)$ is expected to initially grow linearly with a rate $\sim E^\frac{1}{4}$ and then saturate in a time $\tau_{\text{eq}}\sim E^{-1/4} \log E$. This behavior can in principle be tested via numerical investigations. The numerical analysis involves the unitary evolution of a pure state under a chaotic quantum Hamiltonian, which is beyond the scope of the present paper. A preliminary numerical analysis of the growth of the entanglement entropy in interacting systems prepared in a semiclassical state can be found in \cite{Hackl:2017ndi}.

We note that the conjecture is expected to apply only to initial states which are semiclassical, i.e. states with average energy much larger than the energy gap, small spread in energy and, in general, small spread around a point in phase space.
On the other hand, when the energy $E$ of the initial state is comparable to the energy gap of the Hamiltonian, classical orbits of that energy have an action comparable to $\hbar$ and there is no reason to expect that they provide a useful tool for predicting the behavior of the entanglement entropy in the linear regime of figure~\ref{fig:A-Illustration}. In fact, recent results from quantum field theories with a gravity duals \cite{Maldacena:2015waa} show that\,---\,at low energy\,---\,the rate of growth of the entanglement entropy is bounded from above by the energy of the subsystem divided by $\hbar$ and therefore deviates from the semiclassical prediction \cite{Berenstein:2015yxu}.

\subsection{Summary}
We studied the relationship between entropy production in classical dynamical systems and the growth of the entanglement entropy in their quantum analogue in the semiclassical regime. Most importantly, we found that in both cases the production rates are given by a sum over Lyapunov exponents $\lambda_i$ characterizing stable and unstable phase space directions. For classical systems, there is a standard notion of rate, the Kolmogorov-Sinai entropy rate
\begin{equation}
	h_{\mathrm{KS}}=\sum_{\lambda_i>0}\lambda_i
\end{equation}
given by the sum over all positive Lyapunov exponents. We have shown that for the associated quantum system and a subsystem $A$, the production rate of the entanglement entropy $S_A(t)\sim \Lambda_A\,t$ is given by a subsystem exponent $\Lambda_A$. We have shown that $\Lambda_A\leq h_{\mathrm{KS}}$ and found that this inequality is saturated for sufficiently large subsystems. Moreover we found that the rate $\Lambda_A$ is independent of the initial state of the system and\,---\,except for a set of measure zero of subsystems\,---\,it depends on the choice of subsystem $A$ only via its classical dimension $N_A$, i.e.
\begin{align}
\Lambda_{A\,\mathrm{generic}}=\sum^{2N_A}_{i=1}\lambda_{i}\,,
\end{align}
where $\lambda_i$ are the $2N_A$ largest Lyapunov exponents of the system. Our rigorous derivation of this result is based on the assumption of unstable quadratic Hamiltonian and Gaussian initial state. The derivation takes into account the case of time-dependent Hamiltonians with Floquet instabilities.

The derivation of the main theorem proving $S_A(t)\sim \Lambda_A\,t$ consists of three parts. First, the subsystem exponent $\Lambda_A$ is introduced at the classical level as a generalization of Lyapunov exponents and defined to encode the exponential rate of growth of the volume of a symplectic cube in a subsystem under Hamiltonian evolution. Second, the time evolution of a Gaussian initial state through an unstable quadratic Hamiltonian is conveniently encoded in terms of complex structures or equivalently phase space metrics and their classical Hamiltonian flow. Third, the evolution of entanglement entropy is shown to be asymptotically the same as the one of the Renyi entropy which can then be shown to grow with the rate of the subsystem exponent $\Lambda_A$. We interpret the exponent $\Lambda_A$ as a quantum analogue of the Kolmogorov-Sinai entropy rate of a given subsystem $A$. 

The predicted linear growth of the entanglement entropy shows up in a wide range of physical systems such as unstable quadratic potentials, periodic quantum quenches in many-body quantum systems and instabilities in quantum field theory models. We presented three examples of the latter where entanglement is produced through different mechanisms, namely unstable modes due to a symmetry-breaking instability, parametric resonance in models of post-inflationary reheating, and cosmological perturbations in an inflationary spacetime.

We believe that our results are also relevant in the context of thermalization of isolated quantum systems. A subsystem of a chaotic quantum system is expected to thermalize with equilibrium entropy $S_{\mathrm{eq}}(E)$ determined by the average energy $E$ of the initial state. In the semiclassical regime we conjecture that the time-scale of this equilibration process is $\tau_\mathrm{eq}\sim S_{\mathrm{eq}}(E)/\Lambda_A(E)$ where $\Lambda_A(E)$ is the subsystem exponent of the energy-shell $E$.\\

\medskip

\acknowledgments
We thank Abhay Ashtekar for extensive discussions on the use of complex structures in quantum field theory, and Alejandro Satz for fruitful comments on the use of smearing functions. We thank also Marcos Rigol and Ranjan Modak for discussions on extensions of the presented results to the case of non-Gaussian states. EB thanks Berndt M\"uller for inspiring conversations on the Kolmogorov-Sinai entropy rate. EB acknowledges also extensive discussions with Renaud Parentani and Bei-Lok Hu which took place during the $22^{\mathrm{nd}}$ Peyresq Physics workshop. LH thanks Pavlo Bulanchuk for a suggestion leading to the geometric representation of the R\'{e}nyi entropy. We thank Hal Haggard, Carlo Rovelli and Matteo Smerlak for multiple discussions and feedback at various stages of this project. The work of EB is supported by the NSF grant PHY-1404204. LH is supported by a Frymoyer fellowship. NY acknowledges support from CNPq, Brazil and  from the NSF grant PHY-1505411. This research was supported in part by the Perimeter Institute for Theoretical Physics.

\newpage
\appendix
\section{Dynamical systems and Lyapunov exponents}\label{app}
We summarize relevant properties of Lyapunov exponents in the context of Hamiltonian systems. In particular, we make precise the notion of \emph{regular Lyapunov system}.

\subsection{Linear Hamiltonian systems}\label{app:Hamiltonian}
We consider a $2N$-dimensional linear phase space $V$ with symplectic form $\omega$. A time-dependent Hamiltonian $H$ is a smooth map
\begin{equation}
H: V\times\mathbb{R}\to \mathbb{R}: (\xi,t)\to H(\xi,t)\,.
\end{equation}
The equations of motion are given by
\begin{equation}
\dot{\xi}^a(t)=\Omega^{ab}(dH)_b(t)\,,
\end{equation}
where $\Omega^{ab}$ satisfies $\omega_{ac}\Omega^{bc}=\delta^b_a$ and $(dH)_b(t)$ is the gradient of $H$ at time $t$. The solution of these equations can be conveniently described by a flow
\begin{equation}
\Phi_t: V\to V: \xi\to \Phi_t\,\xi\,.
\end{equation}
This map is a diffeomorphism that preserves the symplectic form $\omega$, namely the push-forward satisfies $(\Phi_t)_*\omega=\omega$. For a given point $\xi_0\in V$, the push-forward $(\Phi_t)_*$ maps a tangent vector $\delta\xi\in T_{\xi_0} V$ to the tangent vector $(\Phi_t)_*\delta\xi\in T_{\Phi_t(\xi_0)}V$. Due to the linearity of $V$, we can identify the tangent spaces at all points with $V$ itself. Formally, we have an isomorphism $\phi_{\xi}: V\to T_{\xi}V$ that maps $v\in V$ to the tangent vector $\phi_{\xi}v\in T_\xi V$ that acts on a function $f: V\to\mathbb{R}$ as $\phi_{\xi}v(f)=\frac{d}{dt}f(\xi+tv)$. Using $\phi_\xi$, we can to define the linear map $M_{\xi_0}(t): V\to V$
\begin{equation}
M_{\xi_0}(t)=\phi_{\Phi_t(\xi_0)}^{-1}\circ(\Phi_t)_*\circ\phi_{\xi_0}\,,
\end{equation}
that corresponds to the above push-forward after we identify $T_{\xi_0}V$ and $T_{(\Phi_t)\xi_0}V$ with $V$.\par
In the special case, where the Hamiltonian $H$ is given by a linear quadratic function
\begin{equation}
H(\xi,t)=f_a(t)\xi^a+\frac{1}{2}h_{ab}(t)\xi^a\xi^b
\end{equation}
for every $t$, the Hamiltonian flow is given by an inhomogeneous symplectic transformation $(M(t),\eta(t))$ via
\begin{equation}
\Phi_t\xi_0=M(t)\xi_0+\eta(t)\,,
\end{equation}
whose differential is given by $M_{\xi_0}(t)=M(t)$, independent of $\xi_0$. This follows from the fact that a quadratic Hamiltonian gives rise to linear and homogeneous equations of motion. The symplectic group element $M(t)$ is formally given by the time-ordered exponential
\begin{equation}
M(t)^a{}_b=\mathcal{T}\exp\left(\int^t_0 dt' K(t')^a{}_b\right)\quad \text{with}\quad K(t)^a{}_b=\Omega^{ac}h(t)_{cb}\,,
\end{equation}
where the generator $K(t)$ is an element of the symplectic Lie algebra $\mathrm{sp}(2N)$.\par
In the general case, where $H$ is not quadratic, we can still find the time-dependent generator $K_{\xi_0}(t)^a{}_b=\Omega^{ac}h_{\xi_0}(t)_{cb}$. In this case, however, the generator also depends on the initial $\xi_0$. To find $h_{\xi_0}(t)_{cb}$, we just need to Taylor expand the Hamiltonian $H(t)$ along the trajectory $\xi(t)=\Phi_t\xi_0$ which amounts to finding its Hessian
\begin{equation}
h_{\xi_0}(t)_{ab}=\partial_a\partial_b\,H(t)\big\rvert_{\Phi_t\xi_0}\,.
\end{equation}
At this point, we understand that the difference between the special (quadratic) and general case (arbitrary Hamiltonian) in regards of the linear map $M_{\xi_0}(t)$ are the following:
\begin{itemize}
	\item If the Hamiltonian is quadratic or affine quadratic, the linear map $M(t)$ describing the push-forward of the Hamiltonian flow $\Phi_t$ is independent of the starting point $\xi_0$ and completely characterized by the quadratic part $h(t)_{ab}$ of $H(t)$.
	\item For a more general Hamiltonian, we can still compute its quadratic part $h_{\xi_0}(t)_{ab}$ as Hessian of $H(t)$ along the trajectory $\xi(t)$. This means $h_{\xi_0}(t)_{ab}$ depends on the initial condition $\xi_0$ and the corresponding solution $\xi(t)$ with $\xi(t)=\xi_0$. In particular, the quadratic map $M_{\xi_0}(t)$ will differ for different initial conditions $\xi_0$.
\end{itemize}
The linear symplectic map $M_{\xi_0}(t)$ contains all the information about how two sufficiently close trajectories converge or diverge. This behavior will be captured in the so called Lyapunov exponents.\par
In order to define Lyapunov exponents, we need to equip phase space $V$ with a positive definite metric $g_{ab}$ that gives rise to a norm $\lVert\delta\xi\rVert=\sqrt{g_{ab}\delta\xi^a\delta\xi^b}$. Equivalently, we can use the inverse metric $G^{ab}$ to define the norm $\lVert\ell\rVert=\sqrt{G^{ab}\ell_a\ell_b}$ on the dual phase space $V^*$. We will show that Lyapunov exponents are actually independent of the specific choice of a positive metric. In order to show this, it is useful to have the following theorem at hand.
\begin{proposition}
	Given a finite dimensional, real vector space $V$ and two distinct positive metrics $g$ and $\tilde{g}$, we can compute the following two values
	\begin{equation}
	a:=\min_{\lVert v\rVert_g=1}\lVert v\rVert_{\tilde{g}}>0\,,\quad b:=\max_{\lVert v\rVert_g=1}\lVert v\rVert_{\tilde{g}}>0\,,
	\end{equation}
	which allow us to relate norms and angles measured by the different metrics:
	\begin{itemize}
		\item \textbf{Norm inequality}\\
		Given a vector $v\in V$, its norm $\lVert v\rVert_{\tilde{g}}$ with respect to $\tilde{g}$ is related to $\lVert v\rVert_g$ via:
		\begin{equation}
		a\lVert v\rVert_{g}\leq\lVert v\rVert_{\tilde{g}}\leq b\lVert v\rVert_{g}\,.\label{eq:norm-inq}
		\end{equation}
		\item \textbf{Angle inequality}\\
		Given an angle $\tilde{\psi}$ between two vectors measured with respect to $\tilde{g}$, it is related to the angle $\psi$ measured with respect to $g$ via the following inequality:
		\begin{equation}
		\frac{1-(b/a)^2\tan^2{(\psi/2)}}{1+(b/a)^2\tan^2{(\psi/2)}}\leq\cos\tilde\psi\leq\frac{1-(a/b)^2\tan^2{(\psi/2)}}{1+(a/b)^2\tan^2{(\psi/2)}}\,.\label{eq:angle-inq}
		\end{equation}
		This inequality can be simplified to the slightly weaker version given by:
		\begin{equation}
		\frac{a\psi}{b}\leq \tilde{\psi}\leq \frac{b\psi}{a}\,.\label{eq:angle-inq-simple}
		\end{equation}
		\item \textbf{Volume inequality}\\
		Given the $d$-volume $\mathrm{Vol}_{\tilde{g}}(\mathcal{V}_A)$ of some region $\mathcal{V}_A$ in an arbitrary $d$-dimensional subspace $A\subset V$ measured by the metric $\tilde{g}$, it is related to the $d$-volume $\mathrm{Vol}_{g}(\mathcal{V}_A)$ measured by $g$ via the following inequality:
		\begin{equation}
		a^{d}\,\mathrm{Vol}_{\tilde{g}}(\mathcal{V}_A)\leq\mathrm{Vol}_{\tilde{g}}(\mathcal{V}_A)\leq b^{d}\, \mathrm{Vol}_{g}(\mathcal{V}_A)\,.
		\label{eq:volume-inq}
		\end{equation}
	\end{itemize}
If consider the same equations for the dual phase space $V^*$ with the replacements $g\to G$ and $\tilde{g}\to \tilde{G}$, all inequalities hold if we replace $a\to 1/b$ and $b\to1/b$.
\end{proposition}
\begin{proof}
	Let us prove the different inequalities:
	\begin{itemize}
		\item Norm inequality\\
		Let us take two different norms induced by the two positive metrics $g$ and $\tilde{g}$. Over a finite dimensional vector space $V$ the set $S=\{v\in V|\lVert v\rVert_g=1\}$ is compact. This means that the continuous function $\lVert v\rVert_{\tilde{g}}$ will take a minimal and maximum value on $S$:
		\begin{equation}
		a:=\min_{v\in S}\lVert v\rVert_{\tilde{g}}>0\,,\quad b:=\max_{v\in S}\lVert v\rVert_{\tilde{g}}>0\,.
		\end{equation}
		Linearity of the induced norm implies than the inequality that we wanted to prove:
		\begin{equation}
		a\lVert v\rVert_{g}\leq\lVert v\rVert_{\tilde{g}}\leq b\lVert v\rVert_{g}\quad\text{for all}\quad v\in V\,.
		\end{equation}
		\item Angle inequality\\
		Let us choose a two-dimensional plane $P\subset V$. On this plane, we have the restricted metrics $g|_P$ and $\tilde{g}|_P$. The two are related by a linear map $D: P\to P$ with
		\begin{equation}
		(\tilde{g}|_P)_{ab}=D^c{}_aD^d{}_b\,(g|_P)_{cd}\,,
		\end{equation}
		where $D$ is not unique. We can always choose it to be diagonalizable with ordered eigenvalues $d_i$ and eigenvectors $e_i$. At this point, we can identify the inner product with respect to $\tilde{g}$ as the one with respect to $g$ after having acted with $D$ on the vectors. This implies $a\leq d_i\leq b$ to not violate the norm inequality. Let us choose two unit vectors $v,w\in P$ that form an angle $\psi$ with respect to $g$ and whose angle bisector lies at an angle of $\phi$ to $e_1$:
		\begin{align}
		v&=\cos(\phi+\psi/2)e_1+\sin(\phi+\psi/2)e_2\\
		w&=\cos(\phi-\psi/2)e_1+\sin(\phi-\psi/2)e_2
		\end{align}
		We can compute the deformed angle $\tilde{\psi}(\psi,\phi)$ from the deformed vectors
		\begin{align}
		Dv&=d_1\cos(\phi+\psi/2)e_1+d_2\sin(\phi+\psi/2)e_2\\
		Dw&=d_1\cos(\phi-\psi/2)e_1+d_2\sin(\phi-\psi/2)e_2\,,
		\end{align}
		by using the arctangent rules with respect to $g$ based on $\langle v,w\rangle_{\tilde{g}}=\langle Dv, Dw\rangle_g$:
		\begin{equation}
		\tilde{\psi}(\psi,\phi)=\arctan\left(d_2/d_1\,\tan(\phi+\psi/2)\right)-\arctan\left(d_2/d_1\,\tan(\phi-\psi/2)\right)\,.
		\end{equation}
		By taking the derivative with respect to $\phi$, we can find the minimum and maximum of this function for fixed $\psi$. The minimum is at $\phi=0$ and the maximum at $\phi=\pi/2$ (recall that we chose $d_2>d_1$). Evaluating $\tilde{\psi}(\psi,\phi)$ at these values leads to the inequality
		\begin{equation}
		\frac{1-(d_2/d_1)^2\tan^2{(\psi/2)}}{1+(d_2/d_1)^2\tan^2{(\psi/2)}}\leq\cos\tilde\psi\leq\frac{1-(d_1/d_2)^2\tan^2{(\psi/2)}}{1+(d_1/d_2)^2\tan^2{(\psi/2)}}\,.
		\end{equation}
		This interval becomes maximal when $d_1/d_2$ is as small as possible, but for a given metric $\tilde{g}$, we have $d_1/d_2\in[a/b,1]$ for any plane $P\in V$. Thus, we find the following bound
		\begin{equation}
		\frac{1-(b/a)^2\tan^2{(\psi/2)}}{1+(b/a)^2\tan^2{(\psi/2)}}\leq\cos\tilde\psi\leq\frac{1-(a/b)^2\tan^2{(\psi/2)}}{1+(a/b)^2\tan^2{(\psi/2)}}\,,
		\end{equation}
		which holds in general. For small angles, we can Taylor expand this and find
		\begin{equation}
		\frac{a\psi}{b}\leq \tilde{\psi}\leq \frac{b\psi}{a}\,.
		\end{equation}
		\item Volume inequality:\\
		If we use a metric to measure the volume of some region $\mathcal{V}_A\subset A$, we use the Lebesgue measure in $\mathbb{R}^d$ by identifying with $A$ with $\mathbb{R}^d$ by choosing an orthonormal basis in $A$. For two metrics $g$ and $\tilde{g}$, linearity implies that there exists a unique number $c$, such that $\mathrm{Vol}_{\tilde{g}}(\mathcal{V}_A)=c\mathrm{Vol}_{g}(\mathcal{V}_A)$ holds for any region $\mathcal{V}_A\subset A$. In order to bound this constant, we can use the norm inequality to show that the $d$-dimensional unit ball $B^d_{\tilde{g}}=\left\{v\in A\,\text{with}\,\lVert v\rVert_{\tilde{g}}\leq 1\right\}$ contains the ball $B^d_{g}(a)=\left\{v\in A\,\text{with}\,\lVert v\rVert_{\tilde{g}}\leq a\right\}$ and is contained in the ball $B^d_{g}(b)=\left\{v\in A\,\text{with}\,\lVert v\rVert_{\tilde{g}}\leq b\right\}$. This implies $a^d\leq c\leq b^d$ which leads to the volume inequality
		\begin{equation}
		a^{d}\,\mathrm{Vol}_{g}(\mathcal{V}_A)\leq\mathrm{Vol}_{\tilde{g}}(\mathcal{V}_A)\leq b^{d}\, \mathrm{Vol}_{g}(\mathcal{V}_A)\,,
		\end{equation}
		we wanted to prove.
	\end{itemize}
	If we replace $V\to V^*$ and accordingly $g\to G$ and $\tilde{g}\to\tilde{G}$, we can run exactly the same arguments, but we need to compute
	\begin{equation}
	1/b=\min_{\lVert v\rVert_G=1}\lVert v\rVert_{\tilde{G}}\,,\quad 1/a=\max_{\lVert v\rVert_G=1}\lVert v\rVert_{\tilde{G}}>0\,.
	\end{equation}
	This follows from the fact that $a$ and $b$ are the smallest and largest eigenvalue of the linear map $(G\tilde{g})^a{}_b=G^{ac}\tilde{g}_{cb}$. Under above replacement, we need to consider its inverse map $(g\tilde{G})_a{}^b=g_{ac}\tilde{G}^{cb}$ whose smallest and largest eigenvalues are therefore $1/b$ and $1/a$, respectively.
\end{proof}

\subsection{Lyapunov exponents}\label{app:Lyapunov}
In what follows, we restrict ourselves to quadratic systems where $M(t)$ is independent of the initial value $\xi_0$. This generalizes to non-quadratic systems by replacing $M(t)$ by $M_{\xi_0}(t)$. In this case, Lyapunov exponents and vectors depend on the specific trajectory $\xi(t)=\Phi_t(\xi_0)$.
\begin{definition}[Lyapunov exponent]
	Given a linear Hamiltonian flow $M(t):V\to V$ and a vector $\delta\xi\in V$, we define the \underline{Lyapunov exponent} $\lambda_{\delta\xi}$ as the limit
	\begin{equation}
	\lambda_{\delta\xi}=\lim_{t\to\infty}\frac{1}{t}\log\frac{\lVert M(t)\,\delta\xi\rVert_g}{\lVert\delta\xi\rVert_g}\,,
	\end{equation}
	provided it exists. This definition is independent of the positive definite metric $g$ that induces the norm $\lVert\cdot\rVert$. Analogously, we define the Lyapunov exponent of a dual vector $\ell\in V^*$ as the limit
	\begin{equation}
	\lambda_{\ell}=\lim_{t\to\infty}\frac{1}{t}\log\frac{\lVert M\tra(t)\,\ell\rVert_G}{\lVert\ell\rVert_G}\,,
	\end{equation}
	provided it exists. Here, the definition is independent of the inverse metric $G$.
\end{definition}
\begin{proof}
	We need to prove the independence of this definition from the chosen norm $\lVert\cdot\rVert_g$ induced by some metric $g$. We can use the norm inequality (\ref{eq:norm-inq}) to show $\lVert M(t)\,\delta\xi\lVert_{\tilde{g}}=c_t\lVert M(t)\,\delta\xi\lVert_{g}$ with factor $c_t\in[a,b]$. Let $\lambda_p$ be the Lyapunov exponent of $\delta\xi\in V$ with respect to the norm $\lVert\cdot\rVert_g$. We can now compute
	\begin{equation}
	\tilde\lambda_{\delta\xi}=\lim_{t\to\infty}\frac{1}{t}\log\frac{\lVert M(t)\, \delta\xi\rVert_{\tilde{g}}}{\lVert\delta\xi\rVert_{\tilde{g}}}=\lim_{t\to\infty}\frac{1}{t}\log\frac{\lVert M(t)\, \delta\xi\rVert_{g}}{\lVert \delta\xi\rVert_{g}}+\underbrace{\lim_{t\to\infty}\frac{c_t}{t}\frac{\lVert \delta\xi\rVert_{g}}{\lVert \delta\xi\rVert_{\tilde{g}}}}_{=0}=\lambda_{\delta\xi}\,,
	\end{equation}
	where the second term vanishes because $c_t$ is a bounded function. For dual Lyapunov vectors $\ell\in V^*$, we can use the same arguments where only our bounds for $c_t$ change to $c_t\in [1/b,1/a]$.
\end{proof}
To characterize the Lyapunov exponents of all vectors in a $2N$-dimensional vector space, it is sufficient to select a representative sample of $2N$ vectors. Such a basis is called Lyapunov basis and is defined as follows.

\begin{definition}[Lyapunov basis and spectrum]
Given the linear flow $M(t)$, we define the limit matrix 
\begin{equation}
L_a{}^b\equiv\lim_{t\to\infty}\frac{1}{2t}\log \Big(g_{ac}\,M^{c}{}_d\,G^{de}\,M^{b}{}_e \Big)\,,
\end{equation}
provided it exists. We then define a complete set of eigenvectors as \underline{Lyapunov basis} $\mathcal{D}_L=(\ell^1,\dots,\ell^{2N})$ if it is chosen as Darboux basis, such that $\{\ell^i,\ell^{2N-i+1}\}=1$ for $i=1,\dots,N$ are the only non-trivial Poisson brackets and such that the associated Lyapunov exponents $\lambda_i:=\lambda_{\ell_i}$ are ordered with $\lambda_i\geq\lambda_{i+1}$. The set $(\lambda_1,\dots,\lambda_{2N})$ is called \underline{Lyapunov spectrum}.
\end{definition}
\begin{proof}
The construction of a Lyapunov basis as eigenvectors of the limiting matrix $L$ is an important part of Oseledets multiplicative ergodic theorem. A comprehensible proof with further details can be found \cite{eckmann1985ergodic}. The fact that the eigenvectors can always be chosen to form a Darboux basis follows from the fact that $L_a{}^b$ is an element of the symplectic algebra $\mathrm{sp}(2N,\mathbb{R})$.
\end{proof}

When restricting to a subsystem $A\subset V$, it is natural to ask what is the Lyapunov spectrum of the subsystem.
\begin{definition}[Subsystem Lyapunov basis and spectrum]
Given the linear flow $M(t)$ and a symplectic subspace $A\subset V$, we define the \underline{subsystem Lyapunov basis} of $A$ as the $2N_A$ vectors $(\ell^1_A,\dots,\ell^{2N_A}_A)$ with associated \underline{subsystem Lyapunov spectrum}
\begin{equation}
	\lambda_1^A\geq\dots\geq\lambda^A_{2N_A}\,,
\end{equation}
such that a linear observable $\theta\in A^*$ with $\theta=\sum^{2N_A}_{i=1}c_i\ell^{i}_A$ has Lyapunov exponent $\lambda^A_j$ where $j\geq$ is the smallest number, such that $T_j\neq 0$.
\end{definition}
The subsystem Lyapunov spectrum does in general not consist of conjugated pairs $(\lambda,-\lambda)$. Moreover, it is important to emphasize that the Lyapunov spectrum of $A$ is defined as those Lyapunov exponents of linear observables $\theta\in A^*$, rather than of perturbations $\delta\xi\in A$, because the two are not the same.\\

The following proposition explains in detail how one can compute the subsystem Lyapunov basis and spectrum when the Lyapunov basis and spectrum of the full system is known.
\begin{proposition}\label{app:prop2}
Given a the linear flow $M(t)$ with Lyapunov basis $\mathcal{D}_L$ and Lyapunov spectrum $(\lambda_1,\dots,\lambda_{2N})$, we can compute the subsystem Lyapunov basis and spectrum of a subsystem $A\subset V$ using the following three steps:
\begin{enumerate}
	\item Choose a Darboux basis $\mathcal{D}_A=(\theta^1,\dots,\theta^{2N_A})$ of the symplectic subspace $A^*\subset V^*$.
	\item Compute the unique transformation matrix $T$ that expresses $\mathcal{D}_A$ in terms of the Lyapunov basis $\mathcal{D}_{L}=(\ell^1,\dots,\ell^{2N})$:
	\begin{equation}
	\left(\begin{array}{c}
	\theta^1\\
	\vdots\\
	\theta^{2N_A}
	\end{array}\right)=\left(\begin{array}{ccc}
	\smash{\fbox{\color{black}\rule[-37pt]{0pt}{1pt}$\,\,T^1_1\,\,$}} & \cdots & \smash{\fbox{\color{black}\rule[-37pt]{0pt}{1pt}$\,T_{2N}^1$}}\\
	\vdots & \ddots & \vdots\\
	\myunderbrace{\mystrut{1.5ex}T_1^{2N_A}}{\vec{t}_1} & \cdots & \myunderbrace{\mystrut{1.5ex}T_{2N}^{2N_A}}{\vec{t}_{2N}}
	\end{array}\right)\left(\begin{array}{c}
	\ell^1\\
	\vdots\\
	\ell^{2N}
	\end{array}\right)\color{white}{\begin{array}{c}
		.\\ \\ \\ \\ ,
		\end{array}}\color{black}
	\label{eq:Tv2}
	\end{equation}
	\vspace{.5em}
	We refer to the $2N$ columns of $\,T$ as $\vec{t}_i$.
	\item Find the first $2N_A$ linearly independent \footnote{Here we mean that $\vec{t}_i$ cannot be expressed as a linear combination of the vectors $(\vec{t}_1,\cdots,\vec{t}_{i-1})$ standing to the left in the matrix $T$.} columns $\vec{t}_i$ of $T$ which we can label by $\vec{t}_{i_k}$ with $k$ ranging from $1$ to $2N_A$. The result is a map $k\mapsto i_k \in (1,\ldots,2N)$ with $i_{k+1}>i_k$.
\end{enumerate}
The subsystem Lyapunov spectrum is given by $(\lambda_1^A,\dots,\lambda^A_{2N_A})$ with $\lambda_k^A=\lambda_{i_k}$ and the subsystem Lyapunov basis is given by $(\ell^1_A,\dots,\ell^{2N_A}_A)$ with
\begin{equation}
	\ell^k_A=(U^{-1}\theta)^k\,,
\end{equation}
where $U=\left(\vec{t}_{i_1},\dots,\vec{t}_{i_{2N_A}}\right)$ is the invertible $2N_A\times 2N_A$ matrix consisting of the columns $\vec{t}_{i_k}$.
\end{proposition}
\begin{proof}
	The rectangular matrix $T$ in (\ref{eq:Tv2}) allows us to express the elements of the Darboux basis $\mathcal{D}_A$ of the subsystem in terms of the Lyapunov basis, $\theta^r=\sum^{2N}T^r_i \ell^i$. Denoting the columns of $T$ by $\vec{t}_i$ we can select the first $2N_A$ linearly independent columns in the ordered set $(\vec{t}_1,\dots,\vec{t}_{2N})$. We label them $\vec{t}_{i_k}$ and organize them in the $2N_A\times 2N_A$ square matrix $U$,
	\begin{equation}
	U=\left(\begin{array}{c|c|c}
	\vec{t}_{i_1} & \dots & \vec{t}_{i_{2N_A}}
	\end{array}\right)\,.
	\end{equation}
	Due to their linear independence, the inverse $U^{-1}$ exists and turns $T$ into an upper triangular matrix $\tilde{T}$ of the form
	\begin{equation}
	\tilde{T}=U^{-1}T=\left(\begin{array}{lll}
	0 &\cdots \, 0\;  1 \,*\, * \cdots\cdots\cdots\cdots\cdots\cdots&*\\
	0 &\cdots \cdots\cdots \, 0\;  1 \,*\, *\cdots\cdots\cdots\cdots&*\\
	\,\vdots &\qquad\qquad \vdots\qquad\qquad \vdots\qquad\qquad&\vdots\\
	0 &\cdots \cdots\cdots\cdots \cdots \cdots 0\;  1 \,*\, *\;\cdots&*
	\end{array}\right)\,,
	\end{equation}
	where the $*$ represents an unspecified value. Acting with $U^{-1}$ on the left and the right-hand side of (\ref{eq:Tv2}) and acting on $\theta^k$, we find
	\begin{equation}
	\ell_A^k:=(U\theta)^k=\ell^{i_k}+\sum^{2N}_{j>i_k}\tilde{T}^k_j\, \ell^j\,,
	\end{equation}
	where $\ell_A^k=(U^{-1}\theta)^k=\sum^{2N_A}_{j=1}U^k_j\theta^j$. Clearly, the vectors $\ell_A^k$ satisfy 
	\begin{equation}
	\lim_{t\to\infty}\frac{1}{t}\log \lVert M\tra(t)\ell_A^k\rVert/\lVert\ell_A^k\rVert\,=\,\lambda_{i_k}\,.
	\label{eq:}
	\end{equation}
Given an arbitrary vector $\theta=\sum^{2N_A}_{i=1}c_i\ell^i_A$, its Lyapunov exponent is clearly given by the $\lambda^A_k=\lambda_{i_k}$ where $k$ is the smalles $i\geq 1$, for which $c_i$ is non-zero.
\end{proof}
In our geometric representations of the R\'{e}nyi entropy, we are interested in how the volume of some initial region changes under the Hamiltonian flow $M\tra(t)$. Due to the linearity of $M\tra(t)$, we can restrict ourselves to studying the time-dependent volume of parallelepipeds laying in some subspace $A\subset V$. The evolution will in general evolve this parallelepiped out of $A$.
\begin{definition}[Subsystem exponent]
	Given the linear flow $M(t)$ and a symplectic subspace $A\subset V$ of dimension $2N_A$, we can define the \underline{subsystem exponent} as the limit
	\begin{equation}
	\Lambda_A=\lim_{t\to\infty}\frac{1}{t}\log\frac{\mathrm{Vol}_G\left(M\tra(t)\mathcal{V}_A\right)}{\mathrm{Vol}_G\left(\mathcal{V}_A\right)}\,,
	\end{equation}
	provided it exists. The set $\mathcal{V}_A\subset A$ is an arbitrary parallelepiped spanning all dimensions of $A$. This definition is independent of the metric that one uses to measure the volume and independent of the choice of parallelepiped $\mathcal{V}_A$.
\end{definition}
\begin{proof}
We need to prove the independence of this definition from the choice of positive definite metric $G$. We can use the volume inequality (\ref{eq:volume-inq}) which ensures that for a different metric $\tilde{G}$, we have $\mathrm{Vol}_{\tilde{G}}\left(M\tra(t)\mathcal{V}_A\right)=c_t\mathrm{Vol}_G\left(M\tra(t)\mathcal{V}_A\right)$ with $c_t\in [(1/b)^{2N_A},(1/a)^{2N_A}]$. We compute
\begin{equation}
	\tilde{\Lambda}_A=\lim_{t\to\infty}\frac{1}{t}\log\frac{\mathrm{Vol}_{\tilde{G}} \left(M\tra(t)\,\mathcal{V}_A\right)}{\mathrm{Vol}_{\tilde{G}}\left(\mathcal{V}_A\right)}=\lim_{t\to\infty}\frac{1}{t}\log\frac{\mathrm{Vol}_{G} \left(M\tra(t)\,\mathcal{V}_A\right)}{\mathrm{Vol}_{G}\left(\mathcal{V}_A\right)}+\underbrace{\lim_{t\to\infty}\frac{c_t}{t}\log\frac{\mathrm{Vol}_{G} \left(\mathcal{V}_A\right)}{\mathrm{Vol}_{\tilde{G}}\left(\mathcal{V}_A\right)}}_{=0}=\Lambda_A\,,
\end{equation}
where the second term vanishes because $c_t$ is a bounded function.
\end{proof}

\subsection{Regular Hamiltonian systems}\label{app:regular}
The central theorem of this paper connects quantum mechanical entanglement with the classical notion of Lyapunov exponents. In order to avoid technical complications, we introduce the class of \emph{regular Hamiltonian Lyapunov systems}. Most standard Hamiltonian systems that one studies in classical or quantum physics with a finite number of bosonic degrees of freedom fall into this class.
\begin{definition}
	A \underline{regular Hamiltonian Lyapunov system} consists of a finite dimensional phase space and a (possibly time-dependent) Hamiltonian $H(t): V\to\mathbb{R}$ with linerized flow $M_{\xi_0}(t)$, such that the following two conditions are satisfied:
	\begin{enumerate}
		\item[(i)] All Lyapunov exponents are well defined. This means that for an arbitrary initial condition $\xi_0$ as well as for every initial separation $\delta\xi\in T_{\xi_0}V$, the corresponding limit
		\begin{equation}
		\lambda_{\delta\xi}=\lim_{t\to\infty}\frac{1}{t}\log\frac{\lVert M(t)\delta\xi\rVert_G}{\lVert \delta\xi(0)\rVert_G}
		\end{equation}
		exists.
		\item[(ii)] All Lyapunov exponents appear in conjugate pairs $(\lambda,-\lambda)$, such that the geometric multiplicity of the two conjugate exponents agrees.
	\end{enumerate}
\end{definition}
In short, condition (i) excludes systems with above-exponential or below-exponential growth, while condition (ii) excludes systems where two or more vectors become exponentially fast collinear under evolution by $M_{\xi_0}(t)$. Let us give an example for each condition that is not a regular Hamiltonian Lyapunov system. For both examples, we consider a single degree of freedom, such that we can express everything with respect to the Darboux basis $\mathcal{D}_V=(q,p)$.
\begin{enumerate}
	\item[(i)] Above-exponential growth and decay\\
	The time-dependent quadratic Hamiltonian $H(t)=e^{t}qp$ leads to the Hamiltonian flow
	\begin{equation}
	M(t)=\left(\begin{array}{cc}
	e^{e^t} & 0\\
	0 & e^{-e^t}
	\end{array}\right)\,,
	\end{equation}
	for which the Lyapunov exponents are ill defined because the defining limits do not exist. Thus, this system violates the first condition of regular Hamiltonian Lyapunov systems.
	\item[(ii)] Exponential collinearity\\
	The time dependent quadratic Hamiltonian $H(t)=\frac{1}{2}e^{t}p^2$ leads to the Hamiltonian flow 
	\begin{equation}
	M(t)=\left(\begin{array}{cc}
	1 & 0\\
	e^t & 1
	\end{array}\right)\,,
	\end{equation}
	which has Lyapunov exponents given by $\lambda_1=1$ and $\lambda_2=0$. The symplectic volume is still preserved under time evolution because arbitrary initial vectors become exponentially fast collinear, for instance
	\begin{equation}
	M\tra(t)q=q\,,\qquad M(t)\,p=p+e^{t}q\,,
	\end{equation}
	where the angle between the two vectors behaves as
	\begin{equation}
	\theta(t)=\cos^{-1}\left(\frac{\langle M\tra(t)q,M\tra(t)p\rangle_G}{\lVert M\tra(t)q\rVert_G\,\lVert M\tra(t)p\rVert_G}\right)\sim e^{-t}\quad\text{as}\quad t\to\infty\,,
	\end{equation}
	regardless of which positive definite metric we use. Clearly, this system does not have two conjugate Lyapunov exponents and does not fall into the class of regular Hamiltonian Lyapunov systems.
\end{enumerate}

The following notion of Lyapunov defect is important to show that for regular Hamiltonian systems the subsystem exponent can be simply computed using theorem~\ref{th:LambdaA}.
\begin{definition}[Subsystem defect]
	Given the inverse linear flow $M(t)$ and a subsystem $A\subset V$ with Lyapunov associated subsystem Lyapunov spectrum $(\lambda^A_1,\dots,\lambda^A_{2N_A})$, we define the \underline{Lyapunov defect}
	\begin{equation}
	\Lambda^*_A=\sum^{2N_A}_{i=1}\lambda^A_i-\Lambda_A\,,
	\end{equation}
	where $\mathcal{V}_{A}\subset A^*$ is an arbitrary $2N_A$-dimensional parallelepiped in $A$. If this limit exists, it is independent of the metric $G$ with which we measure the volume and we have $\Lambda^*_A\geq 0$.
\end{definition}
\begin{proof}
	The volume of a parallelepiped can be computed from the length of its $2N_A$ sides $M\tra(t)\,\ell^i$ and the $(2N_A-1)$ angles $\psi_i(t)$, which is the angle between $M\tra(t)\ell^i$ and the hyperplane spanned by the vectors $M\tra(t)\ell^j$ with $j=1,\dots,i-1$. The time dependent volume is then given by
	\begin{equation}
	\mathrm{Vol}_G\left(M\tra(t)\mathcal{V}_{A}\right)=\prod^{2N_A}_{i=1}\lVert M\tra(t)\,\ell^i\rVert\sin\psi_i(t)\,.
	\label{eq:Lyapunov-defect-volume}
	\end{equation}
	Given two distinct metrics $G$ and $\tilde{G}$, we can use the angle and length inequalities from above, to find the volume inequality
	\begin{equation}
	(1/b)^{2N_A}\,\mathrm{Vol}_G(M\tra(t)\mathcal{V}_{A})\leq\mathrm{Vol}_{\tilde{G}}(M\tra(t)\mathcal{V}_{A})\leq (1/a)^{2N_A}\,\mathrm{Vol}_G(M\tra(t)\mathcal{V}_{A})\,.
	\end{equation}
	This inequality already insures that the above limit is independent of the chosen metric. Moreover, the explicit expression in (\ref{eq:Lyapunov-defect-volume}) shows also that the volume is bounded from above by
	\begin{equation}
	\mathrm{Vol}_G(M\tra(t)\mathcal{V}_{A})\leq \prod^{2N_A}_{i=1}\lVert M\tra(t)\ell^i\rVert\propto \exp\sum^{2N_A}_{i=1}\lambda^A_i t\,.
	\end{equation}
	This implies $\Lambda_A\leq\sum^{2N_A}$ and thus, $\Lambda^*_A\geq 0$.
\end{proof}

For regular Hamiltonian systems, we can prove the following statements that we will need in the proof of our central theorem of this paper.
\begin{proposition}\label{app:prop3}
	In a regular Hamiltonian system, the Lyapunov defect $\Lambda^*_A$ of any subspace $A\in V$ vanishes. This implies that the asymptotic behavior of any volume $\mathcal{V}_A\subset A\subset V$ is given by
	\begin{equation}
	\Lambda_A=\sum^{2N_A}_{i=1}\lambda_i^A
	\end{equation}
	where $\lambda^A_i$ refers to subsystem Lyapunov spectrum of $A$.
\end{proposition}
\begin{proof}
	Let us recall that there is a special class of metrics on $V$, for which every symplectic transformation $M(t)$ and thus also $M\tra(t)$ preserves the $2N$-dimensional volume. These are all the metrics that give rise to the same volume form as the one induced by the symplectic form. This implies that the asymptotic behavior of every $2N$-dimensional region $\mathcal{V}\subset V$ shows the following asymptotic behavior
	\begin{equation}
	\Lambda_V=\lim_{t\to\infty}\frac{1}{t}\log\frac{\mathrm{Vol}_G\left(M\tra(t)\mathcal{V}\right)}{\mathrm{Vol}_G\left(\mathcal{V}\right)}=0\,,
	\end{equation}
	which holds with respect to all metrics $G$.\\
	From our previous discussion, we also recall that we must have
	\begin{equation}
	\Lambda_V=\sum^{2N}_{i=1}\lambda_i-\Lambda^*_V\,.
	\end{equation}
	If all Lyapunov exponents $\lambda_i$ come in conjugate pairs with equal multiplicities the sum in this expression vanishes. Thus, we have $\Lambda_V=-\Lambda^*_V$ which implies $\Lambda^*_V=0$ due to $\Lambda_V=0$.\\
	At this point, we only need to show that $\Lambda^*_V=0$ for the full system implies that we also have $\Lambda^*_A=0$ for all subsystems $A\subset V$. This follows from the fact that we can choose an initial parallelepiped $\mathcal{V}=\mathcal{V}_A\times\mathcal{V}_B$ with well known inequality
	\begin{equation}
	\mathrm{Vol}_G\left(M\tra(t)\mathcal{V}\right)\leq\mathrm{Vol}_G\left(M\tra(t)\mathcal{V}_A\right)\,\mathrm{Vol}_G\left(M\tra(t)\mathcal{V}_B\right)\,.
	\end{equation}
	This inequality implies $-\Lambda_V^*\leq -\Lambda_A^*-\Lambda_B^*$ where we recall $\lambda_A^*>0$ and $\Lambda_B^*>0$. Thus, $\Lambda_V^*=0$ implies that $\Lambda_A^*=0$ for all subspaces $A\subset V$ leading to $\Lambda_A=\sum^{2N_A}_{i=1}\lambda_i^A$.
\end{proof}
Let us emphasize that proposition~\ref{app:prop2} and~\ref{app:prop3} together provide an alternative full proof of theorem~\ref{th:LambdaA}, the main result of this paper. Put simply, the subsystem exponent $\Lambda_A$ for regular Hamiltonian systems is just given by the sum over the subsystem spectrum $\lambda_i^A$ which can be computed using the procedure explained in theorem~\ref{th:LambdaA} or equivalently in proposition~\ref{app:prop2}.

\section{Gaussian states and quadratic time-dependent Hamiltonians}\label{sec:quantum}
We review how symplectic methods and complex structures provide a tool for describing Gaussian states and their quantum evolution. These methods are instrumental in the derivation of a relation between symplectic volumes and the asymptotic growth of the entanglement entropy.

\subsection{Bosonic quantum systems and the symplectic group}\label{sec2:Fock}
We consider a quantum system with $N$ bosonic degrees of freedom \cite{berezin:1966}. The Hilbert space $\mathcal{H}$ of the system carries a regular representation of the commutation relations
\begin{equation}
	[\,\hat{\xi}^a,\hat{\xi}^b\,]=\mathrm{i}\,\Omega^{ab}\,.
	\label{eq:CR}
\end{equation}
Here $\Omega^{ab}$ is the symplectic structure discussed in section~\ref{sec2:phasespace} and the operators $\hat{\xi}^a$ can be understood as the quantization of the classical linear observables $\xi^a$ with  Poisson brackets $\{\xi^a,\xi^b\}=\Omega^{ab}$. A Fock representation of the commutation relations (\ref{eq:CR}) is obtained by introducing creation and annihilation operators with canonical commutation relations $[\hat{a}_i,\hat{a}_j^\dagger]=\delta_{ij}$, $[\hat{a}_i,\hat{a}_j]=0$, $[\hat{a}_i^\dagger,\hat{a}_j^\dagger]=0$. These operators define a set of orthonormal vectors $|n_1,\dots,n_N;\mathcal{D}\rangle$ with $n_i\in\mathbb{N}$, a Fock basis. The Fock vacuum $|0,\dots,0;\mathcal{D}\rangle$ is defined by
\begin{equation}
	\hat{a}_i\, |0,\dots,0;\mathcal{D}\rangle\,=0\,,\qquad i=1,\dots,N
	\label{eq:Fock vacuum}
\end{equation}
and the $n$-excitations state $|n_1,\dots,n_N;\mathcal{D}\rangle$ by
\begin{equation}
	|n_1,\dots,n_N;\mathcal{D}\rangle=\left(\prod^N_{i=1}\frac{(\hat{a}^\dagger_i)^{n_i}}{\sqrt{n_i!}}\right)|0,\dots,0;\mathcal{D}\rangle\,.
	\label{}
\end{equation}
The Hilbert space $\mathcal{H}$ is obtained by completing the span of these vectors in the norm induced by the scalar product $\langle 0,\dots,0;\mathcal{D}|0,\dots,0;\mathcal{D}\rangle=1$. The label $\mathcal{D}$ refers to a Darboux basis $\mathcal{D}=(q_i,p_i)$ of the classical phase space $V$. It enters in the definition of the representation of the commutation relations (\ref{eq:CR}) in the following way. We define position and momentum operators $\hat{q}_i=q_{ia}\,\hat{\xi}^a\,$, $\hat{p}_i=p_{ia}\,\hat{\xi}^a$ with $\Omega^{ab}q_{ia}\,p^j_{b}=\delta^{ij}$ and relate them to the creation and annihilation operators via\footnote{Following our index convention, it would be more natural to write $\hat{a}_i=a_{ib}\hat{\xi}^b$ to emphasize their relation to vectors in the complexified phase space $V_{\mathbb{C}}$, but we follow the standard convention of writing creation and annihilations operators as $\hat{a}_i^\dagger$ and $\hat{a}_i$.}
\begin{equation}
	\hat{a}_i=\frac{\hat{q}_i+\mathrm{i}\,\hat{p}_i}{\sqrt{2}}\,,\qquad \hat{a}^\dagger_i=\frac{\hat{q}_i-\mathrm{i}\,\hat{p}_i}{\sqrt{2}}\,.
	\label{eq:adef}
\end{equation}
These relations can be inverted to represent the operator $\hat{\xi}^a$ in terms of $\hat{a}_i$ and $\hat{a}_i^\dagger$,
\begin{equation}
	\hat{\xi}^b=\sum_{i=1}^N (u^b_i\,\hat{a}_i+u^{*b}_i\,\hat{a}^\dagger_i)\,,
	\label{}
\end{equation}
with coefficients $u^a_i$  determined by the choice of Darboux basis $\mathcal{D}$. With these definitions, $[\,\hat{\xi}^a,\hat{\xi}^b\,]=\mathrm{i}\,\Omega^{ab}$ on the Hilbert space $\mathcal{H}$.\\

The Hilbert space $\mathcal{H}$ carries a projective unitary representation of the inhomogeneous symplectic group $\mathrm{ISp}(2N,\mathbb{R})= \mathbb{R}^{2N}\ltimes \mathrm{Sp}(2N,\mathbb{R})$, which is the semi-direct product of phase space translations and the symplectic group \cite{folland:1989,degosson:2006,woit2015quantum}. An element of $\mathbb{R}^{2N}\ltimes \mathrm{Sp}(2N,\mathbb{R})$ can be uniquely parametrized by a pair $(\eta,M)$,
\begin{equation}
	\hat{\xi}^a\mapsto\, M^a{}_b\, \hat{\xi}^b\,+\,\eta^a\qquad\textrm{with}\quad M^a{}_b\in \mathrm{Sp}(2N,\mathbb{R})\quad\textrm{and}\quad \eta^a\in \mathbb{R}^{2N}\,.
	\label{}
\end{equation}
A unitary representation of the inhomogeneous symplectic group,
\begin{equation}
	U(M,\eta)\,\hat{\xi}^a\,U(M,\eta)^{-1}=\,M^a{}_b\, \hat{\xi}^b\,+\,\eta^a\,,
	\label{eq:UMz}
\end{equation}
is provided by the unitary operator $U(M,\eta)$ given by
\begin{equation}
	U(M,\eta)=\exp\big(\mathrm{i}\,\Omega_{ab}\eta^a\hat{\xi}^b\big)\,\exp\big(\mathrm{i}\,\frac{1}{2}h_{ab}\hat{\xi}^a\hat{\xi}^b\big)\,,
	\label{}
\end{equation}
where the symmetric matrix $h_{ab}$ is defined in terms of the generator of a symplectic transformation by $M^a{}_b=e^{\Omega^{ac}h_{cb}}$.

An immediate consequence of (\ref{eq:UMz}) is that, for systems with a finite number of degrees of freedom, two Fock space representations associated to different choices of Darboux basis $\mathcal{D}$ and $\tilde{\mathcal{D}}=M\mathcal{D}$ are related by the unitary transformation $U(M)$. This is a special case of the Stone-von Neumann theorem \cite{v1931eindeutigkeit,folland:1989}. A second consequence is that classical quadratic observables $\mathcal{O}=\frac{1}{2}h_{ab}\xi^a\xi^b$ promoted to operators  $\hat{\mathcal{O}}$ with symmetric (Weyl) ordering have commutation relations that reproduce the classical Poisson brackets, $[\hat{\mathcal{O}}_1,\hat{\mathcal{O}}_2]=\mathrm{i}\,\{\mathcal{O}_1,\mathcal{O}_2\}$.\footnote{This property cannot be extended to higher order observables as shown by the Groenewold-Van Hove no-go theorem \cite{groenewold1946principles}.} 
A third consequence  of (\ref{eq:UMz}) is that the unitary evolution generated by a quadratic Hamiltonian can be fully described in terms of linear symplectic transformations in phase space. This fact plays a major role in the analysis of this paper.

\subsection{Gaussian states and the complex structure $J$}
In (\ref{eq:Fock vacuum}) we defined the Fock vacuum $|0,\dots,0;\mathcal{D}\rangle$ as the state annihilated by all operators $a_i$ associated to an arbitrary choice of Darboux basis $\mathcal{D}$ in $(V,\Omega)$. Gaussian states provide a generalization of this notion. The relevant structure needed to define a Gaussian state is a complex structure $J^a{}_b$ compatible with the symplectic structure $\Omega^{ab}$ defined on phase space. A compatible complex structure $J^a{}_b$ is a linear map on phase space that (i) squares to minus the identity, (ii) is symplectic and (iii) gives rise to a symmetric positive definite metric $g_{ab}$:
\begin{equation}
	\textrm{(i)}\quad J^a{}_c\,J^c{}_b=-\delta^a{}_b,\qquad\textrm{(ii)}\quad J^a{}_c\, J^b{}_d\,\Omega^{cd}=\Omega^{ab} ,\qquad\textrm{(iii)}\quad
	g_{ab}=\omega_{ac}\,J^c{}_b\,.
	\label{eq:Jg}
\end{equation}
We define also the map $G^{ab}$ obtained by raising the indices of the metric $g_{ab}$ with the symplectic structure $\Omega^{ab}$,
\begin{equation}
	G^{ab}\equiv \Omega^{ac}g_{cd}\Omega^{db}\,,
	\label{}
\end{equation}
Note that by construction $G^{ab}$ is the inverse of the metric $g_{ab}$, i.e.  $G^{ac}g_{cb}=\delta^a{}_b$.\\

We define the Gaussian state $|J,\zeta\rangle$ as the state annihilated by the operator $a^b_{J\zeta}$, i.e. the solution of the equation
\begin{equation}
	a^b_{J\zeta}|J,\zeta\rangle=0\qquad \textrm{with}\qquad a^b_{J\zeta}\equiv \frac{(\hat{\xi}^b-\zeta^b)+\mathrm{i}\,J^b{}_a (\hat{\xi}^a-\zeta^a)}{\sqrt{2}}\,,
	\label{eq:defGaussian}
\end{equation}
where  $J^a{}_b$ is a compatible complex structure and $\zeta^a\in \mathbb{R}^{2N}$ a vector in phase space. This expression provides a formalization and generalization of (\ref{eq:Fock vacuum}).\\

The Fock vacuum $|0,\dots,0;\mathcal{D}\rangle$ defined in (\ref{eq:Fock vacuum}) is an example of Gaussian state. It corresponds to the complex structure $J^a{}_b=\sum_i \big(\Omega^{ac}q_{ic}\, q_{ib}+\Omega^{ac}p_{ic}\, p_{ib}\big)$ and zero shift vector $\zeta^a=0$, i.e. $|0,\dots,0;\mathcal{D}\rangle=|J,0\rangle$. Different choices $J$ and $\tilde{J}$ of complex structure are related by a symplectic transformation, $\tilde{J}=M^{-1}J M$. In the language of creation and annihilation operators this operation corresponds to a Bogoliubov transformation \cite{berezin:1966}. Given a choice of Fock vacuum $|J,0\rangle$, the state $|\tilde{J},0\rangle$ obtained by acting with a Bogoliubov transformation is generally called a squeezed vacuum \cite{walls2007quantum}. On the other hand, a displaced Fock vacuum corresponds to a translation $\zeta^a$ in phase space, $|J,\zeta\rangle$, also called a coherent state. For any choice of Darboux basis $\mathcal{D}=(q_i,p_i)$, the Schr\"odinger representation function $\psi(q_i)=\langle q_i|J,\zeta\rangle$ of a Gaussian state is given by a complex Gaussian function of $q_i$, which explains their name.\\

The one-point and the two-point correlation functions of a Gaussian state can be computed directly from the definition (\ref{eq:defGaussian}) and are given by
\begin{align}
	\langle J,\zeta|\,\hat{\xi}^a\,|J,\zeta\rangle\;=&\;\;\zeta^a\,,\label{eq:1point}\\[5pt]
	\langle J,\zeta|\,\hat{\xi}^a\,\hat{\xi}^b\,|J,\zeta\rangle\;=&\;\;\frac{G^{ab}+\mathrm{i}\,\Omega^{ab}}{2}\;+\;\zeta^a\,\zeta^b\,\qquad \textrm{with}\qquad G^{ab}=-J^a{}_c\Omega^{cb}\,.
	\label{eq:2point}
\end{align}
Higher $n$-point functions are determined by Wick theorem applied to the operator $\hat{\xi}^a-\zeta^a$. This property corresponds to the absence of non-Gaussianities: Correlations are completely determined by $J$ and $\eta$. Conversely, given the expectation value $\zeta^a$ and the connected symmetric part $G^{ab}$ of the $2$-point correlation function, the Gaussian state $|J,\zeta\rangle$ is determined by $J^a{}_b=-G^{ac}\omega_{cb}$.

\subsection{Quadratic time-dependent Hamiltonians}
We consider a quadratic time-dependent Hamiltonian $H(t)$,
\begin{equation}
	H(t)=\frac{1}{2}h_{ab}(t)\,\hat{\xi}^a\hat{\xi}^b+f_a(t)\,\hat{\xi}^a\,.
	\label{eq:quadHhat}
\end{equation}
The unitary evolution operator solves the Schr\"odinger equation $\,\mathrm{i}\,\frac{\partial}{\partial t}U(t)=H(t)\,U(t)$ and is given by the time-ordered exponential
\begin{equation}
	U(t)=\mathcal{T}\exp\left(-\mathrm{i}\int_0^tH(t')dt'\right)\,.
	\label{}
\end{equation}
The evolution of the observable $\hat{\xi}^a$ is then given by
\begin{equation}
	U(t)\,\hat{\xi}^a\,U(t)^{-1}\,=\,M^a{}_b(t)\, \hat{\xi}^b\,+\,\eta^a(t)\,,
	\label{}
\end{equation}
where $M^a{}_b(t)$ and $\eta^a(t)$ are defined by the classical Hamiltonian evolution and given in (\ref{eq:Mh}).\\

An important property of Gaussian states is that they provide exact solutions of the Schr\"odinger equation for a time-dependent quadratic Hamiltonian (\ref{eq:quadHhat}),
\begin{equation}
	\mathrm{i}\frac{\partial}{\partial t}\,|J_t,\zeta_t\rangle=H(t)|J_t,\zeta_t\rangle\,.
	\label{eq:Schr}
\end{equation}
Given a Gaussian state $|J_0,\zeta_0\rangle$ at the time $t=0$, the state at the time $t$ is 
\begin{equation}
	|J_t,\zeta_t\rangle\,=\,U(t)\,|J_0,\zeta_0\rangle\,,
	\label{eq:Jt}
\end{equation}
with $J_t$ and $\zeta_t$ determined as follows. The equation (\ref{eq:Schr}) defined on the Hilbert space $\mathcal{H}$ results in linear equations for the matrix $J_t$ and the vector $\eta_t$ on phase space,
\begin{align}
	&\frac{\partial}{\partial t}J_t=K(t)\,J_t-J_t \,K(t)\,,\\[1em]
	&\frac{\partial}{\partial t}\zeta_t=K(t)\,\zeta_t\,+\,k(t)\,,
	\label{}
\end{align}
with the matrix $K^a{}_b(t)=(\Omega^{ac}h_{cb}(t))$ and the vector $k^a(t)=(\Omega^{ab}f_b(t))$ are defined in terms of the parameters of the quadratic time-dependent Hamiltonian $H(t)$, (\ref{eq:quadHhat}). The linear equations for the complex structure $J_t$ and the shift $\eta_t$ can be solved as time-ordered series,
\begin{align}
	&J_t=\,M^{-1}(t) \;J_0\;M(t)\\
	&\zeta_t=M(t)\zeta_0\,+M(t)\int_0^t M^{-1}(t')\,k(t')\,dt'
	\label{}
\end{align}
where $J_0$ and $\zeta_0$ are initial conditions and 
\begin{equation}
	M(t)=\mathcal{T}\exp\left(\int_0^t K(t')\,dt'\right)
	\label{}
\end{equation}
is the symplectic matrix discussed in (\ref{eq:Mh}). 

Given an initial state $|J_0,\zeta_0\rangle$ and a quadratic time-dependent Hamiltonian $H(t)$, the evolution of the one-point and two-point correlation functions are given by (\ref{eq:1point}) and (\ref{eq:2point}) with $\zeta=\zeta_t$ and $G^{ab}=-J_t{}^a{}_c\,\Omega^{cb}$.

\subsection{Subsystems and the restricted complex structure}
We consider a bosonic quantum system consisting of two subsystems $A$ and $B$ with $N_A$ and $N_B$ degrees of freedom. The Hilbert space of the system decomposes in the tensor product of the Hilbert spaces of the two subsystems,
\begin{equation}
	\mathcal{H}=\mathcal{H}_A\otimes \mathcal{H}_B\,.
	\label{}
\end{equation}
The density matrix $\rho_A$ of a pure state $|\psi\rangle\in\mathcal{H}$ restricted to the subsystem $A$ is defined by
\begin{equation}
	\rho_A=\mathrm{Tr}_{\mathcal{H}_B}(|\psi\rangle\langle\psi|)\,.
	\label{}
\end{equation}
With this definition, the expectation value of any observable in the subsystem $A$ can be computed directly from the density matrix as a trace over the Hilbert space $\mathcal{H}_A$,
\begin{equation}
	\langle\psi|\mathcal{O}_A|\psi\rangle=\mathrm{Tr}_{\mathcal{H}_A}(\mathcal{O}_A\,\rho_A)\,.
	\label{}
\end{equation}
From an operational point of view a subsystem is determined by a subalgebra of observables, i.e. by a restriction of the set of measurements performed on the system. We discuss how the choice of subalgebra of observables identifies the subsystem $A$, its complement $B$, and allows us to compute the density matrix of a Gaussian state.\\

The observables of a bosonic quantum system form a Weyl algebra $\mathcal{A}_V=\mathrm{Weyl}(2N,\mathbb{C})$ generated by linear observables $\hat{\xi}^a$ with commutation relations $[\,\hat{\xi}^a,\hat{\xi}^b\,]=\mathrm{i}\,\Omega^{ab}\,$. We define a subsystem with $N_A$ degrees of freedom by choosing a subalgebra $\mathcal{A}_A\subset \mathcal{A}_V$ generated by a set of $N_A$ linear observables $\hat{\theta}^r$,  
\begin{equation}
	\hat{\theta}^r=\theta_a^r\,\hat{\xi}^a\qquad \textrm{with} \qquad r=1,\dots,2N_A
	\label{}
\end{equation}
and canonical commutation relations 
\begin{equation}
	[\,\hat{\theta}^r,\hat{\theta}^s\,]=\mathrm{i}\,\Omega_A^{rs}
	\label{eq:commOmegaA}
\end{equation}
where $\Omega_A^{rs}=\Omega^{ab}\theta_a^r\theta_b^s$ is required to be a symplectic structure on the vector space $A=\mathbb{R}^{2N_A}$, so that the couple $(A,\Omega_A)$ is a symplectic vector space. The Hilbert space $\mathcal{H}_A$ is obtained as a Fock representation of the Weyl algebra $\mathcal{A}_A=\mathrm{Weyl}(2N_A,\mathbb{C})$ as discussed in section~\ref{sec2:Fock}. We call $\phi_i,\pi_i$ a set of canonical observables in $A$ associated to the Darboux basis $\mathcal{D}_A=(\phi_1,\dots,\phi_{N_A},\pi_1,\dots,\pi_{N_A})$.

The algebra of observables describing the rest of the system is given by $\mathcal{A}'_A$, the commutant of $\mathcal{A}_A$ in  $\mathcal{A}_V$ defined by
\begin{equation}
	\mathcal{A}'_A\equiv \{\mathcal{O}\in \mathcal{A}_V\;|\;[\mathcal{O}_A,\mathcal{O}]=0\;\;\textrm{for all}\;\; \mathcal{O}_A\in \mathcal{A}_A\},
	\label{}
\end{equation}
i.e. the set of all operators which commute with all operators in $\mathcal{A}_A$. Here, the commutant $\mathcal{A}'_A$ is generated by linear operators with coefficients in $B$, the symplectic complement of $A$. Let us consider the subalgebra $\mathcal{A}_B\subset \mathcal{A}_V$ generated by a set of $N_B$ linear observables $\hat{\varTheta}^k$,  
\begin{equation}
	\hat{\varTheta}^k=\varTheta_a^k\,\hat{\xi}^a\qquad \textrm{with} \qquad k=1,\dots,2N_B
	\label{}
\end{equation}
and canonical commutation relations 
\begin{equation}
	[\,\hat{\varTheta}^k,\hat{\varTheta}^h\,]=\mathrm{i}\,\Omega_B^{kh}\qquad \textrm{and}\qquad [\,\hat{\theta}^r,\hat{\varTheta}^k\,]=0
	\label{}
\end{equation}
where $\Omega_B^{kh}=\Omega^{ab}\varTheta_a^k \varTheta_b^h$ is required to be a symplectic structure on the vector space $B^*=\mathbb{R}^{2N_B}$, so that $(B^*,\Omega_B)$ is a symplectic space. The requirement that $B^*$ is the symplectic complement of $A^*$ results in the commutation relation $[\,\hat{\theta}^r,\hat{\varTheta}^k\,]=0$. The Hilbert space $\mathcal{H}_B$ is obtained as a Fock representation of the Weyl algebra $\mathcal{A}'_A=\mathcal{A}_B=\mathrm{Weyl}(2N_B,\mathbb{C})$. We call $\Phi_i,\Pi_i$ a set of canonical observables in $B^*$ dual to the Darboux basis $\mathcal{D}_B=(\Phi_1,\dots,\Phi_{N_B},\Pi_1,\dots,\Pi_{N_B})$.

The subalgebra  $\mathcal{A}_A$ has a trivial center,\footnote{We give an example of subsystem defined by a subalgebra with non-trivial center. Consider a bosonic system with $N=3$ degrees of freedom. The algebra $\mathcal{A}_V$ of observables of the system is generated by elements of the Darboux basis $\mathcal{D}_V=(q_1,q_2,q_3,p_1,p_2,p_3)$. Let us consider the subalgebra $\mathcal{A}_C$ generated by $(q_1,p_1,\,q_2)$. Its commutant is $\mathcal{A}'_C=(q_3,p_3,\,q_2)$. As a result this subalgebra has a non-trivial center $\mathcal{Z}_C\equiv\mathcal{A}_C\cap \mathcal{A}'_C=(\mathbbm{1},q_2)''$. In this case the algebra of observables of the system decomposes in $\mathcal{A}_V=\oplus_\lambda\;(\mathcal{A}_C{}^{(\lambda)}\otimes \mathcal{A}'_C{}^{(\lambda)})$ and the Hilbert space decomposes in a direct sum of tensor products $\mathcal{H}=\oplus_\lambda\;(\mathcal{H}_C{}^{(\lambda)}\otimes \mathcal{H}'_C{}^{(\lambda)})$ where $\lambda$ is a basis of simultaneous eigenstates of the operators in the center (eigenstates of $q_2$ in this example). Choosing a symplectic subspace as done in (\ref{eq:commOmegaA}) guaranties that the center of the subalgebra is trivial and the Hilbert space decomposes into a tensor product.} i.e. $\mathcal{A}_A\cap\mathcal{A}'_A=\mathbbm{1}$. As a result the algebra of observables of the systems decomposes in a tensor product over the subsystem $A$ and its complement, $\mathcal{A}_V=\mathcal{A}_A\otimes \mathcal{A}_B$, and the Hilbert space of the system decomposes in the tensor product $\mathcal{H}=\mathcal{H}_A\otimes \mathcal{H}_B$. This decomposition reproduces at the quantum level the decomposition of phase space $V$ in two symplectic complements $A$ and $B$ with Darboux basis $\mathcal{D}_V=(\mathcal{D}_A,\mathcal{D}_B)$.\\

Given a subsystem $A$, the Gaussian state $|J,\zeta\rangle\in \mathcal{H}$ admits a Schmidt decomposition that selects the Darboux basis $\mathcal{D}_A$ and $\mathcal{D}_B$ in the two complementary subsystems so that the state can be written in the form \cite{Bianchi:2015fra}
\begin{equation}
	|J,\zeta\rangle=\!\sum_{n_i=0}^\infty\!\!\left(\prod_{i=1}^{N_e}\sqrt{\frac{2\,(\nu_i -1)^{n_i}}{(\nu_i +1)^{n_i+1}}\!\!}\;\right)U(\zeta_A)|n_1, .\,.\,,n_{N_e},0, .\,.\,;\mathcal{D}_A\rangle\otimes U(\zeta_B)|n_1, .\,.\,,n_{N_e},0, .\,.\,;\mathcal{D}_B\rangle.
	\label{eq:Schmidt}
\end{equation}
The unitary operator $U(\zeta_A)$ generates a shift in $A$ with parameter $\zeta_A^r=\theta^r_a\zeta^a$. Note that $U(\zeta)=U(\zeta_A)\otimes U(\zeta_B)$. The parameters $\nu_i$ are the positive eigenvalues of the matrix $[\mathrm{i}J]_A=(\mathrm{i}\theta_a^r\,J^a{}_b \,\vartheta^b_s)$ obtained as the restriction to $A$ of $\mathrm{i} J$,
\begin{equation}
	\textrm{Eig}(\,[\mathrm{i}J]_A)=\{\pm\nu_i\} \qquad \textrm{with}\quad \nu_i\geq 1\,.
	\label{}
\end{equation}
The condition  $\nu_i\geq 1$ follows from the fact that the matrix $[J]_A$ is the restriction of a complex structure in $V$. We define the number of \emph{entangled pairs} in the decomposition $\mathcal{H}_A\otimes\mathcal{H}_B$ as the number of non-trivial terms in the sum in (\ref{eq:Schmidt}). This is also the number of positive eigenvalues of $[\mathrm{i}J]_A$ that differ from $+1$, or equivalently the rank of the matrix $\mathbbm{1}-([\mathrm{i}J]_A)^2$,
\begin{equation}
	N_e\equiv \mathrm{rank}\big(\mathbbm{1}-([\mathrm{i}J]_A)^2\big)\;\leq \;\min(N_A,N_B)\,.
	\label{}
\end{equation}
Note that the positive eigenvalues $\nu_i\neq 1$ of $[\mathrm{i}J]_A$ and of $[\mathrm{i}J]_B$ coincide, see figure~\ref{fig:entangled-pairs}. If all $\nu_i=1$, then $N_e=0$ and the Gaussian state $|J,\zeta\rangle$ factorizes in a tensor product of Gaussian states.

The reduced density matrix $\rho_A$ of a Gaussian state is immediate to obtain once the Schmidt decomposition is known,
\begin{equation}
	\rho_A(J,\zeta)=\mathrm{Tr}_{\mathcal{H}_A}\big(|J,\zeta\rangle\langle J,\zeta|\big)=\,U(\zeta_A)\,\rho_A(J)\,U(\zeta_A)^{-1}
	\label{eq:rho=UrhoU}
\end{equation}
where
\begin{equation}
	\rho_A(J)=\sum_{n_i=0}^\infty\left(\prod_{i=1}^{N_e} \frac{2}{\nu_i+1}\Bigg(\frac{\nu_i -1}{\nu_i +1}\Bigg)^{\!\!n_i}\right)\;|n_1, .\,.\,,n_{N_e},0, .\,.\,;\mathcal{D}_A\rangle\langle n_1, .\,.\,,n_{N_e},0, .\,.\,;\mathcal{D}_A|.
	\label{eq:rhoJ}
\end{equation}
We note that the density matrix can be written in the compact operatorial form
\begin{equation}
	\rho_A(J)=e^{-H_A}
	\label{}
\end{equation}
with the modular Hamiltonian $H_A$ given by
\begin{equation}
	H_A=\frac{1}{2}q_{rs}\hat{\theta}^r\hat{\theta}^s\,+\,E_0\qquad\textrm{and}\qquad q_{rs}=2\,\mathrm{i}\,\omega_{rk}\,\mathrm{arcoth}\big(\mathrm{i}\,J^k{}_s\big),
	\label{}
\end{equation}
where $(J^r{}_s)=(\theta_a^r\,J^a{}_b\,\vartheta^b_s)=[J]_A$, $\;(\omega_{rs})=(\vartheta^a_r\,\omega_{ab} \vartheta^b_s)=[\omega]_A=[\Omega^{-1}]_A$, and $E_0$ is a constant that fixes the normalization $\mathrm{Tr}_{\mathcal{H}_A} \rho_A\,=1$.\\

\begin{figure}
	\centering
	\noindent\makebox[\linewidth]{
		\includegraphics[width=.75\linewidth]{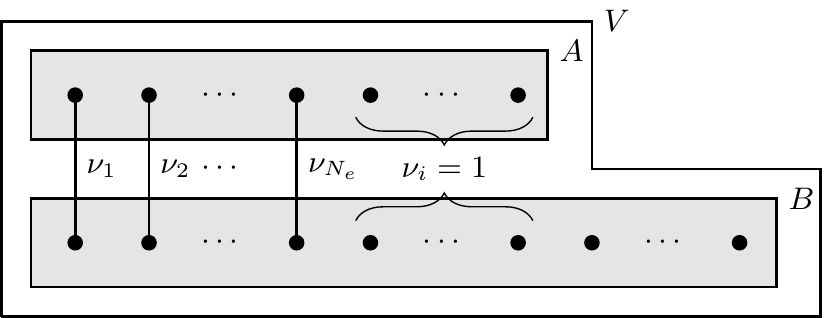}
	}
	\caption{\emph{Entanglement structure of Gaussian states}. We illustrate the entanglement structure of an arbitrary squeezed vacuum $|J\rangle$ with subsystems $A$ and $B$: We can always find a Darboux frame $\mathcal{D}_V=(\mathcal{D}_A,\mathcal{D}_B)$, such that only pairs of degrees of freedom are entangled across $A$ and $B$ with squeezing parameters $\nu_i$. Every black dot represents a degree of freedom, or equivalently a conjugate variable pair $(\varphi_i,\pi_i)$ appearing as basis vectors in $\mathcal{D}_A$ or $\mathcal{D}_B$, every link represents the entanglement between the two connected degrees of freedom. Note that we take $N_A\leq N_B$ and find that only up to $N_A$ pairs can be entangled. The remaining $N_B-N_A$ degrees of freedom in subsystem $B$ do not have a partner in subsystem $A$ leading to squeezing parameters $\nu_i=0$ for $i>N_A$. This is the reason why the maximal number of entangled degrees of freedom is dictated by the smaller of the two subsystems.}
	\label{fig:entangled-pairs}
\end{figure}

\subsection{Entanglement entropy and R\'{e}nyi entropy of Gaussian states}\label{sec2:renyi}
Complete knowledge of the state of a system does not imply knowledge of the state of its subsystems. This genuinely quantum-mechanical property is captured by the notion of entanglement entropy. The entanglement entropy $S_A(|\psi\rangle)$ of a pure state $|\psi\rangle$ restricted to the subsystem $A$ is given by the von Neumann entropy of the reduced state,
\begin{equation}
	S_A(|\psi\rangle)\equiv -\mathrm{Tr}_{\mathcal{H}_A}\big(\rho_A\log \rho_A\big)\,.
	\label{}
\end{equation}
To compute the entanglement entropy it is useful to introduce the function $Z(\beta)$ defined as the trace of the density matrix raised to the power $\beta$,
\begin{equation}
	Z(\beta)=\mathrm{Tr}_{\mathcal{H}_A}\big(\rho_A{}^\beta\big)\,.
	\label{eq:Zbeta}
\end{equation}
By construction $Z(0)=N_A$ and $Z(1)=1$. The function $Z(\beta)$ provides an efficient method for computing the entanglement entropy,
\begin{equation}
	S_A(|\psi\rangle)=\left.\Big(1-\beta\frac{\partial}{\partial \beta}\Big)\log Z(\beta)\right|_{\beta=1}.
	\label{}
\end{equation}
In the case of a Gaussian state $|J,\zeta\rangle$, the function $Z(\beta)$ can be expressed in term of the eigenvalues of $[\mathrm{i}J]_A$ using formulae (\ref{eq:rho=UrhoU}) and (\ref{eq:rhoJ}),
\begin{equation}
	\log Z(\beta)=-\sum_{i=1}^{N_e}\log\left(\Big(\frac{\nu_i+1}{2}\Big)^\beta-\Big(\frac{\nu_i-1}{2}\Big)^\beta\right)\,.
	\label{}
\end{equation}
It can also be expressed as a trace over the vector space $A$ of a function of the matrix $[\mathrm{i}J]_A$, 
\begin{equation}
	\log Z(\beta)=-\frac{1}{2}\,\mathrm{tr}\,\log\left| 
	\left|\frac{\mathbbm{1}+[\mathrm{i}J]_A}{2}\right|^\beta - \left|\frac{\mathbbm{1}-[\mathrm{i}J]_A}{2}\right|^\beta 
	\right|\,.
	\label{}
\end{equation}
The entanglement entropy of a Gaussian state can be computed from $Z(\beta)$ and expressed in terms of the eigenvalues $\nu_i$, \cite{braunstein2005quantum,ferraro2005gaussian,weedbrook2012gaussian,adesso:2014co}
\begin{equation}
	S_A(|J,\zeta\rangle)=\sum_{i=1}^{N_e} S(\nu_i)\qquad\textrm{where}\quad 
	S(\nu)\equiv \frac{\nu+1}{2}\log\frac{\nu+1}{2}-\frac{\nu-1}{2}\log\frac{\nu-1}{2}\,,
	\label{eq:sentropy}
\end{equation}
or equivalently in terms of the matrix $[\mathrm{i}J]_A$, \cite{Bianchi:2015fra}
\begin{equation}
	S_A(|J,\zeta\rangle)=\textrm{tr}\left(\frac{\mathbbm{1}+[\text{i} J]_A}{2}\log\Big|\frac{\mathbbm{1}+[\text{i} J]_A}{2}\Big|\right)\,.
	\label{eq:BHY}
\end{equation}
We can also compute the R\'{e}nyi entropy of order two,\footnote{The R\'{e}nyi entropy of order $n$ is defined as $R^{(n)}_A(|\psi\rangle)\equiv-\frac{1}{n-1}\log\textrm{Tr}_{\mathcal{H}_A}(\rho_A{}^n)=-\frac{1}{n-1}\log Z(n)$.}
\begin{equation}
	R_A(|\psi\rangle)\equiv -\log \textrm{Tr}_{\mathcal{H}_A}(\rho_A{}^2)\;=\;-\log Z(2)\,.
	\label{}
\end{equation}
The R\'{e}nyi entropy of a Gaussian state is
\begin{equation}
	R_A(|J,\zeta\rangle)=\sum_{i=1}^{N_e}\log \nu_i\,,
	\label{}
\end{equation}
which can be expressed in terms of the determinant of the matrix $[\mathrm{i}J]_A$,
\begin{equation}
	R_A(|J,\zeta\rangle)=\frac{1}{2}\log\big|\det\big([\mathrm{i}J]_A\big)\big|\,.
	\label{}
\end{equation}
This expression plays a central role in the analysis presented in this paper in section~\ref{sec:growth}.

\providecommand{\href}[2]{#2}\begingroup\raggedright\endgroup


\begin{thebibliography}{100}

\bibitem{polkovnikov2011colloquium}
A.~Polkovnikov, K.~Sengupta, A.~Silva, and M.~Vengalattore, {\it Colloquium:
  Nonequilibrium dynamics of closed interacting quantum systems},  {\em Reviews
  of Modern Physics} {\bf 83} (2011), no.~3 863,
  [\href{http://arxiv.org/abs/1007.5331}{{\tt arXiv:1007.5331}}].

\bibitem{gogolin2016equilibration}
C.~Gogolin and J.~Eisert, {\it Equilibration, thermalisation, and the emergence
  of statistical mechanics in closed quantum systems},  {\em Rep.Prog.Phys.}
  {\bf 79} (2016), no.~5 056001, [\href{http://arxiv.org/abs/1503.07538}{{\tt
  arXiv:1503.07538}}].

\bibitem{d2016quantum}
L.~D'Alessio, Y.~Kafri, A.~Polkovnikov, and M.~Rigol, {\it From quantum chaos
  and eigenstate thermalization to statistical mechanics and thermodynamics},
  {\em Advances in Physics} {\bf 65} (2016), no.~3 239--362,
  [\href{http://arxiv.org/abs/1509.06411}{{\tt arXiv:1509.06411}}].

\bibitem{Zurek:1994wd}
W.~H. Zurek and J.~P. Paz, {\it {Decoherence, chaos, and the second law}},
  {\em Phys.Rev.Lett.} {\bf 72} (1994) 2508,
  [\href{http://arxiv.org/abs/gr-qc/9402006}{{\tt gr-qc/9402006}}].

\bibitem{miller1999signatures}
P.~A. Miller and S.~Sarkar, {\it Signatures of chaos in the entanglement of two
  coupled quantum kicked tops},  {\em Phys.Rev.E} {\bf 60} (1999), no.~2 1542.

\bibitem{pattanayak1999lyapunov}
A.~K. Pattanayak, {\it Lyapunov exponents, entropy production, and
  decoherence},  {\em Phys.Rev.Lett.} {\bf 83} (1999), no.~22 4526,
  [\href{http://arxiv.org/abs/chao-dyn/9911017}{{\tt chao-dyn/9911017}}].

\bibitem{monteoliva2000decoherence}
D.~Monteoliva and J.~P. Paz, {\it Decoherence and the rate of entropy
  production in chaotic quantum systems},  {\em Phys.Rev.Lett.} {\bf 85}
  (2000), no.~16 3373, [\href{http://arxiv.org/abs/quant-ph/0007052}{{\tt
  quant-ph/0007052}}].

\bibitem{tanaka2002saturation}
A.~Tanaka, H.~Fujisaki, and T.~Miyadera, {\it Saturation of the production of
  quantum entanglement between weakly coupled mapping systems in a strongly
  chaotic region},  {\em Phys.Rev.E} {\bf 66} (2002), no.~4 045201,
  [\href{http://arxiv.org/abs/quant-ph/0209086}{{\tt quant-ph/0209086}}].

\bibitem{kim2013ballistic}
H.~Kim and D.~A. Huse, {\it Ballistic spreading of entanglement in a diffusive
  nonintegrable system},  {\em Physical review letters} {\bf 111} (2013),
  no.~12 127205.

\bibitem{Calabrese:2005in}
P.~Calabrese and J.~L. Cardy, {\it {Evolution of entanglement entropy in
  one-dimensional systems}},  {\em J.Stat.Mech.} {\bf 0504} (2005) P04010,
  [\href{http://arxiv.org/abs/cond-mat/0503393}{{\tt cond-mat/0503393}}].

\bibitem{Calabrese:2004eu}
P.~Calabrese and J.~L. Cardy, {\it {Entanglement entropy and quantum field
  theory}},  {\em J.Stat.Mech.} {\bf 0406} (2004) P06002,
  [\href{http://arxiv.org/abs/hep-th/0405152}{{\tt hep-th/0405152}}].

\bibitem{cotler2016entanglement}
J.~S. Cotler, M.~P. Hertzberg, M.~Mezei, and M.~T. Mueller, {\it Entanglement
  growth after a global quench in free scalar field theory},  {\em JHEP} {\bf
  2016} (2016), no.~11 166, [\href{http://arxiv.org/abs/1609.00872}{{\tt
  arXiv:1609.00872}}].

\bibitem{balasubramanian2011thermalization}
V.~Balasubramanian, A.~Bernamonti, J.~de~Boer, N.~Copland, B.~Craps,
  E.~Keski-Vakkuri, B.~M{\"u}ller, A.~Sch{\"a}fer, M.~Shigemori, and
  W.~Staessens, {\it Thermalization of strongly coupled field theories},  {\em
  Phys.Rev.Lett.} {\bf 106} (2011), no.~19 191601,
  [\href{http://arxiv.org/abs/1012.4753}{{\tt arXiv:1012.4753}}].

\bibitem{balasubramanian2011holographic}
V.~Balasubramanian, A.~Bernamonti, J.~de~Boer, N.~Copland, B.~Craps,
  E.~Keski-Vakkuri, B.~M{\"u}ller, A.~Sch{\"a}fer, M.~Shigemori, and
  W.~Staessens, {\it Holographic thermalization},  {\em Phys.Rev.D} {\bf 84}
  (2011), no.~2 026010, [\href{http://arxiv.org/abs/1103.2683}{{\tt
  arXiv:1103.2683}}].

\bibitem{Hartman:2013qma}
T.~Hartman and J.~Maldacena, {\it {Time Evolution of Entanglement Entropy from
  Black Hole Interiors}},  {\em JHEP} {\bf 1305} (2013) 014,
  [\href{http://arxiv.org/abs/1303.1080}{{\tt arXiv:1303.1080}}].

\bibitem{Liu:2013iza}
H.~Liu and S.~J. Suh, {\it {Entanglement Tsunami: Universal Scaling in
  Holographic Thermalization}},  {\em Phys.Rev.Lett.} {\bf 112} (2014) 011601,
  [\href{http://arxiv.org/abs/1305.7244}{{\tt arXiv:1305.7244}}].

\bibitem{Liu:2013qca}
H.~Liu and S.~J. Suh, {\it {Entanglement growth during thermalization in
  holographic systems}},  {\em Phys.Rev.} {\bf D89} (2014), no.~6 066012,
  [\href{http://arxiv.org/abs/1311.1200}{{\tt arXiv:1311.1200}}].

\bibitem{muller2011entropy}
B.~M{\"u}ller and A.~Sch{\"a}fer, {\it Entropy creation in relativistic heavy
  ion collisions},  {\em ?Int.J.Mod.Phys.E} {\bf 20} (2011), no.~11
  2235--2267, [\href{http://arxiv.org/abs/1110.2378}{{\tt arXiv:1110.2378}}].

\bibitem{kunihiro2010chaotic}
T.~Kunihiro, B.~M{\"u}ller, A.~Ohnishi, A.~Sch{\"a}fer, T.~T. Takahashi, and
  A.~Yamamoto, {\it Chaotic behavior in classical yang-mills dynamics},  {\em
  Phys.Rev.D} {\bf 82} (2010), no.~11 114015,
  [\href{http://arxiv.org/abs/1008.1156}{{\tt arXiv:1008.1156}}].

\bibitem{Hashimoto:2016wme}
K.~Hashimoto, K.~Murata, and K.~Yoshida, {\it {Chaos in chiral condensates in
  gauge theories}},  {\em Phys.Rev.Lett.} {\bf 117} (2016), no.~23 231602,
  [\href{http://arxiv.org/abs/1605.08124}{{\tt arXiv:1605.08124}}].

\bibitem{sekino2008fast}
Y.~Sekino and L.~Susskind, {\it Fast scramblers},  {\em JHEP} {\bf 2008}
  (2008), no.~10 065, [\href{http://arxiv.org/abs/0808.2096}{{\tt
  arXiv:0808.2096}}].

\bibitem{VanRaamsdonk:2010pw}
M.~Van~Raamsdonk, {\it {Building up spacetime with quantum entanglement}},
  {\em Gen.Rel.Grav.} {\bf 42} (2010) 2323--2329,
  [\href{http://arxiv.org/abs/1005.3035}{{\tt arXiv:1005.3035}}].

\bibitem{Bianchi:2012ev}
E.~Bianchi and R.~C. Myers, {\it {On the Architecture of Spacetime Geometry}},
  {\em Class.Quant.Grav.} {\bf 31} (2014) 214002,
  [\href{http://arxiv.org/abs/1212.5183}{{\tt arXiv:1212.5183}}].

\bibitem{Bianchi:2015edy}
E.~Bianchi, L.~Hackl, and N.~Yokomizo, {\it {Entanglement time in the
  primordial universe}},  {\em Int.J.Mod.Phys.} {\bf D24} (2015), no.~12
  1544006, [\href{http://arxiv.org/abs/1512.08959}{{\tt arXiv:1512.08959}}].

\bibitem{Susskind:2014moa}
L.~Susskind, {\it {Entanglement is not enough}},  {\em Fortsch.Phys.} {\bf 64}
  (2016) 49--71, [\href{http://arxiv.org/abs/1411.0690}{{\tt
  arXiv:1411.0690}}].

\bibitem{Jefferson:2017sdb}
R.~A. Jefferson and R.~C. Myers, {\it {Circuit complexity in quantum field
  theory}},  \href{http://arxiv.org/abs/1707.08570}{{\tt arXiv:1707.08570}}.

\bibitem{Chapman:2017rqy}
S.~Chapman, M.~P. Heller, H.~Marrochio, and F.~Pastawski, {\it {Towards
  Complexity for Quantum Field Theory States}},
  \href{http://arxiv.org/abs/1707.08582}{{\tt arXiv:1707.08582}}.

\bibitem{latora1999kolmogorov}
V.~Latora and M.~Baranger, {\it Kolmogorov-sinai entropy rate versus physical
  entropy},  {\em Phys.Rev.Lett.} {\bf 82} (1999), no.~3 520,
  [\href{http://arxiv.org/abs/chao-dyn/9806006}{{\tt chao-dyn/9806006}}].

\bibitem{falcioni2005production}
M.~Falcioni, L.~Palatella, and A.~Vulpiani, {\it Production rate of the
  coarse-grained gibbs entropy and the kolmogorov-sinai entropy: A real
  connection?},  {\em Phys.Rev.E} {\bf 71} (2005), no.~1 016118,
  [\href{http://arxiv.org/abs/nlin/0407056}{{\tt nlin/0407056}}].

\bibitem{kolmogorov1958new}
A.~N. Kolmogorov, {\it A new metric invariant of transient dynamical systems
  and automorphisms in lebesgue spaces},  in {\em Dokl.Akad.Nauk SSSR (NS)},
  vol.~119, pp.~861--864, 1958.

\bibitem{Sinai:2009}
Y.~Sinai, {\it {K}olmogorov-{S}inai entropy},  {\em Scholarpedia} {\bf 4}
  (2009), no.~3 2034. {revision 91406}.

\bibitem{zaslavsky2008hamiltonian}
G.~M. Zaslavsky, {\em Hamiltonian chaos and fractional dynamics}.
\newblock Oxford University Press, 2008.

\bibitem{vulpiani2010chaos}
M.~Cencini, F.~Cecconi, and A.~Vulpiani, {\em Chaos: from simple models to
  complex systems}, vol.~17.
\newblock World Scientific, 2010.

\bibitem{kunihiro2009towards}
T.~Kunihiro, B.~M{\"u}ller, A.~Ohnishi, and A.~Sch{\"a}fer, {\it Towards a
  theory of entropy production in the little and big bang},  {\em
  Progr.Theor.Ph.} {\bf 121} (2009), no.~3 555--575,
  [\href{http://arxiv.org/abs/0809.4831}{{\tt arXiv:0809.4831}}].

\bibitem{Asplund:2015osa}
C.~T. Asplund and D.~Berenstein, {\it {Entanglement entropy converges to
  classical entropy around periodic orbits}},  {\em Ann.Phys.} {\bf 366} (2016)
  113--132, [\href{http://arxiv.org/abs/1503.04857}{{\tt arXiv:1503.04857}}].

\bibitem{Bianchi:2015fra}
E.~Bianchi, L.~Hackl, and N.~Yokomizo, {\it {Entanglement entropy of squeezed
  vacua on a lattice}},  {\em Phys.Rev.} {\bf D92} (2015), no.~8 085045,
  [\href{http://arxiv.org/abs/1507.01567}{{\tt arXiv:1507.01567}}].

\bibitem{Vidmar:2017uux}
L.~Vidmar, L.~Hackl, E.~Bianchi, and M.~Rigol, {\it {Entanglement Entropy of
  Eigenstates of Quadratic Fermionic Hamiltonians}},  {\em Phys. Rev. Lett.}
  {\bf 119} (2017), no.~2 020601, [\href{http://arxiv.org/abs/1703.02979}{{\tt
  arXiv:1703.02979}}].

\bibitem{holevo2011probabilistic}
A.~S. Holevo, {\em Probabilistic and statistical aspects of quantum theory},
  vol.~1.
\newblock Springer, 2011.

\bibitem{vedral2002role}
V.~Vedral, {\it The role of relative entropy in quantum information theory},
  {\em Reviews of Modern Physics} {\bf 74} (2002), no.~1 197.

\bibitem{ohya2004quantum}
M.~Ohya and D.~Petz, {\em Quantum entropy and its use}.
\newblock Springer Science \& Business Media, 2004.

\bibitem{Hackl:2017ndi}
L.~Hackl, E.~Bianchi, R.~Modak, and M.~Rigol, {\it {Entanglement production in
  bosonic systems: Linear and logarithmic growth}},
  \href{http://arxiv.org/abs/1710.04279}{{\tt arXiv:1710.04279}}.

\bibitem{floquet1883equations}
G.~Floquet, {\it Sur les equations differentielles lineaires},  {\em Ann.ENS}
  {\bf 12} (1883), no.~1883 47--88.

\bibitem{chicone1999or}
C.~Chicone, {\em Ordinary Differential Equations with Applications}.
\newblock Springer, 1999.

\bibitem{ashtekar1980geometrical}
A.~Ashtekar and A.~Magnon-Ashtekar, {\it A geometrical approach to external
  potential problems in quantum field theory},  {\em Gen.Rel.Grav.} {\bf 12}
  (1980), no.~3 205--223.

\bibitem{ashtekar1975quantum}
A.~Ashtekar and A.~Magnon, {\it Quantum fields in curved space-times},  in {\em
  Proceedings of the Royal Society of London A: Mathematical, Physical and
  Engineering Sciences}, vol.~346, pp.~375--394, The Royal Society, 1975.

\bibitem{wald1994quantum}
R.~M. Wald, {\em Quantum field theory in curved spacetime and black hole
  thermodynamics}.
\newblock University of Chicago Press, 1994.

\bibitem{haag2012local}
R.~Haag, {\em Local quantum physics: Fields, particles, algebras}.
\newblock Springer, 2012.

\bibitem{shale1962linear}
D.~Shale, {\it Linear symmetries of free boson fields},  {\em Transactions of
  the American Mathematical Society} (1962) 149--167.

\bibitem{shale1964states}
D.~Shale and W.~F. Stinespring, {\it States of the clifford algebra},  {\em
  Ann.Math.} (1964) 365--381.

\bibitem{ottesen2008infinite}
J.~T. Ottesen, {\em Infinite dimensional groups and algebras in quantum
  physics}, vol.~27.
\newblock Springer, 2008.

\bibitem{berezin:1966}
F.~Berezin, {\em The method of second quantization}.
\newblock Pure and applied physics. Academic Press, 1966.

\bibitem{Birrell:1982ix}
N.~D. Birrell and P.~C.~W. Davies, {\em {Quantum Fields in Curved Space}}.
\newblock Cambridge Monographs on Mathematical Physics. Cambridge Univ. Press,
  Cambridge, UK, 1984.

\bibitem{parker2009quantum}
L.~Parker and D.~Toms, {\em Quantum field theory in curved spacetime: quantized
  fields and gravity}.
\newblock Cambridge University Press, 2009.

\bibitem{Hawking:1974sw}
S.~Hawking, {\it {Particle Creation by Black Holes}},  {\em Commun.Math.Phys.}
  {\bf 43} (1975) 199--220.

\bibitem{schwinger1951gauge}
J.~Schwinger, {\it On gauge invariance and vacuum polarization},  {\em
  Phys.Rev.} {\bf 82} (1951), no.~5 664.

\bibitem{Greiner:1985ce}
W.~Greiner, B.~Muller, and J.~Rafelski, {\em {Quantum electrodynamics of strong
  fields}}.
\newblock Springer, 1985.

\bibitem{casimir1948attraction}
H.~B. Casimir, {\it On the attraction between two perfectly conducting plates},
   in {\em Proc.KNAW}, vol.~51, pp.~793--795, 1948.

\bibitem{moore1970quantum}
G.~T. Moore, {\it Quantum theory of the electromagnetic field in a
  variable-length one-dimensional cavity},  {\em J.Math.Phys.} {\bf 11} (1970),
  no.~9 2679--2691.

\bibitem{Torre:1998eq}
C.~G. Torre and M.~Varadarajan, {\it {Functional evolution of free quantum
  fields}},  {\em Class.Quant.Grav.} {\bf 16} (1999) 2651--2668,
  [\href{http://arxiv.org/abs/hep-th/9811222}{{\tt hep-th/9811222}}].

\bibitem{Agullo:2015qqa}
I.~Agullo and A.~Ashtekar, {\it {Unitarity and ultraviolet regularity in
  cosmology}},  {\em Phys.Rev.D} {\bf 91} (2015), no.~12 124010,
  [\href{http://arxiv.org/abs/1503.03407}{{\tt arXiv:1503.03407}}].

\bibitem{sorkin1983entropy}
R.~D. Sorkin, {\it On the entropy of the vacuum outside a horizon},  in {\em
  Tenth International Conference on General Relativity and Gravitation (held in
  Padova, 4-9 July, 1983), Contributed Papers}, vol.~2, pp.~734--736, 1983.
\newblock \href{http://arxiv.org/abs/1402.3589}{{\tt arXiv:1402.3589}}.

\bibitem{srednicki1993entropy}
M.~Srednicki, {\it {Entropy and area}},  {\em Phys.Rev.Lett.} {\bf 71} (1993)
  666--669, [\href{http://arxiv.org/abs/hep-th/9303048}{{\tt hep-th/9303048}}].

\bibitem{eisert2010colloquium}
J.~Eisert, M.~Cramer, and M.~B. Plenio, {\it Colloquium: Area laws for the
  entanglement entropy},  {\em Rev.Mod.Phys.} {\bf 82} (2010), no.~1 277,
  [\href{http://arxiv.org/abs/0808.3773}{{\tt arXiv:0808.3773}}].

\bibitem{Hollands:2017dov}
S.~Hollands and K.~Sanders, {\it {Entanglement measures and their properties in
  quantum field theory}},  \href{http://arxiv.org/abs/1702.04924}{{\tt
  arXiv:1702.04924}}.

\bibitem{bombelli1986quantum}
L.~Bombelli, R.~K. Koul, J.~Lee, and R.~D. Sorkin, {\it Quantum source of
  entropy for black holes},  {\em Phys.Rev.D} {\bf 34} (1986), no.~2 373.

\bibitem{Casini:2008wt}
H.~Casini and M.~Huerta, {\it {Remarks on the entanglement entropy for
  disconnected regions}},  {\em JHEP} {\bf 03} (2009) 048,
  [\href{http://arxiv.org/abs/0812.1773}{{\tt arXiv:0812.1773}}].

\bibitem{Bianchi:2014bma}
E.~Bianchi, T.~De~Lorenzo, and M.~Smerlak, {\it {Entanglement entropy
  production in gravitational collapse: covariant regularization and solvable
  models}},  {\em JHEP} {\bf 06} (2015) 180,
  [\href{http://arxiv.org/abs/1409.0144}{{\tt arXiv:1409.0144}}].

\bibitem{holzhey1994geometric}
C.~Holzhey, F.~Larsen, and F.~Wilczek, {\it {Geometric and renormalized entropy
  in conformal field theory}},  {\em Nucl.Phys.} {\bf B424} (1994) 443--467,
  [\href{http://arxiv.org/abs/hep-th/9403108}{{\tt hep-th/9403108}}].

\bibitem{Bianchi2017entropy}
E.~Bianchi and A.~Satz, {\it Entropy of a subalgebgra of observables and the
  geometric entanglement entropy}, .

\bibitem{weinberg1995quantum}
S.~Weinberg, {\em The quantum theory of fields}, vol.~2.
\newblock Cambridge University press, 1995.

\bibitem{Strocchi:2015uaa}
F.~Strocchi, {\em Symmetry breaking}, vol.~643.
\newblock Springer, 2005.

\bibitem{calzetta1988nonequilibrium}
E.~Calzetta and B.-L. Hu, {\it Nonequilibrium quantum fields: Closed-time-path
  effective action, wigner function, and boltzmann equation},  {\em Phys.Rev.D}
  {\bf 37} (1988), no.~10 2878.

\bibitem{Berges:2015kfa}
J.~Berges, {\it {Nonequilibrium Quantum Fields: From Cold Atoms to Cosmology}},
   {\em Lecture Notes of the Les Houches Summer School} (2015)
  [\href{http://arxiv.org/abs/1503.02907}{{\tt arXiv:1503.02907}}].

\bibitem{Berges:2002cz}
J.~Berges and J.~Serreau, {\it {Parametric resonance in quantum field theory}},
   {\em Phys.Rev.Lett.} {\bf 91} (2003) 111601,
  [\href{http://arxiv.org/abs/hep-ph/0208070}{{\tt hep-ph/0208070}}].

\bibitem{traschen1990particle}
J.~H. Traschen and R.~H. Brandenberger, {\it Particle production during
  out-of-equilibrium phase transitions},  {\em Phys.Rev.D} {\bf 42} (1990),
  no.~8 2491, [\href{http://arxiv.org/abs/hep-th/9405187}{{\tt
  hep-th/9405187}}].

\bibitem{kofman1994reheating}
L.~Kofman, A.~Linde, and A.~A. Starobinsky, {\it Reheating after inflation},
  {\em Phys.Rev.Lett.} {\bf 73} (1994), no.~24 3195,
  [\href{http://arxiv.org/abs/hep-th/9405187}{{\tt hep-th/9405187}}].

\bibitem{allahverdi2010reheating}
R.~Allahverdi, R.~Brandenberger, F.-Y. Cyr-Racine, and A.~Mazumdar, {\it
  Reheating in inflationary cosmology: theory and applications},  {\em
  Annu.Rev.Nucl.Part.Sci.} {\bf 60} (2010) 27--51,
  [\href{http://arxiv.org/abs/1001.2600}{{\tt arXiv:1001.2600}}].

\bibitem{amin2015nonperturbative}
M.~A. Amin, M.~P. Hertzberg, D.~I. Kaiser, and J.~Karouby, {\it Nonperturbative
  dynamics of reheating after inflation: a review},  {\em ?Int.J.Mod.Phys.D}
  {\bf 24} (2015), no.~01 1530003, [\href{http://arxiv.org/abs/1410.3808}{{\tt
  arXiv:1410.3808}}].

\bibitem{mrowczynski1995reheating}
S.~Mr{\'o}wczy{\'n}ski and B.~M{\"u}ller, {\it Reheating after supercooling in
  the chiral phase transition},  {\em Phys.Lett.B} {\bf 363} (1995), no.~1-2
  1--4, [\href{http://arxiv.org/abs/nucl-th/9507033}{{\tt nucl-th/9507033}}].

\bibitem{busch2014quantum}
X.~Busch, R.~Parentani, and S.~Robertson, {\it Quantum entanglement due to a
  modulated dynamical casimir effect},  {\em Phys.Rev.A} {\bf 89} (2014), no.~6
  063606, [\href{http://arxiv.org/abs/1404.5754}{{\tt arXiv:1404.5754}}].

\bibitem{fedichev2004cosmological}
P.~O. Fedichev and U.~R. Fischer, {\it Cosmological quasiparticle production in
  harmonically trapped superfluid gases},  {\em Phys.Rev.A} {\bf 69} (2004),
  no.~3 033602, [\href{http://arxiv.org/abs/cond-mat/0303063}{{\tt
  cond-mat/0303063}}].

\bibitem{carusotto2010density}
I.~Carusotto, R.~Balbinot, A.~Fabbri, and A.~Recati, {\it Density correlations
  and analog dynamical casimir emission of bogoliubov phonons in modulated
  atomic bose-einstein condensates},  {\em Eur.Phys.J.D} {\bf 56} (2010), no.~3
  391--404, [\href{http://arxiv.org/abs/0907.2314}{{\tt arXiv:0907.2314}}].

\bibitem{jaskula2012acoustic}
J.-C. Jaskula, G.~B. Partridge, M.~Bonneau, R.~Lopes, J.~Ruaudel, D.~Boiron,
  and C.~I. Westbrook, {\it Acoustic analog to the dynamical casimir effect in
  a bose-einstein condensate},  {\em Phys.Rev.Lett.} {\bf 109} (2012), no.~22
  220401, [\href{http://arxiv.org/abs/1207.1338}{{\tt arXiv:1207.1338}}].

\bibitem{Steinhauer:2015saa}
J.~Steinhauer, {\it {Observation of quantum Hawking radiation and its
  entanglement in an analogue black hole}},  {\em Nature Phys.} {\bf 12} (2016)
  959, [\href{http://arxiv.org/abs/1510.00621}{{\tt arXiv:1510.00621}}].

\bibitem{Campo:2005sy}
D.~Campo and R.~Parentani, {\it {Inflationary spectra and partially decohered
  distributions}},  {\em Phys.Rev.} {\bf D72} (2005) 045015,
  [\href{http://arxiv.org/abs/astro-ph/0505379}{{\tt astro-ph/0505379}}].

\bibitem{Polarski:1995jg}
D.~Polarski and A.~A. Starobinsky, {\it {Semiclassicality and decoherence of
  cosmological perturbations}},  {\em Class.Quant.Grav.} {\bf 13} (1996)
  377--392, [\href{http://arxiv.org/abs/gr-qc/9504030}{{\tt gr-qc/9504030}}].

\bibitem{Kiefer:1999sj}
C.~Kiefer, D.~Polarski, and A.~A. Starobinsky, {\it {Entropy of gravitons
  produced in the early universe}},  {\em Phys. Rev.} {\bf D62} (2000) 043518,
  [\href{http://arxiv.org/abs/gr-qc/9910065}{{\tt gr-qc/9910065}}].

\bibitem{Martin:2015qta}
J.~Martin and V.~Vennin, {\it {Quantum Discord of Cosmic Inflation: Can we Show
  that CMB Anisotropies are of Quantum-Mechanical Origin?}},  {\em Phys. Rev.}
  {\bf D93} (2016), no.~2 023505, [\href{http://arxiv.org/abs/1510.04038}{{\tt
  arXiv:1510.04038}}].

\bibitem{arnol2013mathematical}
V.~I. Arnold, {\em Mathematical methods of classical mechanics}, vol.~60.
\newblock Springer, 2013.

\bibitem{eckmann1985ergodic}
J.-P. Eckmann and D.~Ruelle, {\it Ergodic theory of chaos and strange
  attractors},  {\em Rev.Mod.Phys.} {\bf 57} (1985), no.~3 617.

\bibitem{ginelli2007characterizing}
F.~Ginelli, P.~Poggi, A.~Turchi, H.~Chat{\'e}, R.~Livi, and A.~Politi, {\it
  Characterizing dynamics with covariant lyapunov vectors},  {\em
  Phys.Rev.Lett.} {\bf 99} (2007), no.~13 130601,
  [\href{http://arxiv.org/abs/0706.0510}{{\tt arXiv:0706.0510}}].

\bibitem{pesin1977characteristic}
Y.~B. Pesin, {\it Characteristic lyapunov exponents and smooth ergodic theory},
   {\em Russian Mathematical Surveys} {\bf 32} (1977), no.~4 55--114.

\bibitem{bennetin1980lyapunov}
G.~Bennetin, L.~Galgani, A.~Giorgilli, and J.~Strelcyn, {\it Lyapunov
  characteristic exponents for smooth dynamical systems and for hamiltonian
  systems: A method for computing all of them},  {\em Meccanica} {\bf 15}
  (1980), no.~9.

\bibitem{calabrese2007quantum}
P.~Calabrese and J.~Cardy, {\it Quantum quenches in extended systems},  {\em
  Journal of Statistical Mechanics: Theory and Experiment} {\bf 2007} (2007),
  no.~06 P06008.

\bibitem{dechiara2006entanglement}
G.~DeChiara, S.~Montangero, P.~Calabrese, and R.~Fazio, {\it Entanglement
  entropy dynamics of heisenberg chains},  {\em Journal of Statistical
  Mechanics: Theory and Experiment} {\bf 3} (2006) 03001.

\bibitem{fagotti2008evolution}
M.~Fagotti and P.~Calabrese, {\it Evolution of entanglement entropy following a
  quantum quench: Analytic results for the x y chain in a transverse magnetic
  field},  {\em Physical Review A} {\bf 78} (2008), no.~1 010306.

\bibitem{eisler2008entanglement}
V.~Eisler and I.~Peschel, {\it Entanglement in a periodic quench},  {\em
  Annalen der Physik} {\bf 17} (2008), no.~6 410--423.

\bibitem{lauchli2008spreading}
A.~M. L{\"a}uchli and C.~Kollath, {\it Spreading of correlations and
  entanglement after a quench in the one-dimensional bose--hubbard model},
  {\em Journal of Statistical Mechanics: Theory and Experiment} {\bf 2008}
  (2008), no.~05 P05018.

\bibitem{alba2017entanglement}
V.~Alba and P.~Calabrese, {\it Entanglement and thermodynamics after a quantum
  quench in integrable systems},  {\em Proceedings of the National Academy of
  Sciences} (2017) 201703516.

\bibitem{braunstein2005quantum}
S.~L. Braunstein and P.~Van~Loock, {\it Quantum information with continuous
  variables},  {\em Rev.Mod.Phys.} {\bf 77} (2005), no.~2 513,
  [\href{http://arxiv.org/abs/quant-ph/0410100}{{\tt quant-ph/0410100}}].

\bibitem{ferraro2005gaussian}
A.~Ferraro, S.~Olivares, and M.~G. Paris, {\it Gaussian states in continuous
  variable quantum information},  {\em Bibliopolis, Napoli, 2005.} (2005)
  [\href{http://arxiv.org/abs/quant-ph/0503237}{{\tt quant-ph/0503237}}].

\bibitem{weedbrook2012gaussian}
C.~Weedbrook, S.~Pirandola, R.~Garcia-Patron, N.~J. Cerf, T.~C. Ralph, J.~H.
  Shapiro, and S.~Lloyd, {\it Gaussian quantum information},  {\em
  Rev.Mod.Phys.} {\bf 84} (2012), no.~2 621,
  [\href{http://arxiv.org/abs/1110.3234}{{\tt arXiv:1110.3234}}].

\bibitem{adesso:2014co}
G.~Adesso, S.~Ragy, and A.~R. Lee, {\it Continuous variable quantum
  information: Gaussian states and beyond},  {\em Open Systems \& Information
  Dynamics} {\bf 21} (2014), no.~01n02 1440001,
  [\href{http://arxiv.org/abs/1401.4679}{{\tt arXiv:1401.4679}}].

\bibitem{deutsch_91}
J.~M. Deutsch, {\it Quantum statistical mechanics in a closed system},  {\em
  Phys.Rev.A} {\bf 43} (1991), no.~4 2046--2049.

\bibitem{srednicki_94}
M.~Srednicki, {\it Chaos and quantum thermalization},  {\em Phys.Rev.E} {\bf
  50} (1994), no.~2 888--901.

\bibitem{rigol_dunjko_08}
M.~Rigol, V.~Dunjko, and M.~Olshanii, {\it Thermalization and its mechanism for
  generic isolated quantum systems},  {\em Nature} {\bf 452} (2008) 854.

\bibitem{gutzwiller2013chaos}
M.~C. Gutzwiller, {\em Chaos in classical and quantum mechanics}, vol.~1.
\newblock Springer, 2013.

\bibitem{haake2013quantum}
F.~Haake, {\em Quantum signatures of chaos}, vol.~54.
\newblock Springer, 2013.

\bibitem{reichl2013transition}
L.~Reichl, {\em The transition to chaos: conservative classical systems and
  quantum manifestations}.
\newblock Springer, 2013.

\bibitem{biro1994chaos}
T.~Biro, S.~G. Matinyan, and B.~Muller, {\it Chaos and gauge field theory},
  {\em World Sci. Lect. Notes Phys.} {\bf 56} (1994) 1--288.

\bibitem{martens1989classical}
C.~C. Martens, R.~L. Waterland, and W.~P. Reinhardt, {\it Classical,
  semiclassical, and quantum mechanics of a globally chaotic system:
  Integrability in the adiabatic approximation},  {\em The Journal of Chemical
  Physics} {\bf 90} (1989), no.~4 2328--2337.

\bibitem{Maldacena:2015waa}
J.~Maldacena, S.~H. Shenker, and D.~Stanford, {\it {A bound on chaos}},  {\em
  JHEP} {\bf 08} (2016) 106, [\href{http://arxiv.org/abs/1503.01409}{{\tt
  arXiv:1503.01409}}].

\bibitem{Berenstein:2015yxu}
D.~Berenstein and A.~M. Garcia-Garcia, {\it {Universal quantum constraints on
  the butterfly effect}},  \href{http://arxiv.org/abs/1510.08870}{{\tt
  arXiv:1510.08870}}.

\bibitem{folland:1989}
G.~B. Folland, {\em Harmonic Analysis in Phase Space. (AM-122)}.
\newblock Princeton University Press, f first edition~ed., 3, 1989.

\bibitem{degosson:2006}
M.~A. de~Gosson, {\em Symplectic geometry and quantum mechanics}, vol.~166.
\newblock Springer, 2006.

\bibitem{woit2015quantum}
P.~Woit, {\em Quantum theory, groups and representations: An introduction}.
\newblock Springer, 2017.

\bibitem{v1931eindeutigkeit}
J.~v.~Neumann, {\it Die eindeutigkeit der schr{\"o}dingerschen operatoren},
  {\em Math.Ann.} {\bf 104} (1931), no.~1 570--578.

\bibitem{groenewold1946principles}
H.~J. Groenewold, {\it On the principles of elementary quantum mechanics},
  {\em Physica} {\bf 12} (1946), no.~7 405--460.

\bibitem{walls2007quantum}
D.~F. Walls and G.~J. Milburn, {\em Quantum optics}.
\newblock Springer, 2007.

\end{thebibliography}

\end{document}